\documentclass[final,hidelinks,onefignum,onetabnum]{siamart251216}

\usepackage{amsmath,amssymb}
\usepackage{mathtools}
\usepackage{bm}
\usepackage{graphicx}
\usepackage{booktabs}
\usepackage{enumitem}
\usepackage{algorithm}
\usepackage{algpseudocode}
\usepackage{xcolor}

\newcommand{\RR}{\mathbb{R}}
\newcommand{\NN}{\mathbb{N}}
\newcommand{\EE}{\mathbb{E}}

\newcommand{\Var}{\mathrm{Var}}
\newcommand{\Cov}{\mathrm{Cov}}
\newcommand{\DKL}{\mathcal{D}_{\mathrm{KL}}}
\newcommand{\tr}{\operatorname{Tr}}
\newcommand{\diag}{\mathrm{diag}}
\newcommand{\divop}{\operatorname{div}}

\newcommand{\gauss}{\gamma}

\newcommand{\mA}{\bm{A}}
\newcommand{\mb}{\bm{b}}

\newcommand{\mC}{\bm{C}}
\newcommand{\mB}{\bm{B}}
\newcommand{\hmC}{\widehat{\bm{C}}}
\newcommand{\mV}{\bm{V}}
\newcommand{\hmV}{\widehat{\bm{V}}}

\newcommand{\mP}{\bm{P}}
\newcommand{\hmP}{\widehat{\bm{P}}}
\newcommand{\mS}{\bm{S}}

\newcommand{\mQ}{\bm{Q}}

\newcommand{\mI}{\bm{I}}

\newcommand{\mR}{\bm{R}}

\newcommand{\mG}{\bm{G}}

\newcommand{\mPhi}{\bm{\Phi}}

\newcommand{\mM}{\bm{M}}

\newcommand{\St}{\mathrm{St}}
\newcommand{\Gr}{\mathrm{Gr}}
\newcommand{\vtheta}{\bm{\theta}^{d}}

\newcommand{\whvtheta}{\widehat{\bm{\theta}}^{d}}
\newcommand{\vbeta}{\bm{\theta}^{r}}
\newcommand{\hvbeta}{\widehat{\bm{\theta}}^{r}}

\newcommand{\vn}{\bm{n}}

\newcommand{\vx}{\bm{x}}

\newcommand{\vy}{\bm{y}}

\newcommand{\vz}{\bm{z}}
\newcommand{\vZ}{\bm{Z}}
\newcommand{\vu}{\bm{u}}
\newcommand{\vU}{\bm{U}}
\newcommand{\vv}{\bm{v}}
\newcommand{\vw}{\bm{w}}
\newcommand{\vW}{\bm{W}}

\newcommand{\scZ}{\mathcal{Z}}
\newcommand{\scL}{\mathcal{L}}
\newcommand{\scS}{\mathcal{S}}
\newcommand{\scF}{\mathcal{F}}

\newsiamremark{remark}{Remark}
\newsiamremark{assumption}{Assumption}
\crefname{remark}{Remark}{Remarks}
\crefname{assumption}{Assumption}{Assumptions}
\let\origremarkenv\remark
\renewcommand{\remark}{\crefalias{theorem}{remark}\origremarkenv}
\let\origassumptionenv\assumption
\renewcommand{\assumption}{\crefalias{theorem}{assumption}\origassumptionenv}
\crefname{appendix}{Appendix}{Appendices}
\Crefname{appendix}{Appendix}{Appendices}

\headers{Copula Active Subspaces}{J.~Chen and P.~J.~van Leeuwen}

\title{Copula Active Subspaces I: A Score-Covariance Method for Reduced-Order Non-Gaussian Density Estimation\thanks{This is Part~I of a two-part work; the companion paper \emph{Copula Active Subspaces II: Error Decomposition, A Posteriori Estimation, and Sharpness of the Bounds} develops the finite-sample analysis. \funding{This work was supported by the Consortium for Advanced Data Assimilation Research and Education (CADRE), funded by NOAA under grant 2007893.}}}

\author{%
  Joshua Chen\thanks{Department of Atmospheric Science, Colorado State University, Fort Collins, CO 80523 (\email{Joshua.Chen@colostate.edu}).}
  \and
  Peter Jan van Leeuwen\thanks{Department of Atmospheric Science, Colorado State University, Fort Collins, CO 80523 (\email{Peter.vanLeeuwen@colostate.edu}).}
}

\newcommand{\va}{\bm{a}}
\newcommand{\ve}{\bm{e}}

\hypersetup{hypertexnames=false}
\let\CASsection\section

\begin{document}
\maketitle

\begin{abstract}
In Bayesian inference problems with non-Gaussian observation noise, the expected KL divergence of the posterior is bounded by the KL divergence of the noise density estimate, and gradient-based samplers need that density and its gradient evaluable pointwise without an inner solve. We propose Copula Active Subspaces (\textsc{Cas}) to represent this noise density. A componentwise rank transform isolates the noise law's dependence in its copula, and a rank-$r$ reduction keeps only the directions along which that dependence varies. These directions are the leading eigenvectors of the copula score covariance $\mC := \Cov_{\pi_{\vZ}}(\nabla\log c^{Z})$, and their span is the copula active subspace. The matrix $\mC$ vanishes when the coordinates are independent, so these are directions of dependence. The covariance of the data need not identify them. Applied to the copula, the Gaussian-reference KL divergence bound of certified dimension reduction holds with the explicit constant $\tfrac12$ and is minimized over all rank-$r$ reductions by exactly this eigenspace. We give a diagnostic of the remaining truncation error that is computable from the samples alone. Hermite score matching then yields the reduced log-density and its gradient in closed form, with the truncation orders and the Stage-2 regularization constants chosen on validation samples. The reduction replaces a $d$-dimensional density estimation problem by an $r$-dimensional one. On a $d=20$ noise law and a Bayesian inference problem with that noise, \textsc{Cas} lowers noise KL divergence more than fivefold and posterior KL divergence more than sevenfold against Gaussian-copula, product-of-marginals, and PCA-subspace baselines, and lowers noise KL divergence by factors of about $3.5$ and $2.7$ on two further $d=20$ examples.
\end{abstract}



\section{Introduction}
\label{sec:intro}

In a Bayesian inference problem, a prior on an unknown inference variable $\vx$ is updated by the information that an observation $\vy = \mathcal{G}(\vx) + \vn$ provides about it, where the noise $\vn \in \RR^d$ has law $\pi_{\vn}$. Non-Gaussian noise is common in remote sensing \cite{BormannBauerGeer2011}, data assimilation \cite{JanjicBormannBocquet2018survey, HuVanLeeuwenGeer2024, FowlerVanLeeuwen2013}, and image reconstruction \cite{KaipioSomersaloBook}. A Gaussian likelihood then has the wrong noise law for every choice of its covariance, since the Gaussian family does not contain $\pi_{\vn}$. If $\pi_{\vn}$ in the likelihood is replaced with a candidate $\widehat\pi_{\vn}$ and $\widehat\pi_{X|Y}$ denotes the resulting posterior, the KL divergence chain rule gives
\begin{equation}
  \EE_{\vy\sim \pi_Y}\!\left[\DKL\!\left(\pi_{X|Y}(\cdot\mid\vy) \,\big\|\, \widehat\pi_{X|Y}(\cdot\mid\vy)\right)\right]
  \;\le\; \DKL(\pi_{\vn}\,\|\,\widehat\pi_{\vn}),
  \label{eq:posterior-noise-kl}
\end{equation}
so the average posterior divergence is bounded by that of the noise estimate.\footnote{Apply the KL divergence chain rule to the joint laws $\pi(\vx,\vy) = \pi_X(\vx)\,\pi_{\vn}(\vy - \mathcal G(\vx))$ and $\widehat\pi(\vx,\vy) = \pi_X(\vx)\,\widehat\pi_{\vn}(\vy - \mathcal G(\vx))$ in both orders. In the order $(\vx,\vy)$ it gives $\DKL(\pi\,\|\,\widehat\pi) = \DKL(\pi_{\vn}\,\|\,\widehat\pi_{\vn})$, by translation invariance of the KL divergence; in the order $(\vy,\vx)$ it gives the left side of \eqref{eq:posterior-noise-kl} plus the nonnegative term $\DKL(\pi_Y\,\|\,\widehat\pi_Y)$, and dropping this term gives the inequality. Both noise laws are probability densities and $\widehat\pi_{\vn}$ is positive and bounded, so that $\widehat\pi_Y$ and the conditionals $\widehat\pi_{X|Y}$ are well defined. The estimator of \Cref{sec:estimator} is a positive, bounded probability density.} Gradient-based posterior samplers \cite{RobertsTweedie1996, Neal2011HMC, CotterRobertsStuartWhite2013} require the unnormalized log-density $\log\pi_{\vn}$ and its gradient $\nabla\log\pi_{\vn}$ in closed form, and we seek a tractable estimator of both from $N$ i.i.d.\ samples of the noise in $\RR^d$. Where the noise cannot be observed directly, those samples are approximate samples based on residuals $\vn^{(i)} = \vy^{(i)} - \mathcal G(\vx^\star)$ formed at a nominal variable value $\vx^\star$, such as a maximum-likelihood estimate or an analysis state. We take the sample as given throughout, and \Cref{sec:limitations} discusses the requirements on its construction. The paper therefore estimates $\pi_{\vn}$ itself, and \Cref{sec:exp-bip} checks the transfer from noise divergence to posterior divergence.

In practice the noise model is commonly simplified. Univariate noise modeling treats the coordinates as independent. A Gaussian likelihood keeps second moments only, and a Gaussian copula admits arbitrary marginals with dependence represented by a correlation matrix. Each simplification omits dependence structure, which enlarges the noise divergence on the right of \eqref{eq:posterior-noise-kl}. The $d$ univariate marginals of $\pi_{\vn}$ are classical one-dimensional estimation problems (\Cref{sec:estimator}), and the dependence among the coordinates is the $d$-dimensional problem. A componentwise rank transform separates the two. It maps the samples to coordinates $\vZ$ with standard-Gaussian marginals, so that everything left to estimate is the dependence, described by the copula density $c^{Z} := \pi_{\vZ}/\gauss_d$, where $\gauss_k$ denotes the standard Gaussian density on $\RR^k$. A rank-$r$ dimension reduction then confines that dependence to a subspace. When $c^{Z}$ varies only along $r \ll d$ directions, which is the regime considered here, those directions span an \emph{active subspace} of the copula, and the reduction leaves an $r$-dimensional estimation problem. We identify these directions from samples and quantify the error of the reduction when $c^{Z}$ also varies outside the subspace.

\paragraph{Related work} The rank transform is the standard tool of the semiparametric copula literature \cite{Sklar1959, Nelsen2006book, LiuLaffertyWasserman2009nonparanormal}, which treats marginals and dependence separately. The reduction adapts the active-subspace principle \cite{ConstantineDowWang2014, Constantine2015book, ZahmConstantinePrieurMarzouk2020, BigoniMarzoukPrieurZahm2022}, extraction of low-dimensional structure from a gradient covariance, from functions and log-densities to the copula score. Unlike likelihood-informed subspaces \cite{CuiMartinMarzoukSolonenSpantini2014, ZahmCuiLawSpantiniMarzouk2022, CuiTongZahm2022}, which locate posterior-relevant directions from forward-map gradients, the reduction here uses only residual samples. Score ratio matching \cite{BaptistaBrennanMarzouk2025} learns the score ratio $\nabla\log(\pi/\rho)$ of a density $\pi$ to a reference $\rho$ with a neural network and reduces dimension along the leading eigenvectors of the second moment of the learned ratio. For a standard Gaussian $\rho$ this second moment is the diagnostic matrix of certified dimension reduction \cite{ZahmCuiLawSpantiniMarzouk2022}. Stage~1 of \textsc{Cas} is an instance of score ratio matching. With $\pi = \pi_{\vZ}$ and $\rho = \gauss_d$, the score ratio is the copula score, its second moment is the copula score covariance $\mC$ of \Cref{sec:method}, and the Hermite score-matching objective \eqref{eq:hyv-quadratic} equals, up to an additive constant, the score-ratio-matching objective of \cite[Thm.~3.2]{BaptistaBrennanMarzouk2025}. \textsc{Cas} adds two elements to this setting. The componentwise rank transform comes first, after which $\gauss_d$ matches every marginal and $\mC$ vanishes under independence. The model is a Hermite polynomial class with a Tikhonov penalty, on which the objective is quadratic and is minimized in closed form without training a network. Among non-Gaussian density-reduction methods, transport-based inference and lazy maps \cite{MarzoukMoselhyParnoSpantini2016, MorrisonBaptistaMarzouk2017, BaptistaMarzoukZahm2023representation, BrennanBigoniZahmSpantiniMarzouk2020} are the closest relatives. Ours is a Hermite polynomial counterpart, estimated by regularized linear solves, with a closed-form likelihood.

\paragraph{Contributions}
\begin{enumerate}[leftmargin=*,topsep=2pt,itemsep=2pt]
\item \emph{A two-stage estimator of the copula.} A componentwise rank transform isolates the noise law's dependence in its copula density, which we reduce to a low-dimensional subspace and estimate by Hermite copula score matching, returning the reduced log-density and its gradient in closed form for gradient-based posterior samplers.

\item \emph{Bound-optimal dimension reduction with a closed-form diagnostic.} Applied to the copula, the Gaussian-reference KL divergence bound of certified dimension reduction \cite{ZahmCuiLawSpantiniMarzouk2022} has the explicit constant $\tfrac12$, and its minimizer over rank-$r$ subspaces is the leading eigenspace of the copula score covariance. We give a closed-form, sample-computable diagnostic of the rank-$r$ truncation error.

\item \emph{Empirical validation.} \textsc{Cas} improves noise KL divergence more than fivefold and posterior KL divergence more than sevenfold over Gaussian-copula, product-of-marginals, and PCA-subspace baselines, on a $d=20$ non-Gaussian copula example problem and a Bayesian inference problem with the same noise, and the gains persist across three structurally distinct noise laws.
\end{enumerate}

\Cref{sec:method} constructs the reduction and its bound-optimal subspace, and \Cref{sec:estimator} gives the two-stage estimator. \Cref{sec:theory} analyzes the error at the estimated subspace, including the closed-form diagnostic of the rank-$r$ truncation error. The experiments of \Cref{sec:experiments} use three example problems.

\section{Copula dimension reduction}
\label{sec:method}

This section concerns the true law $\pi_{\vn}$, with no estimation. The copula-score covariance $\mC$ equals zero exactly when the coordinates are independent, and its top-$r$ eigenspace minimizes a KL divergence bound over all rank-$r$ reductions (\Cref{thm:kl-trunc}, \Cref{cor:optimal-subspace}).

Given continuous univariate marginal CDFs $F_1,\ldots,F_d: \RR\to(0,1)$ and i.i.d.\ samples $\{\vn^{(k)}\}_{k=1}^N \subset \RR^d$ from $\pi_{\vn}$, the componentwise rank transform
\begin{equation}
  \scZ: \RR^d \to \RR^d, \qquad
  \scZ(\vx)_i := \Phi^{-1}\bigl(F_i(x_i)\bigr), \quad i=1,\ldots,d,
  \label{eq:rank-transform}
\end{equation}
sends $\pi_{\vn}$ to a measure $\pi_{\vZ}$ with standard-Gaussian marginals and the same copula \cite{Sklar1959, Nelsen2006book}, separating marginal estimation from dependence estimation, and we say the transform \emph{rank-Gaussianizes} the marginals. The ranks of the sample, and therefore the rank-transformed sample, are unchanged under coordinatewise strictly increasing transformations of the residuals \cite{LiuLaffertyWasserman2009nonparanormal}. The rank-$r$ reduction uses the factorization $\gauss_d = \gauss_r\otimes\gauss_{d-r}$ under any splitting of $\RR^d$ into an $r$-dimensional subspace and its orthogonal complement, and the estimator uses the orthonormality of the Hermite polynomial basis of \Cref{sec:method-core} in $L^2(\gauss_d)$.

In $\vZ$-coordinates, writing $c^{Z} := \pi_{\vZ}/\gauss_d$ for the copula density on standard-Gaussian marginals,\footnote{Sklar's theorem writes $\pi_{\vZ}(\vz) = c(\Phi(z_1),\dots,\Phi(z_d))\,\gauss_d(\vz)$ with $c$ a density on $[0,1]^d$, on the uniform arguments $\Phi(z_i)$. We work in the rank-Gaussianized coordinates and write $c^{Z}(\vz) := \pi_{\vZ}(\vz)/\gauss_d(\vz) = c(\Phi(\vz))$, calling $c^{Z}$ the copula, $\scL = \log c^{Z}$ the log-copula, and $\scS = \nabla_{\vz}\scL = \nabla\log\pi_{\vZ} + \vz$ the copula score. The reduced analogue $c^{U}_r(\vu) := \pi_{\vU}(\vu)/\gauss_r(\vu)$ on $\vu = \mV_r^\top\vz$ is likewise a density relative to $\gauss_r$. The marginals of $\pi_{\vU}$ need not be standard Gaussian, so $c^{U}_r$ is a copula density only in this extended sense.} we define the log-copula $\scL: \RR^d\to\RR$, the copula score $\scS: \RR^d\to\RR^d$, and the symmetric positive semidefinite copula-score covariance $\mC\in\RR^{d\times d}$:
\begin{equation}
\begin{gathered}
  \scL(\vz) := \log c^{Z}(\vz),\qquad
  \scS(\vz) := \nabla \scL(\vz),\\[2pt]
  \mC := \Cov_{\pi_{\vZ}}\!\bigl(\scS(\vZ)\bigr) \;=\; \EE_{\pi_{\vZ}}\bigl[\scS(\vZ)\otimes\scS(\vZ)\bigr],
\end{gathered}
  \label{eq:L-T-C-def}
\end{equation}
under the regularity assumption below; the two forms in \eqref{eq:L-T-C-def} coincide because the copula score is centered ($\EE_{\pi_{\vZ}}[\scS] = 0$, as $\nabla\log\pi_{\vZ}$ integrates to zero, \Cref{app:regularity}, and $\pi_{\vZ}$ has zero-mean marginals). The score vanishes a.e.\ under independence ($\scL\equiv 0$), and conversely $\mC = 0$ forces $\scS = 0$ a.e., so $\mC$ is zero exactly under independence. For a unit direction $\vv\in\RR^d$, $\vv^\top\mC\vv = \EE_{\pi_{\vZ}}[(\partial_{\vv}\scL)^2]$ is the mean squared derivative of the log-copula along $\vv$, equal to zero if and only if $\partial_{\vv}\scL = 0$ a.e., that is, if and only if $c^{Z}$ does not vary along $\vv$. $\mC$ is an average gradient outer product, and the same object appears in the diagnostic matrices of gradient- and likelihood-informed dimension reduction \cite{ZahmConstantinePrieurMarzouk2020, BigoniMarzoukPrieurZahm2022, ChenArnaudBaptistaZahm2024} (which use forward-map and log-likelihood gradients) and in the feature-learning mechanism of \cite{RadhakrishnanBeagleholePanditBelkin2024}.

\begin{assumption}[Regularity of the log-copula]
\label{ass:tempered}
The copula density $c$ of $\pi_{\vn}$ is strictly positive on $(0,1)^d$, the log-copula $\scL = \log(c\circ\Phi)$ is continuously differentiable on $\RR^d$, and:
\textup{(T1)} \emph{pointwise sub-Gaussian density bound}, $\pi_{\vZ}(\vz) \le K_0(1+\|\vz\|^{2m})\gauss_d(\vz/\sigma_0)$ for some $K_0\ge 1$, $\sigma_0>1$, $m\ge 0$; \textup{(T2)} \emph{polynomial growth of the score}, $\|\nabla\scL(\vz)\|_2 \le \alpha' + \beta'\|\vz\|^q$ for some $\alpha', \beta' \ge 0$, $q\in\NN$.
\end{assumption}

\noindent (T1)--(T2) give $\scL\in L^2(\gauss_d)$ and $\scS \in L^2(\pi_{\vZ};\RR^d)$, and in particular $\tr(\mC)<\infty$ (\Cref{app:regularity}), as needed for \eqref{eq:L-T-C-def} to be finite. The density bound (T1) excludes heavier-than-Gaussian tails. \Cref{app:regularity} collects the technical consequences of the assumption.

A rank-$r$ dimension reduction approximates $\pi_{\vn}$ within the family whose copula factor depends on only $r$ projected coordinates. It is specified by a matrix $\mV_r$ with $r$ orthonormal columns, $\mV_r^\top\mV_r = \mI_r$, together with a density on the \emph{active subspace} $\mathrm{span}(\mV_r)$; we write $\St(r,\RR^d)$ for the set of such matrices and $\Gr(r,\RR^d)$ for the set of rank-$r$ subspaces of $\RR^d$. The rank $r$ trades approximation accuracy against the statistical error of estimating the reduced density from $N$ samples. At $r=0$ the family is the product of marginals $\prod_i\pi_i$, the univariate noise model common in practice. At $r=d$ it contains the exact joint law, and a full $d$-dimensional density must be estimated. \Cref{sec:reduced-form} finds the best density on a given $\mV_r$, and \Cref{sec:kl-trunc} shows that the top-$r$ eigenspace of $\mC$ minimizes a KL divergence upper bound over $\mV_r$.

\subsection{Reduced model form}
\label{sec:reduced-form}

The \emph{reduced-model space} for a given $\mV_r \in \St(r,\RR^d)$, rank transform $\scZ$, and marginals $\{\pi_i\}$ is
\begin{equation}
  \mathcal M(\mV_r, \scZ, \{\pi_i\}) \;:=\; \Bigl\{\, \pi : \pi(\vx) = c'\bigl(\mV_r^\top \scZ(\vx)\bigr)\textstyle\prod_{i=1}^d \pi_i(x_i)\,\Bigr\},
  \label{eq:reduced-form}
\end{equation}
where $c'$ ranges over probability densities on $\RR^r$ with respect to $\gauss_r$. Here and throughout, $c$ with a superscript denotes a copula-type density relative to the Gaussian reference of matching dimension, the prime marks a generic member of the family \eqref{eq:reduced-form}, and hats mark estimates; the optimal member $c_r^{U}$ is identified below. Set the \emph{active coordinates} $\vU := \mV_r^\top\vZ$ and the \emph{complement coordinates} $\vW := \mV_\perp^\top\vZ$, where $\mV_\perp \in \St(d-r,\RR^d)$ is an orthonormal completion of $\mV_r$, $[\mV_r\ \mV_\perp] \in \mathrm{O}(d)$. Under the Gaussian reference the two blocks are independent, $\gauss_d = \gauss_r\otimes\gauss_{d-r}$ with $\vU\sim\gauss_r$ and $\vW\sim\gauss_{d-r}$; write $\pi_{\vW\mid\vU}$ for the conditional law of $\vW$ given $\vU$ under $\pi_{\vZ}$. In these coordinates a member of $\mathcal M(\mV_r, \scZ, \{\pi_i\})$ has $\vZ$-density $c'(\vu)\,\gauss_d(\vz)$: it places an arbitrary marginal $c'\gauss_r$ on the active coordinates and fixes $\vW\mid\vU \sim \gauss_{d-r}$, independent of $\vU$, on the complement. Splitting the KL divergence from $\pi_{\vZ}$ into the divergence of the $\vU$-marginals plus the $\vU$-averaged divergence of the conditionals (the chain rule), the marginal term equals zero at the choice $c' = c_r^{U} := \pi_{\vU}/\gauss_r$, and the conditional term is the same for every member of the family. The KL divergence projection of $\pi_{\vn}$ onto $\mathcal M(\mV_r, \scZ, \{\pi_i\})$ is therefore
\begin{equation}
  \pi_{\vn}(\vx;\mV_r) := c_r^{U}\!\bigl(\mV_r^\top \scZ(\vx)\bigr) \prod_{i=1}^d \pi_i(x_i),
  \label{eq:pi-star-form}
\end{equation}
with residual exactly the conditional term $\EE_{\pi_{\vU}}[\DKL(\pi_{\vW\mid\vU}\,\|\,\gauss_{d-r})]$, so the reduction is exact precisely when $c^{Z}$ does not vary along the complement. The optimality here is over the reduced factor $c'$ at a fixed $\mV_r$, and \Cref{sec:kl-trunc} chooses $\mV_r$.

\subsection{KL divergence upper bound and optimal subspace}
\label{sec:kl-trunc}

\begin{theorem}[KL divergence upper bound on the reduced model {\cite[Cor.~2.10]{ZahmCuiLawSpantiniMarzouk2022}}]
\label{thm:kl-trunc}
Under Assumption~\ref{ass:tempered}, for every $\mV_r \in \St(r,\RR^d)$,
\begin{equation}
  \DKL\bigl(\pi_{\vn} \,\|\, \pi_{\vn}(\,\cdot\,;\mV_r)\bigr) \;\leq\; \tfrac{1}{2}\,\tr\!\bigl((\mI - \mV_r\mV_r^\top)\,\mC\bigr).
  \label{eq:kl-trunc}
\end{equation}
\end{theorem}

\emph{Proof.} This is \cite[Cor.~2.10]{ZahmCuiLawSpantiniMarzouk2022} with reference measure $\mu = \gauss_d$, target $\nu = \pi_{\vZ}$ and $d\nu/d\mu = f = c^{Z}$. The standard Gaussian satisfies the subspace logarithmic Sobolev inequality of \cite[Thm.~2.9]{ZahmCuiLawSpantiniMarzouk2022} with $\kappa = 1$ and $\Gamma = \mI$, since its conditional laws given $\mV_r^\top\vZ$ are standard Gaussian, whose sharp log-Sobolev constant is $1$ \cite{Gross1975}, \cite[Prop.~5.5.1]{BakryGentilLedoux2014}. The matrix $\int\nabla\log f\,\nabla\log f^\top d\nu$ of that corollary is $\mC$. Its approximation $\nu_r^\ast$, with $d\nu_r^\ast/d\mu = \EE_{\mu}(f \mid \mV_r^\top\vz)$, has the $\vZ$-density $c_r^{U}(\mV_r^\top\vz)\,\gauss_d(\vz)$ of \eqref{eq:pi-star-form}, because $\EE_{\gauss_d}[c^{Z}(\vZ) \mid \mV_r^\top\vZ = \vu] = \pi_{\vU}(\vu)/\gauss_r(\vu) = c_r^{U}(\vu)$. The hypotheses of the corollary, $f > 0$ continuously differentiable with $\nabla\log f$ square-integrable, hold under Assumption~\ref{ass:tempered} (\Cref{app:regularity}). The rank transform is invertible almost surely, so the divergence is the same in $\vx$- and $\vZ$-coordinates.

The reduced model \eqref{eq:pi-star-form} keeps the active marginal $\pi_{\vU}$ exact and replaces only the complement law by $\gauss_{d-r}$, so the bounded divergence is the $\vU$-average $\EE_{\pi_{\vU}}[\DKL(\pi_{\vW\mid\vU}\,\|\,\gauss_{d-r})]$. The prefactor $\tfrac12$ is $\kappa/2$ at the sharp constant $\kappa = 1$ of the Gaussian reference. The same bound with a standard Gaussian reference underlies the parameter reduction reviewed in \cite[Prop.~2.2]{BaptistaBrennanMarzouk2025}.

\begin{corollary}[Bound-optimal subspace]
\label{cor:optimal-subspace}
Let $(\lambda_i(\mC), \vv_i(\mC))_{i=1}^d$ be the eigenpairs of $\mC$ ordered $\lambda_1(\mC) \ge \cdots \ge \lambda_d(\mC)$, and $\mV_r^{\mC} := [\vv_1(\mC)\,\cdots\,\vv_r(\mC)]$. The bound of \Cref{thm:kl-trunc} depends only on $\mathrm{span}(\mV_r) \in \Gr(r,\RR^d)$; its minimum is attained on $\mathrm{span}(\mV_r^{\mC})$ with value
\[
  \tfrac{1}{2}\tr\!\bigl((\mI-\mV_r^{\mC}\mV_r^{\mC\top})\mC\bigr) \;=\; \tfrac{1}{2}\sum_{i > r}\lambda_i(\mC),
\]
giving $\DKL\bigl(\pi_{\vn}\,\|\,\pi_{\vn}(\,\cdot\,;\mV_r^{\mC})\bigr) \leq \tfrac{1}{2}\sum_{i>r}\lambda_i(\mC)$.
\end{corollary}

\emph{Proof.} As in \cite[\S 2]{ZahmCuiLawSpantiniMarzouk2022}, the bound of \Cref{thm:kl-trunc} depends on $\mV_r$ only through $\mathrm{span}(\mV_r)$, and minimizing $\tr((\mI-\mP)\mC)$ over rank-$r$ orthogonal projectors is the Eckart--Young low-rank approximation problem in the Frobenius norm \cite{EckartYoung1936}: writing $\mC = \mC^{1/2}\mC^{1/2}$, the quantity $\tr((\mI-\mP)\mC) = \|(\mI-\mP)\mC^{1/2}\|_F^2$ is the squared Frobenius distance from $\mC^{1/2}$ to its projection onto a rank-$r$ subspace, minimized by the top-$r$ eigenspace of $\mC$ with minimum $\sum_{i>r}\lambda_i(\mC)$.

We call $\mathrm{span}(\mV_r^{\mC})$ the \emph{copula active subspace}. The reduction is lossless when the tail eigenvalues $\lambda_{r+1}(\mC),\ldots,\lambda_d(\mC)$ vanish: $\pi_{\vn}(\vx) = \pi_{\vn}(\vx;\mV_r^{\mC})$ a.e. The minimizer of \eqref{eq:kl-trunc} gives the subspace in closed form, and the minimum value bounds the KL divergence of the reduction. Both results use the exact score covariance $\mC$. \Cref{sec:theory} applies \Cref{thm:kl-trunc} to an estimated basis $\hmV_r^{\mC}$.

\section{The \textsc{Cas} estimator}
\label{sec:estimator}

The estimator first estimates the marginals. The empirical distribution functions $\{\widehat F_i\}$ define the empirical rank transform $\widehat\scZ$ (\Cref{alg:cas}, step~1), and kernel density estimates $\{\widehat\pi_i\}$ give the marginal densities. Stage~1 (\Cref{sec:method-core}) computes a basis $\hmV_r^{\mC}$ for the estimated Stage-1 subspace from the rank-Gaussianized samples by Hermite copula score matching. Stage~2 (\Cref{sec:reduced-model}) estimates the reduced log-density $\scL_r := \log c_r^{U} = \log(\pi_{\vU}/\gauss_r)$ on $\RR^r$ from the projected coordinates $\vU = \hmV_r^{\mC\top}\vZ$, with the regularizer of \Cref{sec:ridge}. The full estimator is
\begin{equation}
\begin{gathered}
  \log\widehat\pi_{\vn}(\vx;\hmV_r^{\mC}) = \widehat{\scL}_r(\hmV_r^{\mC\top}\widehat\scZ(\vx)) + \sum_i\log\widehat\pi_i(x_i) - \log\widehat Z_{\vn},\\[2pt]
  \widehat Z_{\vn} := \int_{\RR^d} e^{\widehat{\scL}_r(\hmV_r^{\mC\top}\widehat\scZ(\vx))}\prod_i\widehat\pi_i(x_i)\,d\vx.
\end{gathered}
  \label{eq:cas-estimator}
\end{equation}
The normalizer $\widehat Z_{\vn}$ is a constant, so it shifts the log-density and does not change its gradient. It is finite because the empirical transform $\widehat\scZ$ takes values in a bounded box, on which the polynomial $\widehat\scL_r$ is bounded. The transform $\widehat\scZ$ and the kernel estimates $\widehat\pi_i$ are estimated separately, so the normalization of the family \eqref{eq:reduced-form} does not carry over to their composition, and \eqref{eq:cas-estimator} normalizes the composition directly. When $\widehat Z_r := \int_{\RR^r} e^{\widehat\scL_r}\,d\gauss_r$ is finite, normalizing the reduced factor $e^{\widehat\scL_r}/\widehat Z_r$ first gives the same density, since that constant cancels. The estimate is thus a positive probability density, bounded because each kernel density estimate $\widehat\pi_i$ is bounded, as \eqref{eq:posterior-noise-kl} requires.

\paragraph{Marginals} Estimating the $d$ marginals is the classical part of the composition, with rates available for kernel and log-spline density estimation \cite{Hall1987, Stone1990}. Each $\widehat\pi_i$ is a Gaussian-kernel density estimator with a robust Silverman bandwidth $h\asymp N^{-1/5}$ \cite[\S 3.4]{Silverman1986}, and $\widehat F_i$ is the empirical distribution function of \Cref{alg:cas}, step~1 (\Cref{app:protocol}). We do not analyze the contribution of these estimates to the KL divergence (\Cref{sec:limitations}).

\subsection{Stage 1: subspace identification}
\label{sec:method-core}

Stage 1 returns $\hmV_r^{\mC}$, the matrix of top-$r$ eigenvectors of the empirical score second-moment matrix $\hmC$. Its column space is the estimated Stage-1 subspace analyzed in \Cref{sec:theory}, and \Cref{fig:pedagogical} illustrates it in two dimensions. To compute $\hmC$, Stage 1 expands the log-copula in a Hermite polynomial basis and recovers the coefficients by score matching from a single linear system. The resulting analytical scores give $\hmC$ in closed form.

\begin{figure}[t]
\centering
\includegraphics[width=\textwidth]{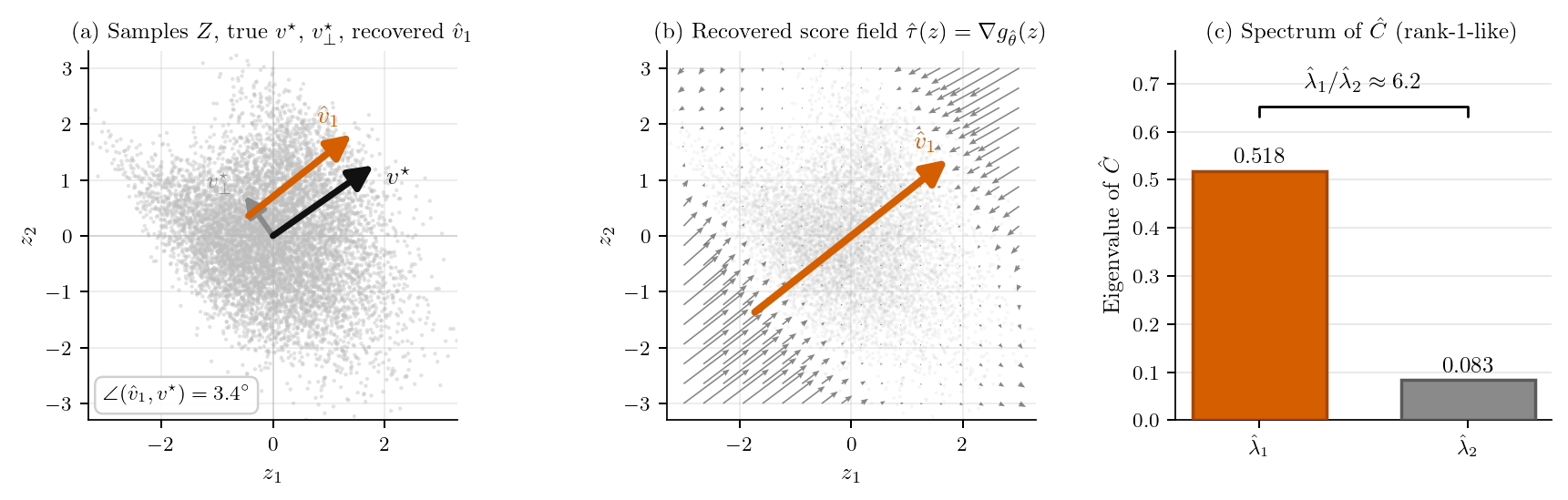}
\caption{\textsc{Cas} is illustrated in two dimensions on a non-Gaussian law whose dependence is concentrated along one non-axis-aligned direction $\vv^\star$, with quadratic conditional-mean dependence. Panel~(a) shows the samples and the top eigenvector $\widehat{\vv}_1$ of $\hmC$, panel~(b) the estimated score field, and panel~(c) the spectrum of $\hmC$.}
\label{fig:pedagogical}
\end{figure}

\paragraph{Hermite polynomial expansion of the log-copula}
Let $h_n$ be the normalized univariate Hermite polynomials ($\EE_{\gauss_1}[h_m h_n] = \delta_{mn}$), $H_\alpha(\vz) := \prod_i h_{\alpha_i}(z_i)$ the tensor basis (orthonormal in $L^2(\gauss_d)$ \cite{Szego1939, Janson1997gaussian}), and
\begin{equation}
  \Lambda_{K,q} \,:=\, \{ \alpha \in \NN_0^d : 1 \leq |\alpha|_1 \leq K,\ |\alpha|_0 \leq q,\ \alpha \neq e_i \text{ for all } i\}
  \label{eq:dict}
\end{equation}
the multi-index set of multivariate Hermite polynomials, whose elements $\alpha$ index the log-copula coefficients $\theta_\alpha$, with maximum total degree $K$ and maximum interaction order $q \leq K$. For Stage 1 we use $\Lambda := \Lambda_{K_1, q_1}$.\footnote{Linear modes $e_i$ are omitted. A pure linear tilt only re-parametrizes the reference, $e^{\delta z_i}\gauss_d(\vz) = e^{\delta^2/2}\gauss_d(\vz - \delta e_i)$ exactly, and within a nonlinear expansion the omission restricts the dictionary.} The truncation orders $(K_1, q_1)$ and the analogous Stage-2 parameters $(K_2, q_2)$ are chosen for each estimate by cross-validation \cite[\S 7.10]{HastieTibshiraniFriedman2009}, minimizing the score-matching objective on validation samples. The same criterion chooses the regularizer constants (\Cref{sec:ridge}).\footnote{Supplement~\ref{sm:adaptivity} describes an optional greedy choice of the Stage-1 degree and interaction order, with reuse of the score-matching quadratic system.}

Writing $\scL(\vz;\vtheta) := \sum_{\alpha \in \Lambda} \theta^{d}_\alpha H_\alpha(\vz)$ for the parametric log-copula factor and $\scS(\vz;\vtheta) := \nabla\scL(\vz;\vtheta)$ for its gradient, the associated unnormalized Lebesgue density is $\pi_{\vZ}(\vz;\vtheta) := e^{\scL(\vz;\vtheta)}\gauss_d(\vz)$, with \emph{$\vz$-density score}
\begin{equation}
  \nabla\log\pi_{\vZ}(\vz;\vtheta) \;=\; \scS(\vz;\vtheta) - \vz,
  \label{eq:hermite-expansion}
\end{equation}
where the $-\vz$ is the Gaussian-reference score $\nabla\log\gauss_d$.

\paragraph{Hermite copula score matching}
Score matching \cite{Hyvarinen2005} determines $\vtheta$ by minimizing the Fisher divergence between the model $\vz$-density score $\nabla\log\pi_{\vZ}(\cdot;\vtheta)$ and the true score $\nabla\log\pi_{\vZ}$. Because $\scL(\cdot;\vtheta)$ is linear in $\vtheta$, the Hyvärinen identity reduces this to a single linear system $\widehat\mA\,\whvtheta = \widehat\mb$ whose entries are sample averages of Hermite-polynomial derivatives over $\{\vZ^{(n)}\}_{n=1}^N$, assembled by evaluating the expansion and its first two derivatives at the samples with no model simulation and no inner optimization loop. We regularize the inversion with a Tikhonov term \cite{EnglHankeNeubauer1996} and solve
\begin{equation}
  \whvtheta \;=\; (\widehat\mA + \mR_{H^2})^{-1}\widehat\mb, \qquad \mR_{H^2} = \kappa\,\|\widehat\mA\|_{\rm op}\,\diag(|\alpha|_1^2)_{\alpha\in\Lambda}, \quad \kappa = 10^{-4}.
  \label{eq:tikhonov}
\end{equation}
\Cref{app:stage2-foundations} derives the closed forms of $\widehat\mA,\widehat\mb$ from the score-matching identity. We call the resulting procedure Hermite copula score matching (HCSM, \Cref{alg:hsm}). It estimates the Hermite expansion of a copula log-density, and hence its score, from $\vz$-coordinate samples. Stage 2 (\Cref{sec:reduced-model}) applies HCSM to the projected samples.

\paragraph{\textsc{Cas} subspace estimator}
From the coefficients $\whvtheta$, the analytical copula score is $\widehat\scS(\vz) := \nabla\scL(\vz;\whvtheta)$; the score matrix is $\mS \in \RR^{N \times d}$, $\mS_{n,:} = \widehat\scS(\vZ^{(n)})^\top$; and the empirical score second-moment matrix is
\begin{equation}
  \hmC \;:=\; \tfrac{1}{N}\,\mS^\top\mS \;=\; \tfrac{1}{N}\sum_{n=1}^N \widehat\scS(\vZ^{(n)})\,\widehat\scS(\vZ^{(n)})^\top,
  \label{eq:Chat-def}
\end{equation}
with eigenpairs $(\widehat\lambda_i, \widehat\vv_i)_{i=1}^d$ ordered $\widehat\lambda_1 \geq \cdots \geq \widehat\lambda_d$. The top-$r$ eigenvectors form $\hmV_r^{\mC} := [\widehat\vv_1 \;\cdots\; \widehat\vv_r] \in \St(r,\RR^d)$, whose column space is the estimated Stage-1 subspace. \Cref{alg:cas} summarizes the procedure.

\begin{algorithm}[h]
\caption{\textsc{Cas} Stage 1 (subspace identification).}
\label{alg:cas}
\small
\begin{algorithmic}[1]
\Require samples $\{\vn^{(k)}\}_{k=1}^N \subset \RR^d$; maximum total degree $K_1$; maximum interaction order $q_1 \leq K_1$; rank $r$; Tikhonov scale $\kappa \geq 0$; oversampling $p$.
\Ensure eigenbasis $\hmV_r^{\mC} \in \St(r,\RR^d)$ spanning the rank-$r$ subspace, top-$r$ eigenvalues $\widehat\lambda_{1:r}$, stored marginal CDFs $\{\widehat F_i\}_{i=1}^d$, coefficient estimate $\whvtheta$, and tail sum $\widehat E_r$.

\State \textbf{Empirical ranks} $\{\widehat F_i\}_{i=1}^d$ from marginals of $\{\vn^{(k)}\}$; rank-Gaussianize $Z_i^{(k)} = \Phi^{-1}(\widehat F_i(\vn_i^{(k)}))$ per \eqref{eq:rank-transform}, using $\widehat F_i(\vn_i^{(k)}) = R_i^{(k)}/(N+1)$ to avoid the tail blow-up of $\Phi^{-1}$. $\triangleright\ O(Nd\log N)$.

\State \textbf{Enumerate} multi-index set $\Lambda \subset \NN_0^d$ by maximum-degree and interaction-order limits $(K_1, q_1)$ per \eqref{eq:dict}.

\State \textbf{Run HCSM} (\Cref{alg:hsm}) on $\{\vZ^{(k)}\}_{k=1}^N$ with multi-index set $\Lambda$ and regularizer $\mR_{H^2}$, $\to \whvtheta \in \RR^{|\Lambda|}$.

\State \textbf{Assemble score matrix} $\mS \in \RR^{N \times d}$: $\mS_{ki} = \scS_{\whvtheta}(\vZ^{(k)})_i$.

\State \textbf{Extract eigenbasis} $(\hmV_r^{\mC}, \widehat\lambda_{1:r}) \gets $ top-$r$ eigenpairs of $\hmC := \mS^\top\mS / N \in \RR^{d\times d}$ via randomized range-finder \cite{HalkoMartinssonTropp2011} applied to $\mS^\top$ with oversampling $p$; $\hmC$ is never assembled. $\triangleright\ O(Nd(r+p))$.

\State \textbf{Tail sum} $\widehat E_r = \tr(\hmC) - \sum_{i\le r}\widehat\lambda_i = \|\mS\|_F^2/N - \sum_{i\le r}\widehat\lambda_i$ as a finite-sample diagnostic, computable in $O(Nd)$ from $\mS$ without forming further eigenvectors. The optional dictionary search of Supplement~\ref{sm:adaptivity} reuses quadratic blocks and a Cholesky factorization under a penalty fixed during the search.
\end{algorithmic}
\end{algorithm}

\subsection{Stage 2: score matching for the reduced density}
\label{sec:reduced-model}

The basis $\hmV_r^{\mC}$ of Stage~1 does not depend on how $c_r^{U}$ is estimated afterwards. On the same subspace, the Hermite estimator below can therefore be replaced by another estimator of $c_r^{U}$ or of its score, such as neural-network score ratio matching \cite{BaptistaBrennanMarzouk2025} or a normalizing flow \cite{PapamakariosEtAl2021}, when an adaptive polynomial approximation in $r$ dimensions is unavailable or inaccurate. Stage~1 still identifies the subspace by a linear score-matching solve and an eigendecomposition. We use Hermite score matching in both stages, which gives $\widehat\scL_r$ and its gradient in closed form.

Stage 2 estimates the reduced factor $c_r^{U}$ on the projected coordinates $\vU = \hmV_r^{\mC\top} \vZ$ by Hermite copula score matching on $\RR^r$. The reduced score $\nabla\scL_r(\vu) = \nabla\log c_r^{U}(\vu)$ relates to the density score by $\partial_j\scL_r(\vu) = \partial_j\log\pi_{\vU}(\vu) + u_j$ since $\pi_{\vU} = c_r^{U}\,\gauss_r$. We use the multi-index set $\Lambda_r := \Lambda_{K_2, q_2}\subset \NN_0^r$ from \eqref{eq:dict}, again excluding constant and linear modes (the reference-mean reparametrization of \Cref{sec:method-core}),\footnote{The second-order modes $H_{2e_j} \in \Lambda_r$ are retained, and represent variance mismatch between $\pi_{\vU}$ and $\gauss_r$.} and parametrize $\log c_r^{U}$ as $\scL_r(\vu;\vbeta) := \sum_{\alpha \in \Lambda_r}\theta^{r}_\alpha H_\alpha(\vu)$, with $\vbeta = (\theta^r_\alpha)_{\alpha\in\Lambda_r}$.

Stage 2 minimizes the Fisher divergence
\begin{equation}
  J(\vbeta)
  \;:=\;
  \tfrac{1}{2}\sum_{j=1}^r\,\EE_{\pi_{\vU}}\!\left[\bigl(\partial_j\scL_r(\vU;\vbeta)-\partial_j\scL_r(\vU)\bigr)^2\right]
  \label{eq:JW-pop}
\end{equation}
between the model and the true reduced scores. By the Hyvärinen identity \cite[Thm.~1]{Hyvarinen2005}, this is again a quadratic in $\vbeta$ \cite[\S4]{Hyvarinen2007}, so Stage 2 is HCSM (\Cref{alg:hsm}) applied to the projected samples $\vu^{(n)} := \hmV_r^{\mC\top}\vZ^{(n)}$ with $d \to r$ and $\Lambda \to \Lambda_r$, assembling the Gram matrix and right-hand side $\widehat\mA_r, \widehat\mb_r$ of \eqref{eq:A-b-formulas}. \Cref{lem:hyv-sm} states the identity and verifies its integration-by-parts hypotheses under Assumption~\ref{ass:tempered}. \Cref{sec:ridge} adds a regularizer and two constraints to this quadratic and gives the resulting estimator.

\subsection{Stage-2 regularization and constraints}
\label{sec:ridge}

The unregularized normal equations $\widehat\mA_r\hvbeta = \widehat\mb_r$ are ill-conditioned when the smallest eigenvalues of $\widehat\mA_r$ fall below their sampling error, which is of order $N^{-1/2}$. The penalty restores well-posedness, in the sense of the penalized score-matching estimator of \cite{Sriperumbudur2017}, and its two terms control that estimator's two error sources, the variance from the small eigenvalues and the truncation residual propagated through $\widehat\mA_r^{-1}$.

\subsubsection{The regularizer}
\label{sec:ridge-reg}

\paragraph{Form}
The regularizer is the diagonal, data-driven matrix
\begin{equation}
  \mR(\widehat\mA_r) \;:=\; \underbrace{C_{\rm samp}\,\diag_k\!\bigl(\widehat\sigma_k^2\,\mI_{p_k}\bigr)}_{\mR_{\rm samp}} \;+\; \underbrace{C_{\rm curv}\,\|\widehat\mA_r\|_{\rm op}\,\diag(|\alpha|_1^2)}_{\mR_{\rm curv}},
  \label{eq:R-combined}
\end{equation}
summed over the blocks $\Lambda^{(k)} := \{\alpha\in\Lambda_r:|\alpha|_1 = k\}$ of multi-indices with total degree $k$, with $p_k := |\Lambda^{(k)}|$ and per-degree energy scale $\widehat\sigma_k^2 := \tfrac{1}{N}\tr(\widehat\mA_r^{(k,k)})$, equal to $k\,p_k/N$ under the Gaussian reference. The $|\alpha|_1^2$ weighting is the squared spectral weight of the Ornstein--Uhlenbeck generator on degree-$k$ Hermite polynomials, a Gaussian Sobolev (curvature) penalty \cite{EnglHankeNeubauer1996}.

The two terms damp the directions of the smallest eigenvalues in complementary ways. $\mR_{\rm samp}$ scales with $\widehat\sigma_k^2$, of order $1/N$, so it decreases as $N$ grows. $\mR_{\rm curv}$ scales with $\|\widehat\mA_r\|_{\rm op}$ and the squared degree $|\alpha|_1^2$. It is asymptotically $N$-independent and largest on the high-degree coefficients. Together they keep the solve well-posed at every sample size studied, $N \in [10^2, 10^5]$.

\paragraph{Conditioning} On each block $\Lambda^{(k)}$, $\mR_{\rm curv}$ adds a positive multiple of the identity, which bounds the condition number of that block of $\widehat\mA_r + \mR$ by $1 + 1/(C_{\rm curv}k^2)$ and bounds the whole matrix away from singularity, $\widehat\mA_r + \mR \succeq \min_k(C_{\rm samp}\widehat\sigma_k^2 + C_{\rm curv}\|\widehat\mA_r\|_{\rm op}k^2)\,\mI$.

\paragraph{Choice of the two constants} The two constants $(C_{\rm samp}, C_{\rm curv})$ are chosen for each problem by cross-validation, minimizing the score-matching objective on validation samples, which in fixed coordinates estimates the reduced Fisher divergence \eqref{eq:JW-pop} up to a constant and requires no reference density and no noise model. The choice involves two global scalars, one per term, and both terms are active at every $N \in [10^2, 10^5]$. Supplement~\ref{sm:s2reg-empirical} specifies the procedure.

\subsubsection{The two constraints}
\label{sec:ridge-constraints}

\paragraph{Mean-zero score} Alongside the regularizer, we impose the empirical score-mean-zero identity
\begin{equation}
  \widehat\mB\hvbeta = 0,
  \qquad
  \widehat B_{j\alpha} := \tfrac{1}{N}\sum_n \partial_j H_\alpha(\vu^{(n)}),
  \label{eq:centering-mat}
\end{equation}
the empirical form of the identity $\EE_{\pi_{\vU}}[\nabla\log c_r^{U}(\vU)] = 0$ (\Cref{app:regularity}), as a linear constraint. \Cref{prop:centering}, stated after the closed form in \Cref{sec:ridge-solve}, records its effect.

\paragraph{Normalizability} The second constraint makes the reduced factor a probability density with respect to $\gauss_r$. $\scL_r(\cdot;\vbeta)$ is a polynomial, in general unbounded above on $\RR^r$, and without a further constraint $e^{\scL_r}$ need not be integrable against $\gauss_r$. We therefore impose two linear inequality families on $\vbeta$, a level bound $\scL_r(\vu;\vbeta) \le \tau$ on a ball containing the data, for a level $\tau$ computed from the data in \Cref{sec:ridge-solve}, and a strict negativity margin on the leading form of $\scL_r$ over the unit sphere (\Cref{app:solver}). Together they bound $\scL_r$ above on all of $\RR^r$, so $e^{\scL_r}$ is integrable against $\gauss_r$ and $e^{\widehat\scL_r}/\widehat Z_r$ is a probability density with respect to $\gauss_r$.

\subsubsection{The estimator and its solve}
\label{sec:ridge-solve}

The estimator solves
\begin{equation}
\begin{gathered}
  \hvbeta = \arg\min_{\vbeta\,\in\,\mathcal C}\;\tfrac{1}{2}{\vbeta}^\top\bigl(\widehat\mA_r + \mR(\widehat\mA_r)\bigr)\vbeta - {\vbeta}^\top\widehat\mb_r,\\[2pt]
  \mathcal C := \bigl\{\vbeta : \widehat\mB\vbeta = 0,\ \scL_r(\vu;\vbeta) \le \tau \text{ for } \|\vu\| \le r_{\max},\ F[\vbeta](\vv) \le -\varepsilon \text{ for } \|\vv\| = 1\bigr\},
\end{gathered}
  \label{eq:stage2-est}
\end{equation}
with $F[\vbeta]$ the leading form of $\scL_r(\cdot;\vbeta)$ and the radius $r_{\max}$ and margin $\varepsilon$ given in \Cref{app:solver,app:protocol}. To compute the level $\tau$, the program is first solved with the centering constraint alone and the estimated $\scL_r$ is evaluated at the $N$ samples. The level $\tau$ is the $99.9$th percentile of those $N$ values. With the centering constraint alone the solution has the closed form
\begin{equation}
  \hvbeta = \widehat\mM^{-1}\widehat\mb_r - \widehat\mM^{-1}\widehat\mB^\top \bigl(\widehat\mB\widehat\mM^{-1}\widehat\mB^\top\bigr)^{-1}\widehat\mB\widehat\mM^{-1}\widehat\mb_r,
  \quad
  \widehat\mM := \widehat\mA_r + \mR(\widehat\mA_r).
  \label{eq:stage2-closedform}
\end{equation}
The solve requires one Cholesky factorization \cite{GolubVanLoan2013} of the $|\Lambda_r| \times |\Lambda_r|$ matrix $\widehat\mM$ and one $r \times r$ linear solve. The sizes of both matrices depend on the reduced dimension $r$ and not on $d$. With $\hvbeta$ in hand, the gradient of \eqref{eq:cas-estimator} follows by the chain rule through $\widehat\scZ$ wherever the transform is differentiable (\Cref{app:protocol}).

Each inequality in $\mathcal C$ is linear in $\vbeta$ because $\scL_r(\vu;\cdot)$ and $F[\cdot](\vv)$ are linear in the coefficients, so $\mathcal C$ is an intersection of half-spaces, hence convex, indexed by the points of the ball and the sphere. The solver imposes the two families at a finite set of points, enlarged iteratively by the points at which a search over the ball and the sphere finds the largest violation, until the search finds none beyond its tolerance or a round limit is reached (\Cref{app:solver}). Each re-solve is a convex quadratic program with the added rows as inequality constraints, and when the solution of \eqref{eq:stage2-closedform} already satisfies both families at every searched point, no points are added and the estimator equals \eqref{eq:stage2-closedform}.

\Cref{prop:centering} describes the effect of the centering constraint on the bias and the variance of the solve.

\begin{proposition}[Properties of the centering constraint]
\label{prop:centering}
Fix the true rank transform, a deterministic basis, the multi-index set $\Lambda_r$, and the regularizer $\mR$, so that the projected sample is i.i.d. Let $\hvbeta$ be the centered Stage-2 estimator \eqref{eq:stage2-closedform} and $\hvbeta_{\rm unc} := \widehat\mM^{-1}\widehat\mb_r$ the unconstrained solve. Let $\mA^\star, \mB^\star, \bm b^\star$ be the $N\to\infty$ limits of $\widehat\mA_r, \widehat\mB, \widehat\mb_r$, with $\mB^\star$ of full row rank, set $\bm M := \mA^\star + \mR$, let $\Pi^\star$ be the $\bm M$-orthogonal projector onto $\mathrm{range}(\bm M^{-1}\mB^{\star\top})$, let $\bm\theta^{r,\star}$ be the truncated minimizer $\mA^\star\bm\theta^{r,\star} = \bm b^\star$, and let $\bar{\bm\theta} := (\mI-\Pi^\star)\bm M^{-1}\bm b^\star$ be the centered solve of the limiting system. Then:

\emph{(i) The constraint acts by projection.} $\hvbeta = (\mI - \Pi_{\widehat\mB})\hvbeta_{\rm unc}$, where $\Pi_{\widehat\mB} := \widehat\mM^{-1}\widehat\mB^\top(\widehat\mB\widehat\mM^{-1}\widehat\mB^\top)^{-1}\widehat\mB$ is the $\widehat\mM$-orthogonal projector onto $\mathrm{range}(\widehat\mM^{-1}\widehat\mB^\top)$, defined whenever $\widehat\mB$ has full row rank.

\emph{(ii) It replaces one bias by another.} The centered solve of the limiting system satisfies the exact identity
\begin{equation}
  \bar{\bm\theta} - \bm\theta^{r,\star}
  = \underbrace{-(\mI - \Pi^\star)\,\bm M^{-1}\mR\,\bm\theta^{r,\star}}_{\text{regularization bias, contracted}}
  \;\underbrace{-\; \Pi^\star\bm\theta^{r,\star}}_{\text{truncation bias, added}},
  \label{eq:centering-bias}
\end{equation}
whose first term is the unconstrained regularization bias $-\bm M^{-1}\mR\bm\theta^{r,\star}$ after an $\bm M$-orthogonal contraction that cannot enlarge it ($\|\mI-\Pi^\star\|_{\bm M} \le 1$), and whose second term is $\Pi^\star\bm\theta^{r,\star} = -\bm M^{-1}\mB^{\star\top}(\mB^\star\bm M^{-1}\mB^{\star\top})^{-1}\EE_{\pi_{\vU}}[\bm\rho]$, fixed by $\mB^\star\bm\theta^{r,\star} = -\EE_{\pi_{\vU}}[\bm\rho]$ for the Fisher truncation residual $\bm\rho := \nabla\scL_r - \Pi_{\Lambda_r}\nabla\scL_r$ ($\Pi_{\Lambda_r}$ the $L^2(\pi_{\vU})$ projection onto $\mathrm{span}\{\nabla H_\alpha : \alpha\in\Lambda_r\}$), and equal to zero exactly when $\EE_{\pi_{\vU}}[\bm\rho] = 0$.

\emph{(iii) The limiting projection is non-expansive.} The $\bm M$-orthogonal projection does not increase the $\bm M$-weighted variance: for every random $\bm\xi$ with finite second moment,
\begin{equation}
  \tr\!\bigl(\mA^\star\Cov((\mI-\Pi^\star)\bm\xi)\bigr) \;\le\; \tr\!\bigl(\bm M\,\Cov(\bm\xi)\bigr).
  \label{eq:centering-var}
\end{equation}

\emph{(iv) The estimator reaches $\bar{\bm\theta}$ at the Monte Carlo rate.} Under \Cref{ass:tempered} the sample averages $(\widehat\mB, \widehat\mM, \widehat\mb_r)$ obey a central limit theorem, so $\|\hvbeta - \bar{\bm\theta}\| = O_p(N^{-1/2})$.
\end{proposition}

\emph{Proof.} See \Cref{app:centering}.

The regularization bias of \eqref{eq:centering-bias} carries the factor $\mR$. Its sampling-noise part decreases with $N$ through the term $\mR_{\rm samp} \asymp 1/N$ of \eqref{eq:R-combined}, and its curvature part need not decrease. The truncation bias is $N$-independent and equals zero when the reduced score is representable in $\Lambda_r$. The sign of the net effect therefore depends on the problem and on $N$, and Supplement~\ref{sec:sm2-centering} measures it on Examples~1--3.

\paragraph{Independence of the basis} Stage 1 identifies a subspace, and any orthonormal basis of that subspace could be handed to Stage 2, so the Stage-2 estimate should not depend on the choice of basis. When the specified maximum interaction order excludes no multi-index ($q_2 \ge \min(K_2, r)$, in which case the condition $|\alpha|_0 \le q_2$ removes nothing from $\Lambda_r$), the estimator with exact optimization and integration is invariant. A basis change $\hmV_r^{\mC} \mapsto \hmV_r^{\mC}\mQ$ with $\mQ \in \mathrm{O}(r)$ rotates the coordinates $\vu$, and a rotation maps the span of the Hermite polynomials of each total degree $k$ onto itself. The span of the model polynomials, the per-degree scales $\widehat\sigma_k^2$ and $\|\widehat\mA_r\|_{\rm op}$ are unchanged. The regularizer $\mR$ multiplies all coefficients of the same total degree by one scalar and therefore commutes with the rotation, and the centering constraint set $\{\vbeta : \widehat\mB\vbeta = 0\}$ maps onto itself. The estimated density $\widehat\pi_{\vn}$ is therefore unchanged. When the maximum interaction order excludes multi-indices ($q_2 < K_2$ and $q_2 < r$), the estimate depends on the basis through the excluded multi-indices. The numerical search and quadrature of \Cref{app:solver,app:protocol} are not rotated with the basis, so in computation the invariance holds up to their tolerances.

\paragraph{Performance of the regularizer}
The comparisons below hold the constants at $(C_{\rm samp}, C_{\rm curv}) = (3, 10^{-5})$, the center of the selection grid of \Cref{app:protocol}, at every $N$. At this pair the estimator stays within $0.18$ nats in held-out log-score of the best alternative regularizer on all three examples of \Cref{sec:experiments} at every $N \in [10^2, 10^5]$, while the unregularized solve raises that score by $2.8$--$3.6$ nats at $N = 100$. On the inference problem of \Cref{sec:exp-bip}, the estimator at this pair is within $0.04$ nats in posterior KL divergence of the Tikhonov regularizer $\mR_{H^2}$ of \eqref{eq:tikhonov} applied to $\widehat\mA_r$ at the best $\kappa$ chosen separately for each $N$. That best $\kappa$ shifts from $3\times10^{-2}$ at $N = 10^2$ to $10^{-5}$ at $N = 5\times10^4$, so no single $\kappa$ on the grid matches it. Choosing $C_{\rm samp}$ on held-out samples for each realization, with $C_{\rm curv}$ fixed, improves on the fixed pair by more than $0.01$ nats only for $N \le 1000$. Supplement~\ref{sm:s2reg-empirical} reports the comparison in full.

\section{Error at the estimated Stage-1 subspace}
\label{sec:theory}

\Cref{thm:kl-trunc} holds for every $\mV_r \in \St(r,\RR^d)$, in particular at the estimated Stage-1 subspace $\hmV_r^{\mC}$ of \Cref{sec:method-core}: $\DKL\bigl(\pi_{\vn}\,\|\,\pi_{\vn}(\,\cdot\,;\hmV_r^{\mC})\bigr) \le \tfrac12\tr\bigl((\mI-\hmP_r)\mC\bigr)$, where $\hmP_r := \hmV_r^{\mC}\hmV_r^{\mC\top}$. The trace is not computable from the samples, because $\mC$ is the covariance of the unknown copula score. \Cref{prop:diagnostic} replaces $\mC$ by the matrix $\mC_\Lambda$ of Stage 1, defined below. At the estimated subspace the trace of $\mC_\Lambda$ exceeds its minimum by at most a quadratic in the recovery angle, and the tail sum $\widehat E_r$ of \Cref{alg:cas} estimates that minimum with a bound on the error of the estimate.

The \emph{parameterized score} $\scS_\Lambda$ is the $L^2(\pi_{\vZ})$ orthogonal projection of the true copula score $\scS$ onto the Hermite score field $\scF_\Lambda := \{\sum_{\alpha\in\Lambda}\theta_\alpha\nabla H_\alpha\}$ over the Stage-1 multi-index set $\Lambda$ of \eqref{eq:dict}, $\mC_\Lambda := \EE_{\pi_{\vZ}}[\scS_\Lambda(\vZ)\,\scS_\Lambda(\vZ)^\top]$ is its second-moment matrix, and the \emph{Hermite truncation error} $\|\scS - \scS_\Lambda\|_{L^2(\pi_{\vZ})}$ does not increase as $\Lambda$ grows. The \emph{projection subspace} $\mV_{r,\Lambda}$ is the leading rank-$r$ eigenspace of $\mC_\Lambda$, and the \emph{estimated Stage-1 subspace} $\hmV_r^{\mC}$ is the output of Stage 1 on $N$ samples. The Stage-1 solve carries the fixed penalty of \eqref{eq:tikhonov}, so the large-$N$ limit of $\hmV_r^{\mC}$ can differ from $\mV_{r,\Lambda}$, and $\mV_{r,\Lambda}$ serves as a benchmark. Neither is the bound-optimal $\mV_r^{\mC}$ of \Cref{cor:optimal-subspace}, the leading rank-$r$ eigenspace of the true score covariance $\mC$. The angle $\|\sin\Theta(\mV_{r,\Lambda}, \mV_r^{\mC})\|_F$ between $\mV_{r,\Lambda}$ and $\mV_r^{\mC}$ is induced by the Hermite truncation alone. Write $\mP_{r,\Lambda} := \mV_{r,\Lambda}\mV_{r,\Lambda}^\top$, and $E_r(\mC_\Lambda) := \sum_{i>r}\lambda_i(\mC_\Lambda) = \tr((\mI-\mP_{r,\Lambda})\mC_\Lambda)$, the sum of the eigenvalues of $\mC_\Lambda$ below the retained rank, for the \emph{rank-$r$ truncation error of $\mC_\Lambda$}.

\begin{proposition}[Rank-$r$ truncation error and its diagnostic]
\label{prop:diagnostic}
For $1 \le r < d$, let $\|\sin\Theta\|_F := \|\sin\Theta(\hmV_r^{\mC}, \mV_{r,\Lambda})\|_F$ be the principal-angle distance between the estimated Stage-1 subspace and the projection subspace. Over rank-$r$ orthogonal projectors $\mP$, the trace $\tr((\mI-\mP)\mC_\Lambda)$ is smallest at $\mP_{r,\Lambda}$, where it equals $E_r(\mC_\Lambda)$; at the estimated Stage-1 subspace it exceeds this minimum by $\tr((\mI-\hmP_r)\mC_\Lambda) - E_r(\mC_\Lambda) = \tr((\mP_{r,\Lambda}-\hmP_r)\mC_\Lambda) \ge 0$, bounded above and below by quadratics in $\|\sin\Theta\|_F$,
\begin{equation}
  (\lambda_r-\lambda_{r+1})(\mC_\Lambda)\,\|\sin\Theta\|_F^2
  \;\le\;
  \tr\bigl((\mP_{r,\Lambda}-\hmP_r)\,\mC_\Lambda\bigr)
  \;\le\;
  (\lambda_1-\lambda_d)(\mC_\Lambda)\,\|\sin\Theta\|_F^2,
  \label{eq:misalign-quadratic}
\end{equation}
and the sample tail sum $\widehat E_r := \sum_{i>r}\widehat\lambda_i$ of the eigenvalues of $\hmC$ estimates the rank-$r$ truncation error of $\mC_\Lambda$ with
\begin{equation}
  \bigl|\widehat E_r - E_r(\mC_\Lambda)\bigr| \;\le\; \sum_{j=1}^{d-r}\sigma_j\bigl(\hmC - \mC_\Lambda\bigr) \;\le\; (d-r)\,\|\hmC - \mC_\Lambda\|_{\rm op},
  \label{eq:Ehat-weyl}
\end{equation}
where $\sigma_1(\cdot) \ge \cdots \ge \sigma_d(\cdot)$ denote singular values in decreasing order.
\end{proposition}

\emph{Proof.} See \Cref{app:thm-trunc}.

The excess is quadratic in the recovery angle $\|\sin\Theta\|_F$, so an angle of order $\epsilon$ changes the trace by order $\epsilon^2$. The rank-$r$ truncation error of $\mC_\Lambda$ is not a sampling error. It does not depend on $N$ and decreases as $r$ grows. The tail sum $\widehat E_r$ estimates it from the samples alone (\Cref{alg:cas}, step~6), and \eqref{eq:Ehat-weyl} bounds the error of that estimate by the leading singular values of $\hmC - \mC_\Lambda$. The stability result of score ratio matching \cite[Thm.~3.3]{BaptistaBrennanMarzouk2025}, applied with $\scS_\Lambda$ as the approximate score ratio of $\pi_{\vZ}$ to $\gauss_d$, gives $\DKL(\pi_{\vn}\,\|\,\pi_{\vn}(\,\cdot\,;\mV_{r,\Lambda})) \le \|\scS - \scS_\Lambda\|_{L^2(\pi_{\vZ})}^2 + E_r(\mC_\Lambda)$. Part~II \cite{CAS_partII} gives the corresponding bound at the estimated subspace $\hmV_r^{\mC}$ and estimates the Hermite truncation term $\|\scS - \scS_\Lambda\|_{L^2(\pi_{\vZ})}$. That truncation error need not be small for the subspace to be identified (\Cref{lem:pilot-subspace}).

\begin{lemma}[Identification under exact rank-$r$ dependence]
\label{lem:pilot-subspace}
Let $\pi_{\vZ}$ satisfy \Cref{ass:tempered} and suppose the dependence is exactly rank $r$: $c^{Z}(\vz) = c_r^{U}(\mV_r^\top\vz)$ for some $\mV_r \in \RR^{d\times r}$ with $\mV_r^\top\mV_r = \mI_r$, equivalently $\pi_{\vZ} = \pi_{\vU}\otimes\gauss_{d-r}$ in the coordinates $\vu = \mV_r^\top\vz$, $\vw = \mV_\perp^\top\vz$ of \Cref{sec:reduced-form}. Fix $K \geq 2$ and let $\Lambda_{K,d}$ be the index set of \eqref{eq:dict} with no maximum interaction order imposed (any $q \geq K$, since $|\alpha|_0 \le |\alpha|_1$), so that $\Lambda_{K,d} = \{\alpha \in \NN_0^d : 2 \le |\alpha|_1 \le K\}$, and let $\scS_\Lambda$ be the infinite-sample score-matching solution over $\Lambda_{K,d}$, the minimizer of $\EE_{\pi_{\vZ}}[\|\sum_{\alpha}\theta_\alpha\nabla H_\alpha - \scS\|^2]$. Then $\scS_\Lambda(\vz) \in \mathrm{span}(\mV_r)$ for every $\vz$, and $\scS_\Lambda$ depends on $\vz$ only through $\vu$. Consequently $\mC_\Lambda = \mV_r\,\mathbf M\,\mV_r^\top$ for a positive semidefinite $\mathbf M \in \RR^{r\times r}$, and whenever $\mC_\Lambda$ attains rank $r$ its top-$r$ eigenspace equals $\mathrm{span}(\mV_r)$, for any value of $\|\scS - \scS_\Lambda\|_{L^2(\pi_{\vZ})}$.
\end{lemma}

\emph{Proof.} See \Cref{app:pilot-subspace}.

\begin{remark}[Identifying the active subspace with a coarse Stage-1 Hermite expansion]
\label{rem:pilot-subspace}
\Cref{lem:pilot-subspace} separates two requirements. Representing the copula score accurately requires a $\Lambda$ rich enough to make $\|\scS - \scS_\Lambda\|_{L^2(\pi_{\vZ})}$ small. Identifying the active subspace under exact rank-$r$ dependence requires only that $\mC_\Lambda$ have rank $r$.

Under exact rank-$r$ dependence, every eigenvector of $\mC_\Lambda$ with a nonzero eigenvalue lies in the active subspace at every $K \geq 2$, so a coarse $\Lambda$ can miss active directions but cannot introduce complement ones. The top-$r$ eigenspace equals the active subspace as soon as $\scS_\Lambda$ has a nonzero component along every active direction, which is the rank-$r$ condition of the lemma. Only the Stage-2 expansion, in $r$ variables, must represent a score accurately.

The lemma concerns the span at a known $r$. Truncation reweights the nonzero eigenvalues, so the lemma does not justify choosing $r$ by an eigenvalue threshold.

The estimator departs from the lemma in two ways. The maximum interaction order $q_1 < K_1$ can introduce components outside the active subspace, and the dependence need not be exactly rank $r$. \Cref{sec:exp-banana} measures the resulting deviation against the reference $\mV_r^{\mC}$ (\Cref{fig:recovery-row}).
\end{remark}

\section{Experiments}
\label{sec:experiments}

\paragraph{Methods} Four noise estimators are compared: (i) the product of marginals (PoM), the independence assumption of univariate noise modeling; (ii) the Gaussian copula, which admits arbitrary marginals but represents dependence only through a correlation matrix; (iii) PCA-HCSM, which shares the full \textsc{Cas} procedure but chooses its rank-$r$ subspace by variance; and (iv) \textsc{Cas}, the rank-$r$ score-covariance subspace with Hermite copula score matching on the projection. PoM and the Gaussian copula are simplifications named in \Cref{sec:intro}, and PCA-HCSM isolates the effect of the subspace choice.

On each of the example problems defined next, the experiments measure subspace recovery, noise KL divergence, and posterior KL divergence in a Bayesian inference problem. Only (iii) and (iv) choose a rank-$r$ subspace, so the subspace-recovery comparison uses those two, and the density and posterior comparisons use all four. \Cref{thm:kl-trunc} bounds the divergence of the reduction at a given subspace, the part of the noise KL divergence controlled by subspace recovery, and \eqref{eq:posterior-noise-kl} bounds the posterior KL divergence by the noise KL divergence. We report results against the number of samples $N$, over $N \in \{500, 1000, 2000, 5000, 10000, 20000, 50000\}$, with the numbers quoted in the text taken at the largest, $N = 50{,}000$. Per-example diagnostics and the sweep over reduced rank $r$ are reported in Supplement~\ref{sm:rate-real}.

\subsection{Example problems}
\label{sec:exp-benchmark}

The three examples share one construction, with a different map creating the dependence in each. Let $\vW \in \RR^{20}$ have independent standard Gaussian coordinates. In each example two of the first four coordinates are replaced by bounded nonlinear functions of the other two, each perturbed by independent Gaussian noise: in Examples~1 and~3 the pair $(W_0, W_1)$ determines the pair $(W_2, W_3)$, and in Example~2 there are two separate input--output pairs, $W_0$ determining $W_1$ and $W_2$ determining $W_3$. Each replaced coordinate is rescaled so that $\EE[W_j] = 0$ and $\Var(W_j) = 1$ for every $j$, and the remaining $16$ coordinates are left as independent standard Gaussians. Four latent coordinates therefore carry dependence, $r^\star = 4$. The latent vector is then mixed by a uniformly random rotation $\mQ \sim \mathrm{Unif}(\mathrm{O}(8))$, drawn from the Haar probability measure on the orthogonal group, acting on the first $8$ coordinates and as the identity on the remaining $12$, so that the four dependent directions are not axis-aligned; the componentwise map $n_j = \mathrm{logit}(\Phi(z_j))$ is then applied to $\vz = \mQ\vW$. Because $\mQ$ is orthogonal and $\EE[\vW\vW^\top] = \mI_{20}$, the latent coordinates satisfy $\EE[\vz\vz^\top] = \mI_{20}$, so their covariance does not distinguish the four dependent directions.

In Example~1, the primary example problem, the two outputs are the real and imaginary parts of $(W_0 + iW_1)^2$, saturated by $\tanh$:
\begin{equation}
\begin{gathered}
  W_0, W_1 \sim \gauss(0, 1), \\
  W_2 = c\,\tanh\!\bigl((W_0^2 - W_1^2)/M\bigr) + \sqrt{1 - a^2}\,\varepsilon_2, \\
  W_3 = c\,\tanh\!\bigl(2 W_0 W_1/M\bigr) + \sqrt{1 - a^2}\,\varepsilon_3,
\end{gathered}
  \label{eq:ex1-block}
\end{equation}
with $\varepsilon_2, \varepsilon_3 \sim \gauss(0,1)$, saturation scale $M = 2.5$, conditional-mean variance fraction $a = 0.7$, and amplitude $c$ fixed by $\Var(c\tanh(\cdot)) = a^2$.

In this law the latent pairwise correlations vanish, $\EE[W_iW_j] = 0$ for $i \ne j$, so the dependence is not visible in second moments. Its two-to-one map $z \mapsto z^2$ ($z = W_0 + iW_1$) makes the conditional law of the two input coordinates given the two output coordinates bimodal and symmetric under $z \mapsto -z$, and the posterior of \Cref{sec:exp-bip} has two branches.

The reference subspace $\mV_r^{\mC}$ of \Cref{cor:optimal-subspace} is the top-$r$ eigenspace of the true score covariance $\mC := \EE_{\vz \sim \pi_{\vZ}}[\nabla \log c^{Z}(\vz)\, \nabla \log c^{Z}(\vz)^\top]$, which has no closed form. Each mixed coordinate has a marginal density given by a Gaussian mixture over the input coordinates, so the copula score $\nabla\log c^{Z}$ can be evaluated directly; we evaluate $\mC$ by Monte Carlo at $N_{\mathrm{ref}} = 10^6$ from that score and take $\mV_r^{\mC}$ to be its top-$r$ eigenspace (Supplement~\ref{sm:rate-real}). The reference spectrum has top-$4$ eigenvalues $\approx 2.29, 2.28, 0.95, 0.92$, with the remaining $16$ eigenvalues summing to $\approx 0.04$. The cutoff $r^\star = 4$ equals the number of coupled latent coordinates.

\paragraph{Two further examples} Examples~2 and~3 change the map from inputs to outputs. The copula score of Example~2 is harder for a Hermite expansion to represent than that of Example~1, and the copula score of Example~3 is easier.

\emph{Example~2.} Each output is an even function of its input: for $j \in \{0, 2\}$, $W_{j+1} = c\,(\tanh(\gamma(W_j^2 - 1)) - m_0) + \sqrt{1 - a^2}\,\varepsilon$ with $\gamma = 0.5$, $a = 0.7$, centering $m_0 = \EE[\tanh(\gamma(W_j^2 - 1))]$, and variance-normalizing $c$. An even map is non-injective. The score covariance still identifies the active subspace. Its top-$4$ score-covariance eigenvalues are $\approx 2.28, 2.16, 0.95, 0.80$, with the remaining $16$ summing to $\approx 0.09$.

\emph{Example~3.} The two outputs are the real and imaginary parts of $(W_0 + iW_1)^3$: $W_2 = c\,\tanh((W_0^3 - 3 W_0 W_1^2)/M) + \sqrt{1 - a^2}\,\varepsilon_2$ and $W_3 = c\,\tanh((3 W_0^2 W_1 - W_1^3)/M) + \sqrt{1 - a^2}\,\varepsilon_3$, with $M = 4$, $a = 0.7$, $c \approx 1.48$; both output means vanish by symmetry. The score is smooth and the easiest of the three to represent, and the three-to-one map $z \mapsto z^3$ yields a three-modal posterior. Its top-$4$ score-covariance eigenvalues are $\approx 4.42, 4.39, 0.97, 0.94$, with the remaining $16$ summing to $\approx 0.04$.

All three examples appear in every stage of \Cref{sec:exp-banana}--\Cref{sec:exp-bip}, with Example~1 supplying the main reported numbers and the posterior illustration of \Cref{fig:bip-posteriors}.

\subsection{Subspace recovery: score covariance versus PCA}
\label{sec:exp-banana}

The latent covariance does not distinguish the four dependent directions, and this subsection checks that the score covariance finds them and that PCA does not. For $r = r^\star = 4$ we measure each estimate against $\mV_r^{\mC}$ over $20$ independent realizations at each $N$, with the multi-index set of each estimate chosen by cross-validation on its own sample (\Cref{sec:method-core}). The \textsc{Cas} angle decreases with $N$, while the PCA and uniformly random angles are $N$-independent and equal within the across-realization spread (\Cref{fig:recovery-row}, all three examples). The \textsc{Cas} angle also depends on the chosen multi-index set. The cross-validation criterion chooses interaction order $q_1 = 2$ in every realization for $N \le 10{,}000$ and a higher interaction order ($q_1 = 3$ on Examples~1--2, $q_1 = 4$ with even degrees on Example~3) in every realization at $N = 50{,}000$, where the mean angle is smaller than at $N = 10{,}000$ by factors of $16$, $6.0$ and $11$ on Examples~1--3. At $N = 20{,}000$ the higher order is chosen in $11$ and $3$ of the $20$ realizations on Examples~1 and~3, and the band is wider there.

A second comparison removes Stage-2 estimation. It evaluates the divergence of the reduction with the exact reduced density on the subspace, $D_{\mathrm{KL}}(\pi_{\vZ}\,\|\,\pi_{\vZ}(\,\cdot\,;\mV_r))$. The \textsc{Cas} divergence decreases with $N$ toward its value at the reference $\mV_r^{\mC}$ on all three examples, while the PCA divergence does not decrease and lies within the spread of uniformly random subspaces (Supplement~\ref{sec:subspace-only-kl}).

\begin{figure}[t]
\centering
\includegraphics[width=\textwidth]{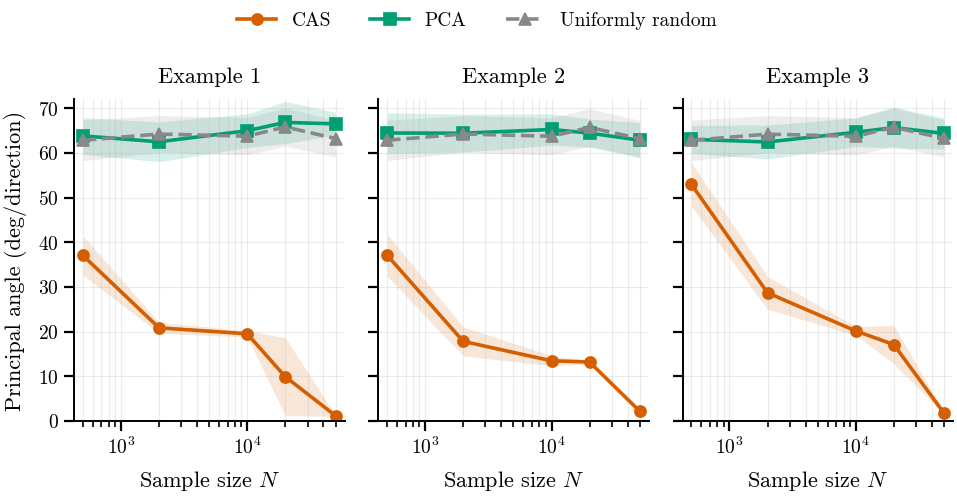}
\caption{Subspace recovery on Examples~1--3 is measured as the per-direction principal angle \cite{BjorckGolub1973} to the bound-optimal subspace $\mV_r^{\mC}$ versus sample size $N$, over $20$ independent realizations with bands of $\pm1$ standard deviation. The \textsc{Cas} angle decreases with $N$ and depends on the chosen multi-index set, while the angles of PCA and of a uniformly random subspace are $N$-independent.}
\label{fig:recovery-row}
\end{figure}

\subsection{Noise KL divergence of the reduced density}
\label{sec:exp-testll}

The rank is $r = r^\star = 4$, and both rank-$r$ methods use the regularizer of \Cref{sec:ridge}. Each method is estimated on a sample of $N$ residuals and evaluated on an independent sample of the same size, over $20$ independent realizations; within a realization all four estimators use the same estimation sample and the same evaluation sample, so the comparison is paired. \textsc{Cas} attains the lowest noise KL divergence (mean $0.128$; \Cref{fig:testll}), versus $0.70$ for PCA-HCSM, $0.70$ for PoM, and $0.70$ for the Gaussian copula. PCA-HCSM and \textsc{Cas} use the same Hermite polynomial basis, regularizer and quasi--Monte Carlo normalizer, so the gap between them is due to the choice of subspace. At $N = 50{,}000$ the across-realization standard deviation is $0.011$ for \textsc{Cas} and at most $0.009$ for the baselines. The same ordering holds on Examples~2 and~3 (\Cref{fig:testll}). As $N$ grows, the three baselines converge to a common noise KL divergence, set by dependence outside all three model classes, and the \textsc{Cas} divergence continues to decrease and is smaller at the largest $N$. Supplement~\ref{sm:rate-real} reports the full sweep over reduced rank $r$.

\begin{figure}[ht]
\centering
\includegraphics[width=\textwidth]{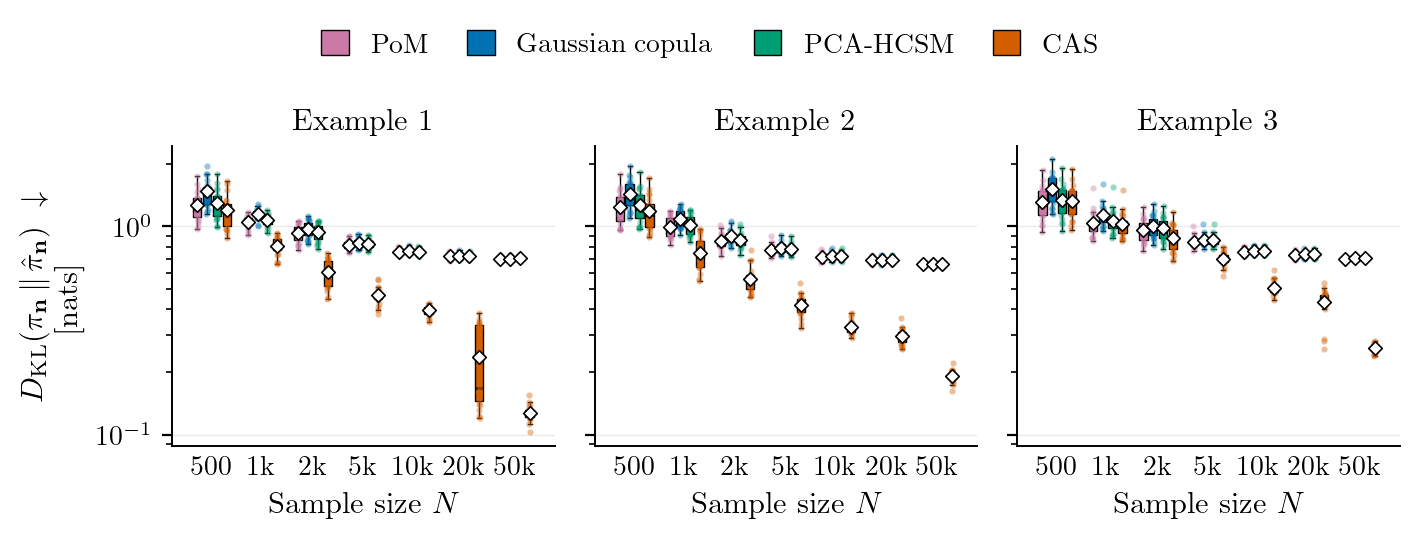}
\caption{Noise KL divergence on Examples~1--3 is shown for $r=4$ across sample sizes, with $20$ independent realizations at each $N$.}
\label{fig:testll}
\end{figure}

\subsection{Bayesian inference problem with non-Gaussian observation noise}
\label{sec:exp-bip}

This subsection measures the effect of the noise estimator on the posterior. The estimated noise density defines the likelihood $\widehat\pi_{\vn}(\vy - \mA\vx)$ of an observation $\vy$, and each candidate posterior is compared with the true-noise posterior.

\paragraph{Problem} The parameter dimension is $d_{\mathrm{par}} = 2$, the observation dimension is $d_{\mathrm{obs}} = 20$, and the prior is $\pi_X(\vx) = \gauss(\vx; 0, \sigma_0^2 I_2)$ with $\sigma_0 = 2$. The forward map observes the two input directions of Example~1, the first two columns of $\mQ$, which are orthonormal:
\begin{equation}
  \vy \;=\; \mA \vx + \vn, \qquad \mA \;=\; \alpha\, \mQ_{:,\,0{:}2}, \qquad \alpha = 1.5.
  \label{eq:bip-forward}
\end{equation}
Varying $\vx$ therefore moves the residual along those two directions alone, and the bimodal conditional law of \Cref{sec:exp-benchmark} gives the posterior two branches. We compare the four noise estimators at rank $r = 4$; for each of the $20$ independent realizations at each $N$, every method is estimated on that realization's residual sample and evaluated on one observation $\vy = \mA\vx^{\star} + \vn$. The KL divergence is computed on an $80\times 80$ grid over $[-6, 6]^2$, with each candidate posterior normalized to unit mass on the grid (\Cref{app:protocol}).

\Cref{fig:bip-kl} aggregates posterior KL divergence across all sample sizes. For $N = 50{,}000$, \textsc{Cas} attains mean posterior KL divergence $0.051$, versus $0.40$ for PCA-HCSM, $0.40$ for PoM, and $0.40$ for the Gaussian-copula baseline. At $N = 500$ the four means lie between $0.40$ and $0.49$ nats. \Cref{fig:bip-posteriors} shows the posteriors for one observation at $N=50{,}000$, chosen because it places mass on both preimage branches. The three baselines each produce a single broad mode between the branches, and only \textsc{Cas} places mass on both. The ranking matches the noise KL divergence of \Cref{fig:testll} and the subspace recovery of \Cref{fig:recovery-row}.

\begin{figure}[t]
\centering
\includegraphics[width=\textwidth]{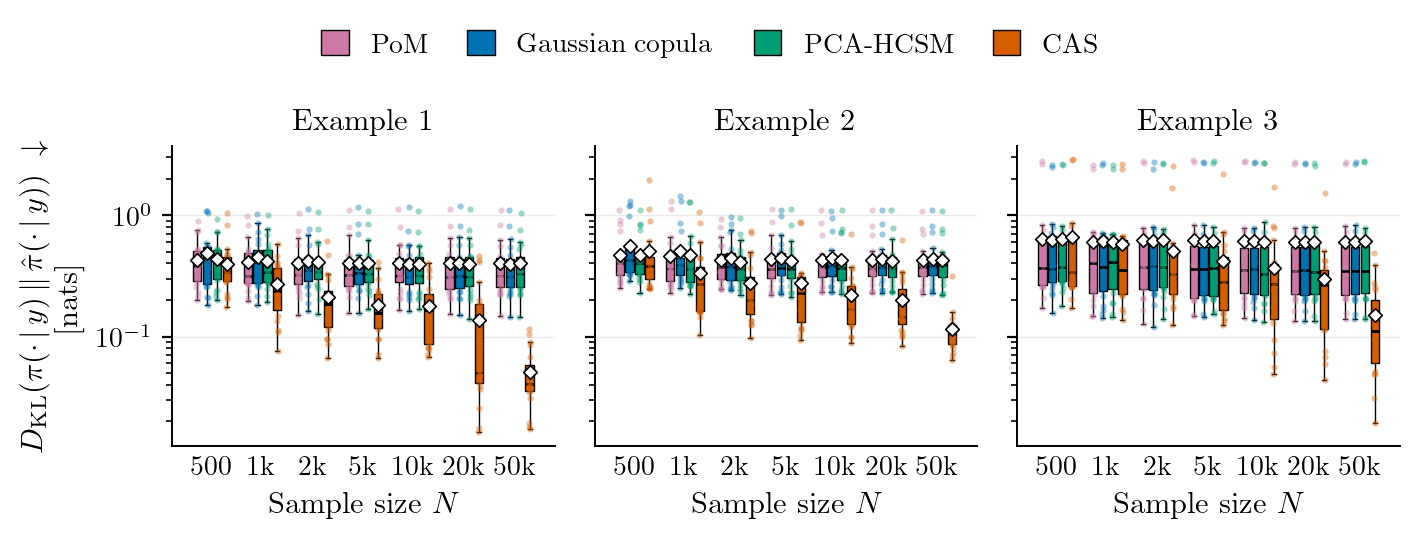}
\caption{Posterior KL divergence from the true-noise posterior is shown on Examples~1--3 across $N$ for $r=4$, with $20$ independent realizations at each $N$.}
\label{fig:bip-kl}
\end{figure}

\begin{figure}[t]
\centering
\includegraphics[width=0.85\textwidth]{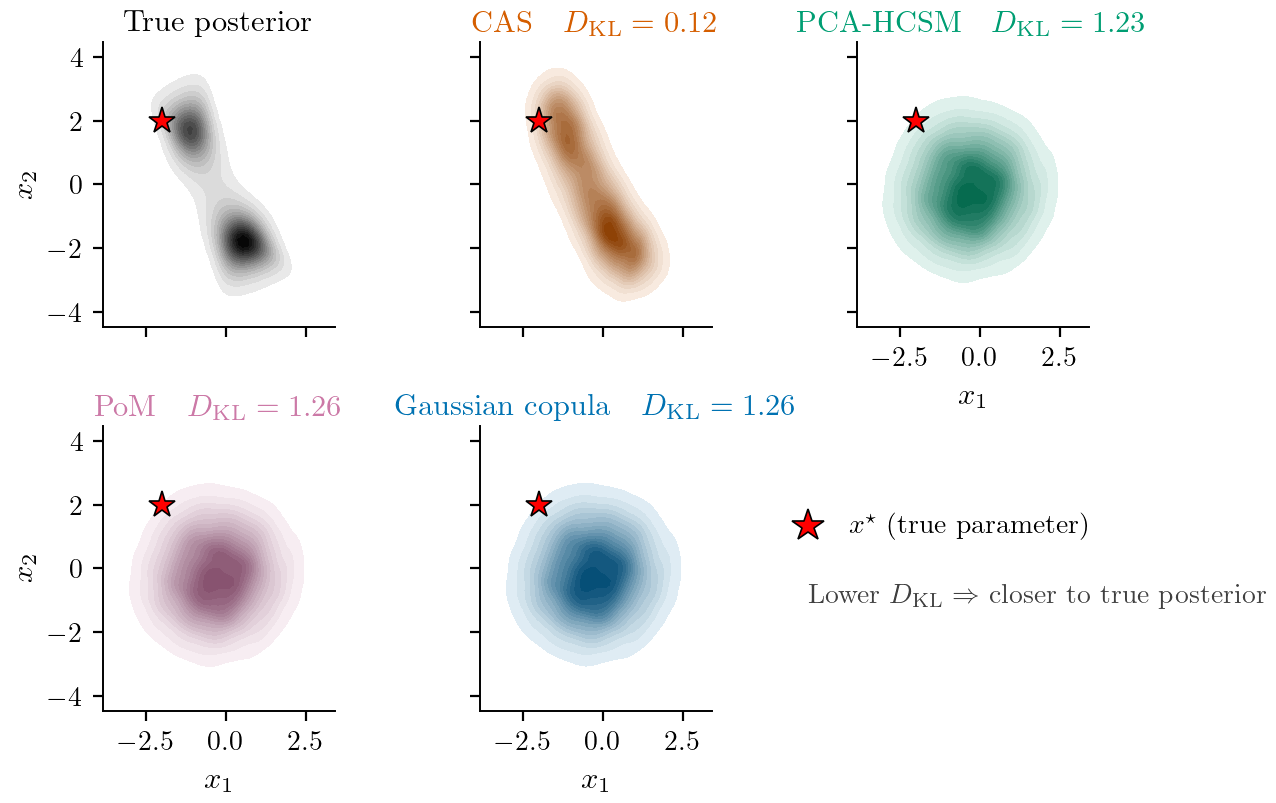}
\caption{Posterior shapes are compared for one observation of the inference problem of \Cref{sec:exp-bip} with $N=50{,}000$, for $r=4$, a $20$-dimensional observation, and a $2$-dimensional parameter; the observation is chosen to place mass on both preimage branches, and the divergences shown are for this observation, and \Cref{fig:bip-kl} gives the means. The true posterior places mass on both branches, and the \textsc{Cas} noise model reproduces both modes, while PCA-HCSM, the product-of-marginals model, and the Gaussian-copula baseline each yield a single broad mode.}
\label{fig:bip-posteriors}
\end{figure}

\section{Limitations and scaling}
\label{sec:limitations}

The reduction is linear in $\vZ$-space: dependence that varies along a curved set of directions is retained only through its variation along the chosen subspace, and transport-based reductions \cite{MarzoukMoselhyParnoSpantini2016, BaptistaMarzoukZahm2023representation, BrennanBigoniZahmSpantiniMarzouk2020, ChenArnaudBaptistaZahm2024} are the nonlinear alternatives. Assumption~\ref{ass:tempered} is a condition on the transformed law. The rank transform absorbs arbitrary continuous marginals, so (T1)--(T2) restrict the dependence structure and impose nothing on the marginal tails, and laws whose density or copula score grows faster than the stated bounds fall outside the analysis. The accuracy of the one-dimensional marginal estimates themselves does depend on the marginal tails, through the kernel choice (\Cref{sec:estimator}). The number of Hermite polynomials $|\Lambda_{K,q}|$ grows as $\binom{d}{q}$ in the interaction order; Supplement~\ref{sm:adaptivity} describes a budget-limited greedy extension for choosing the Stage-1 polynomial space. When the Hermite expansion of the reduced factor converges slowly, Stage~2 can use neural-network score ratio matching or a normalizing flow on the same subspace (\Cref{sec:reduced-model}).

Two further error sources fall outside the accounting above. The randomized eigendecomposition in \Cref{alg:cas} contributes a term to the Stage-1 subspace error, bounded by standard analysis \cite[Thm.~10.6 and \S 10.4]{HalkoMartinssonTropp2011}, and we do not measure it separately. We also do not separate the error of estimating the Stage-2 reduced density on $\hmV_r^{\mC}$ in place of the exact $\mV_r^{\mC}$. The rank transform separates marginal from dependence estimation, and the reduced model does not preserve that separation. Its $i$-th marginal is $\pi_i$ times a factor, and this factor equals one for every reduced factor exactly when row $i$ of $\mV_r$ vanishes. Supplement~\ref{sm:marginals} gives that factor and a bound on its deviation.

Within each estimate the Hermite truncation is fixed, with the truncation and the two constants $(C_{\rm samp}, C_{\rm curv})$ chosen by cross-validation of the score-matching objective (\Cref{sec:ridge}). The triple $(\Lambda, \mR, N)$ is not jointly optimized, and a jointly rate-optimal schedule of truncation and regularization in $N$ is beyond the scope of this work.

Finally, the guarantees of this paper concern population quantities. \Cref{thm:kl-trunc} uses the exact score covariance, and \Cref{prop:diagnostic} does not bound the distance from $\mC_\Lambda$ to $\mC$. The error of the marginal estimates is not bounded here. Part~II \cite{CAS_partII} studies these errors.

All of the above assumes a sample of the noise. When the residuals $\vn^{(i)} = \vy^{(i)} - \mathcal G(\vx^\star)$ stand in for it, they carry the error of the nominal parameter $\vx^\star$ as well as the noise, and the two are not separable from the residuals alone. Obtaining usable samples is a separate problem. The consistency diagnostics of \cite{DesroziersBerreChapnikPoli2005} recover observation-error covariances from innovation and analysis residuals in the Gaussian setting, and expectation--maximization estimates observation error jointly with the state or the model error \cite{DreanoTandeoPulido2017}. The review \cite{TandeoAilliotBocquet2020} covers the innovation-based family. Those methods estimate second moments, and pairing them with an estimator that recovers a full non-Gaussian density is a natural extension of this work.


\section{Conclusion}
\label{sec:conclusion}

\textsc{Cas} estimates non-Gaussian observation noise by componentwise rank transformation followed by a rank-$r$ score-covariance dimension reduction of the resulting copula. The copula score covariance $\mC$ equals zero under independence, so its leading eigenspace consists of dependence directions, and the ordinary covariance need not identify them. The same eigenspace minimizes an explicit KL divergence bound over all rank-$r$ dimension reductions. In the experiments the score covariance recovers the subspace and PCA does not. The recovered subspace lowers the noise KL divergence, and on a Bayesian inference problem the posterior KL divergence falls more than sevenfold at the largest sample size. The truncation orders and the two Stage-2 regularization constants are chosen from data by cross-validation, and the estimator provides $\log\widehat\pi_{\vn}$ and $\nabla\log\widehat\pi_{\vn}$ in closed form wherever the estimated rank transform is differentiable. The companion Part~II \cite{CAS_partII} develops the finite-sample analysis. Two natural extensions are a rotation of the residuals before the componentwise rank transform, and composition with a likelihood-informed subspace \cite{CuiMartinMarzoukSolonenSpantini2014, ZahmCuiLawSpantiniMarzouk2022}.

\section*{Reproducibility}
\label{sec:reproducibility}

Code, saved numerical results, and reproduction instructions are available at \url{https://github.com/joshuawchen/copula_active_subspaces}; the version used for this article is the release tagged \texttt{v1.0-part1}. The script \texttt{reproduce.py} reproduces the figures, using the saved results of the numerical experiments, and recomputes those results on request; the repository README identifies the script behind each figure and table. The \textsc{Cas} estimator is plain Python/NumPy, with no external learned components, and all settings, including the random states of the reported runs, are recorded in \Cref{app:protocol}.

\appendix
\crefalias{section}{appendix}
\crefalias{subsection}{appendix}

\section{Consequences of the regularity assumption}
\label{app:regularity}

This appendix derives the consequences of \Cref{ass:tempered} (stated in \Cref{sec:method}) used in the later proofs. Let $\pi_{\vn}$ be a probability density on $\RR^d$ with continuous marginals $F_1, \ldots, F_d$, and let $c$ be its copula density so that $\pi_{\vZ}(\vz) = c(\Phi(z_1), \ldots, \Phi(z_d))\,\gauss_d(\vz)$ (Sklar's theorem \cite{Sklar1959, Nelsen2006book}), with $\Phi$ the standard Gaussian CDF; thus $c^{Z} = c\circ\Phi$. Writing $\bm\omega := (\Phi(Z_1), \ldots, \Phi(Z_d))$ for $\vZ \sim \gauss_d$ gives $\bm\omega \sim \mathrm{Unif}([0,1]^d)$, so that
\begin{equation}
  \|\scL\|_{L^2(\gauss_d)}^2 \;=\; \EE_{\bm\omega \sim \mathrm{Unif}([0,1]^d)}\bigl[(\log c(\bm\omega))^2\bigr],
  \label{eq:L2-g-copula}
\end{equation}
and finiteness reduces to a property of $c$.

\emph{Consequences.} Five facts follow, each used in the proofs below. First, the marginal $\pi_{\vU} = (\mV_r^\top)_\#\pi_{\vZ}$ on $\RR^r$ inherits a pointwise sub-Gaussian bound from (T1): for any $\mV_r \in \St(r,\RR^d)$,
\begin{equation}
  \pi_{\vU}(\vu) \;\le\; K_1\bigl(1+\|\vu\|^{2m}\bigr)\,\gauss_r(\vu/\sigma_0)
  \qquad\text{for a.e.\ } \vu \in \RR^r,
  \label{eq:piU-envelope}
\end{equation}
with $K_1 = K_1(K_0, \sigma_0, m, r, d) < \infty$. To see it, rotate $\mV_r$ to align with the first $r$ axes, change variables, and integrate the bound in $\vw$. Second, $\nabla\scL \in L^2(\gauss_d;\RR^d)$, by (T2) and Gaussian moments, which is the integrability hypothesis of \cite[Cor.~2.10]{ZahmCuiLawSpantiniMarzouk2022} used in \Cref{thm:kl-trunc}. Third, $\scL \in L^2(\gauss_d)$: the pointwise bound of (T2) makes $\scL$ locally Lipschitz, and integrating along the ray from $\bm 0$ gives $|\scL(\vz)| \le |\scL(\bm 0)| + \alpha'\|\vz\| + \beta'\|\vz\|^{q+1}/(q+1)$, a polynomial in $\|\vz\|$, which lies in every $L^p(\gauss_d)$, $p<\infty$. Fourth, $\tr(\mC)<\infty$, needed to make the right-hand side of \Cref{thm:kl-trunc} finite, equivalently $\scS \in L^2(\pi_{\vZ};\RR^d)$. Under (T1) with finite $\sigma_0$, $\tr(\mC)=\EE_{\pi_{\vZ}}[\|\nabla\scL\|^2] \le \int(\alpha'+\beta'\|\vz\|^q)^2 K_0(1+\|\vz\|^{2m})\gauss_d(\vz/\sigma_0)\,d\vz<\infty$, a polynomial integrated against $\gauss_d(\cdot/\sigma_0)$, so this follows from (T1)--(T2) and need not be assumed separately.

Fifth, the score-matching identity of \Cref{lem:hyv-sm} holds. Its hypotheses \cite[Assumption~3.1]{BaptistaBrennanMarzouk2025} ask for a differentiable model score ratio and density, square-integrable model and true score ratios, and decay of the product of the density and the model score ratio at infinity. For the Hermite parametrization $\scL_r(\vu;\vbeta)=\sum_{\alpha \in \Lambda_r}\theta^{r}_\alpha H_\alpha(\vu)$ with $0\notin \Lambda_r$, every $\partial_j\scL_r(\cdot;\vbeta)$, $\partial_{jj}\scL_r(\cdot;\vbeta)$ and $U_j\partial_j\scL_r(\cdot;\vbeta)$ is a polynomial, so the marginal bound \eqref{eq:piU-envelope} makes each of them integrable and square-integrable under $\pi_{\vU}$ and gives $\partial_j\scL_r(\vu;\vbeta)\,\pi_{\vU}(\vu) \to 0$ as $\|\vu\|\to\infty$. For the true score ratio, differentiation under the integral (dominated by the same bound) gives $\nabla\log\pi_{\vU}(\vu) = \EE[\mV_r^\top\nabla\log\pi_{\vZ}(\vZ) \mid \vU = \vu]$, so conditional Jensen yields $\EE_{\pi_{\vU}}[\|\nabla\log\pi_{\vU}\|^2] \le \EE_{\pi_{\vZ}}[\|\nabla\log\pi_{\vZ}\|^2]$, finite since $\nabla\log\pi_{\vZ} = \scS - \vz$ with $\scS \in L^2(\pi_{\vZ};\RR^d)$ and $\EE[\|\vZ\|^2] = d$; with $\EE[\|\vU\|^2] \le d$, $\partial_j\scL_r = \partial_j\log\pi_{\vU} + U_j \in L^2(\pi_{\vU})$. The same bounds center both scores: $\nabla\pi_{\vZ} = \pi_{\vZ}\,(\scS - \vz)$ lies in $L^1(\RR^d;\RR^d)$, so $\pi_{\vZ}\in W^{1,1}(\RR^d)$ and $\EE_{\pi_{\vZ}}[\nabla\log\pi_{\vZ}] = \int\nabla\pi_{\vZ} = 0$. The same argument on $\RR^r$ gives $\EE_{\pi_{\vU}}[\nabla\log c_r^{U}] = \EE_{\pi_{\vU}}[\nabla\log\pi_{\vU}] + \EE[\vU] = 0$, since the marginals of $\vZ$ are centered. These are the identities invoked at \eqref{eq:L-T-C-def} and \eqref{eq:centering-mat}.

Strict positivity of $c$ on $(0,1)^d$ also rules out perfect-dependence singularities, and $\scZ$ is invertible almost surely through the generalized inverses of the continuous marginals, so the KL divergence is invariant under it.


\section{Score matching: the Hermite linear system}
\label{app:stage2-foundations}

This appendix gives the score-matching identity and the explicit Hermite linear system used in both stages: Stage 1 on $\RR^d$ (\Cref{sec:method-core}) and Stage 2 on the projected coordinates $\RR^r$ (\Cref{sec:reduced-model}).

\paragraph{Stage-1 linear system on $\RR^d$} Score matching \cite{Hyvarinen2005} estimates $\vtheta$ by minimizing the Fisher divergence between the model $\vz$-density score $\nabla\log\pi_{\vZ}(\cdot;\vtheta)$ and the true one, after which the copula score $\scS(\cdot;\widehat\theta)$ follows from \eqref{eq:hermite-expansion}. An integration by parts (the score-matching identity \cite[Thm.~1]{Hyvarinen2005}, whose hypotheses \Cref{lem:hyv-sm} below verifies under Assumption~\ref{ass:tempered}) removes the unknown true score from the objective, leaving a quadratic in $\vtheta$,
\begin{equation}
  \begin{aligned}
  J(\vtheta) &= \EE_{\pi_{\vZ}}\!\bigl[\tfrac12 \|\nabla\log\pi_{\vZ}(\vZ;\vtheta)\|^2 + \divop \nabla\log\pi_{\vZ}(\vZ;\vtheta)\bigr]\\
  &= \tfrac12 {\vtheta}^{\top} \mA \, \vtheta - \mb^\top \vtheta + \mathrm{const},
  \end{aligned}
  \label{eq:hyv-quadratic}
\end{equation}
with $\mA_{\alpha\beta} = \EE_{\pi_{\vZ}}[\nabla H_\alpha(\vZ) \cdot \nabla H_\beta(\vZ)]$ and $\mb_\alpha = \EE_{\pi_{\vZ}}[\vZ \cdot \nabla H_\alpha(\vZ) - \Delta H_\alpha(\vZ)]$. The Hermite polynomial derivative recurrence
\begin{equation}
  \partial_i H_\alpha(\vz) \,=\, \sqrt{\alpha_i}\, H_{\alpha - e_i}(\vz),
  \qquad
  \partial_i^2 H_\alpha(\vz) \,=\, \sqrt{\alpha_i(\alpha_i-1)}\, H_{\alpha - 2 e_i}(\vz)
  \label{eq:deriv-rec}
\end{equation}
expresses $\nabla H_\alpha$ and $\Delta H_\alpha$ in the Hermite basis. Replacing $\pi_{\vZ}$-expectations by sample averages over $\{\vZ^{(n)}\}_{n=1}^N$ gives
\begin{equation}
\begin{aligned}
  \widehat\mA &:= \frac1N\sum_{n=1}^N A(\vZ^{(n)}), & A(\vz)_{\alpha\beta} &:= \nabla H_\alpha(\vz) \cdot \nabla H_\beta(\vz), \\
  \widehat\mb &:= \frac1N\sum_{n=1}^N b(\vZ^{(n)}), & b(\vz)_\alpha &:= \textstyle\sum_{i=1}^d \bigl[z_i\sqrt{\alpha_i}H_{\alpha-e_i}(\vz) - \sqrt{\alpha_i(\alpha_i-1)}H_{\alpha-2e_i}(\vz)\bigr].
\end{aligned}
  \label{eq:A-b-formulas}
\end{equation}
We call this \emph{Hermite copula score matching} (HCSM, \Cref{alg:hsm}): a map from $\vz$-coordinate samples in $\RR^D$, a multi-index set $\Lambda \subset \NN_0^D$, and a regularizer $\mR$ to the coefficient estimate $\widehat\theta = (\widehat\mA+\mR)^{-1}\widehat\mb$. Equation~\eqref{eq:A-b-formulas} is HCSM in $D=d$ dimensions on $\{\vZ^{(n)}\}$ with $\Lambda = \Lambda_{K_1, q_1}$ and $\mR = \mR_{H^2}$ of \eqref{eq:tikhonov}, and Stage 2 invokes it in $D=r$ dimensions on the projected samples with the regularizer of \Cref{sec:ridge}.

\begin{algorithm}[h]
\caption{Hermite copula score matching (HCSM).}
\label{alg:hsm}
\small
\begin{algorithmic}[1]
\Require samples $\{\vw^{(k)}\}_{k=1}^N \subset \RR^D$; multi-index set $\Lambda \subset \NN_0^D$ with $0 \notin \Lambda$; regularizer $\mR \succeq 0$ with $\widehat\mA + \mR \succ 0$, given as a fixed matrix or as a function of the Gram matrix $\widehat\mA$; optional linear constraint matrix $\mathbf B \in \RR^{c\times|\Lambda|}$ of full row rank (default: none).
\Ensure coefficient vector $\widehat\theta \in \RR^{|\Lambda|}$ of the parametric log-density $\scL(\vw;\widehat\theta) = \sum_{\alpha\in\Lambda}\widehat\theta_\alpha H_\alpha(\vw)$.

\State \textbf{Build derivative-feature matrices} $\mPhi_j \in \RR^{N\times|\Lambda|}$, $(\mPhi_j)_{n\alpha} = \partial_j H_\alpha(\vw^{(n)})$, via the Hermite polynomial derivative recurrence \eqref{eq:deriv-rec}.

\State \textbf{Assemble empirical system} $(\widehat\mA, \widehat\mb) = \frac{1}{N}\sum_n (A(\vw^{(n)}), b(\vw^{(n)}))$ from \eqref{eq:A-b-formulas}.

\State If $\mR$ is data-dependent, evaluate $\mR \gets \mR(\widehat\mA)$. Set $\widehat\mM := \widehat\mA + \mR$.

\State \textbf{Solve} $\widehat\mM\widehat\theta = \widehat\mb$ by Cholesky if no $\mathbf B$; else solve the constrained KKT system, equivalently $\widehat\theta = \widehat\mM^{-1}\widehat\mb - \widehat\mM^{-1}\mathbf B^\top(\mathbf B\widehat\mM^{-1}\mathbf B^\top)^{-1}\mathbf B\widehat\mM^{-1}\widehat\mb$. \Return $\widehat\theta$.
\end{algorithmic}
\end{algorithm}

Both stages use the following identity.

\begin{lemma}[Score-matching identity for the reduced score]
\label{lem:hyv-sm}
Under Assumption~\ref{ass:tempered}, the unweighted Fisher divergence \eqref{eq:JW-pop} satisfies
\begin{equation}
  J(\vbeta)
  \;=\;\EE_{\pi_{\vU}}\!\Big[\sum_{j=1}^r\Big(\tfrac12\bigl(\partial_j\scL_r(\vU;\vbeta)\bigr)^2-U_j\,\partial_j\scL_r(\vU;\vbeta)+\partial_{jj}\scL_r(\vU;\vbeta)\Big)\Big]+C,
  \label{eq:sm-hyv}
\end{equation}
with $C=\tfrac12\sum_j\EE[(\partial_j\scL_r(\vU))^2]$, the true-score energy, independent of $\vbeta$.
\end{lemma}

\begin{proof}
With $w := \nabla\scL_r(\cdot;\vbeta)$ and reference $\rho := \gauss_r$, for which $\nabla\log\rho(\vu) = -\vu$, the integrand on the right of \eqref{eq:sm-hyv} is $\tfrac12 w^\top w + \tr(\nabla w) + \nabla\log\rho^\top w$, and $J(\vbeta) = \tfrac12\EE_{\pi_{\vU}}\|w - \nabla\log(\pi_{\vU}/\rho)\|^2$. The identity is therefore the score-ratio-matching identity \cite[Thm.~3.2]{BaptistaBrennanMarzouk2025} without the conditioning variable, which is the integration by parts of \cite[Thm.~1]{Hyvarinen2005} written relative to the Gaussian reference. Its hypotheses hold under Assumption~\ref{ass:tempered} by \Cref{app:regularity}. The Stage-1 identity used in \eqref{eq:hyv-quadratic} is the same statement on $\RR^d$ with $\pi_{\vZ}$ in place of $\pi_{\vU}$, with (T1) in place of \eqref{eq:piU-envelope}.
\end{proof}

\section{Proof of \Cref{prop:diagnostic}}
\label{app:thm-trunc}

\begin{proof}[Proof of \Cref{prop:diagnostic}]
Over rank-$r$ orthogonal projectors $\mP$, the trace $\tr((\mI-\mP)\mC_\Lambda)$ is smallest at $\mP_{r,\Lambda}$, with minimum $E_r(\mC_\Lambda)$ \cite[Prob.~III.6.11]{Bhatia1997}, so the excess at $\hmP_r$ is $\tr((\mP_{r,\Lambda}-\hmP_r)\mC_\Lambda)$. The lower bound in \eqref{eq:misalign-quadratic} is the curvature lemma \cite[Lem.~4.2]{VuLei2013} with $A = \mC_\Lambda$, $E = \mP_{r,\Lambda}$ and $F = \hmP_r$, and it holds trivially when $\lambda_r = \lambda_{r+1}$. For the upper bound, let $\vv_1,\dots,\vv_d$ be an orthonormal eigenbasis of $\mC_\Lambda$ and $a_i := \vv_i^\top\hmP_r\vv_i \in [0,1]$, so that $\sum_i a_i = \tr(\hmP_r) = r$ and
\[
  \tr\bigl((\mP_{r,\Lambda}-\hmP_r)\mC_\Lambda\bigr) \;=\; \sum_{i\le r}\lambda_i(1-a_i) - \sum_{j>r}\lambda_j a_j \;\le\; (\lambda_1-\lambda_d)\,s,
\]
with $s := \sum_{i\le r}(1-a_i) = \sum_{j>r}a_j = \|(\mI-\hmP_r)\mV_{r,\Lambda}\|_F^2 = \|\sin\Theta\|_F^2$. For \eqref{eq:Ehat-weyl}, $\widehat E_r - E_r(\mC_\Lambda)$ is a sum of $d-r$ entries of the vector $(\widehat\lambda_i - \lambda_i(\mC_\Lambda))_{i=1}^d$, which by Lidskii's theorem \cite[Ex.~III.4.3]{Bhatia1997} lies in the convex hull of the permutations of the eigenvalue vector of $\hmC - \mC_\Lambda$. The sum is therefore at most, in modulus, the sum of the $d-r$ largest moduli of the eigenvalues of the symmetric matrix $\hmC - \mC_\Lambda$, which are its singular values $\sigma_1,\dots,\sigma_{d-r}$, each at most $\|\hmC-\mC_\Lambda\|_{\rm op}$.
\end{proof}

\section{Centering constraint: proof of \Cref{prop:centering}}
\label{app:centering}

\begin{proof}[Proof of \Cref{prop:centering}]
\emph{Projector form.} The centered estimator \eqref{eq:stage2-closedform} minimizes the strictly convex quadratic $\tfrac12{\vbeta}^\top\widehat\mM\vbeta - {\vbeta}^\top\widehat\mb_r$ subject to $\widehat\mB\vbeta = 0$, an equality-constrained quadratic program whose first-order conditions are the KKT system \cite[eq.~(16.4)]{NocedalWright2006}. Eliminating the multiplier by the Schur complement of $\widehat\mM$ in that system \cite[\S16.2, eq.~(16.16)]{NocedalWright2006} yields the closed form \eqref{eq:stage2-closedform}, which rearranges to $\hvbeta = (\mI - \Pi_{\widehat\mB})\widehat\mM^{-1}\widehat\mb_r = (\mI - \Pi_{\widehat\mB})\hvbeta_{\rm unc}$. The matrix $\Pi_{\widehat\mB}$ is idempotent with $\widehat\mM\Pi_{\widehat\mB}$ symmetric, hence the $\widehat\mM$-orthogonal projector onto $\mathrm{range}(\widehat\mM^{-1}\widehat\mB^\top)$. Both the elimination and the uniqueness of $(\hvbeta, \lambda)$ require the Schur complement $\widehat\mB\widehat\mM^{-1}\widehat\mB^\top$ to be nonsingular, which holds when $\widehat\mB$ has full row rank, since $\widehat\mM \succ 0$ \cite[Lem.~16.1]{NocedalWright2006}.

\emph{Bias decomposition.} These are deterministic limiting matrices, so \eqref{eq:centering-bias} is an exact identity. From $\bm b^\star = \mA^\star\bm\theta^{r,\star} = (\bm M - \mR)\bm\theta^{r,\star}$ the unconstrained limiting solve is $\bm M^{-1}\bm b^\star = \bm\theta^{r,\star} - \bm M^{-1}\mR\bm\theta^{r,\star}$. Applying $\mI-\Pi^\star$ and subtracting $\bm\theta^{r,\star}$,
\[
  \bar{\bm\theta} - \bm\theta^{r,\star} = (\mI-\Pi^\star)(\bm\theta^{r,\star} - \bm M^{-1}\mR\bm\theta^{r,\star}) - \bm\theta^{r,\star} = -(\mI-\Pi^\star)\bm M^{-1}\mR\bm\theta^{r,\star} - \Pi^\star\bm\theta^{r,\star},
\]
which is \eqref{eq:centering-bias}. Since $\Pi^\star$ is $\bm M$-orthogonal, $\|\mI-\Pi^\star\|_{\bm M} \le 1$, so the first term does not exceed $\|\bm M^{-1}\mR\bm\theta^{r,\star}\|_{\bm M}$ in $\bm M$-norm.

\emph{Basis-truncation identity.} The limiting centering matrix acts as $(\mB^\star\vbeta)_j = \EE_{\pi_{\vU}}[\partial_j\scL_r(\vU;\vbeta)]$, the mean reduced score of the model. The truncated minimizer $\bm\theta^{r,\star}$ has score $\Pi_{\Lambda_r}\nabla\scL_r$, the $L^2(\pi_{\vU})$ projection of the true reduced score onto the Stage-2 multi-index set, so $\mB^\star\bm\theta^{r,\star} = \EE_{\pi_{\vU}}[\Pi_{\Lambda_r}\nabla\scL_r] = \EE_{\pi_{\vU}}[\nabla\scL_r] - \EE_{\pi_{\vU}}[\bm\rho] = -\EE_{\pi_{\vU}}[\bm\rho]$, using $\EE_{\pi_{\vU}}[\nabla\scL_r] = 0$ (\Cref{app:regularity}) and $\bm\rho := \nabla\scL_r - \Pi_{\Lambda_r}\nabla\scL_r$. Substituting into $\Pi^\star\bm\theta^{r,\star} = \bm M^{-1}\mB^{\star\top}(\mB^\star\bm M^{-1}\mB^{\star\top})^{-1}\mB^\star\bm\theta^{r,\star}$ gives
\[
\begin{gathered}
  \Pi^\star\bm\theta^{r,\star} = -\bm M^{-1}\mB^{\star\top}\bigl(\mB^\star\bm M^{-1}\mB^{\star\top}\bigr)^{-1}\EE_{\pi_{\vU}}[\bm\rho],\\
  \|\Pi^\star\bm\theta^{r,\star}\|_{\mA^\star} \le \bigl\|\bm M^{-1}\mB^{\star\top}(\mB^\star\bm M^{-1}\mB^{\star\top})^{-1}\bigr\|_{\mA^\star}\,\|\EE_{\pi_{\vU}}[\bm\rho]\|,
\end{gathered}
\]
the bound finite by nonsingularity of the Schur complement and zero iff $\EE_{\pi_{\vU}}[\bm\rho] = 0$.

\emph{Asymptotic law.} The map $\varphi(\mB, \mM, \mb) := (\mI - \mM^{-1}\mB^\top(\mB\mM^{-1}\mB^\top)^{-1}\mB)\,\mM^{-1}\mb$ gives $\hvbeta = \varphi(\widehat\mB, \widehat\mM, \widehat\mb_r)$ and $\bar{\bm\theta} = \varphi(\mB^\star, \bm M, \bm b^\star)$, and it is continuously differentiable near $(\mB^\star, \bm M, \bm b^\star)$, where $\mB^\star\bm M^{-1}\mB^{\star\top}$ is invertible. The entries of $(\widehat\mB, \widehat\mM, \widehat\mb_r)$ are sample means of Hermite features with finite second moments by \eqref{eq:piU-envelope}. The multivariate central limit theorem and the delta method \cite[Thm.~3.1]{VanDerVaart1998} give asymptotic normality of $\sqrt N(\hvbeta - \bar{\bm\theta})$, and in particular $\|\hvbeta - \bar{\bm\theta}\| = O_p(N^{-1/2})$.

\emph{Variance non-increase.} For any random $\bm\xi$ with finite second moment, write $\widetilde\Pi := \bm M^{1/2}\Pi^\star\bm M^{-1/2}$, an orthogonal projection by $\bm M$-self-adjointness of $\Pi^\star$. Then
\begin{multline*}
  \tr\bigl(\mA^\star\Cov((\mI-\Pi^\star)\bm\xi)\bigr) \le \tr\bigl(\bm M\Cov((\mI-\Pi^\star)\bm\xi)\bigr) \\
  = \tr\bigl((\mI-\widetilde\Pi)\,\bm M^{1/2}\Cov(\bm\xi)\bm M^{1/2}\,(\mI-\widetilde\Pi)\bigr) \le \tr\bigl(\bm M\Cov(\bm\xi)\bigr),
\end{multline*}
using $\mA^\star \preceq \bm M$ (as $\mR \succeq 0$) and that an orthogonal projection does not increase the trace. This is \eqref{eq:centering-var}.
\end{proof}


\section{The constrained Stage-2 solve}
\label{app:solver}

The estimator \eqref{eq:stage2-est} minimizes the strictly convex quadratic of \eqref{eq:stage2-closedform} over the convex set $\mathcal C$, assumed nonempty. A polynomial is bounded above on all of $\RR^r$ only if its degree is even and its leading form is nonpositive on the unit sphere. Let $K_{\rm eff}$ be the largest even integer at most $K_2$ and $F[\vbeta](\vv) := \sum_{|\alpha|_1 = K_{\rm eff}}\theta^r_\alpha\,\vv^\alpha/\sqrt{\alpha!}$ the corresponding leading form, the coefficients of degree $K_2$ being set to zero when $K_2$ is odd. We impose the stronger condition $F[\vbeta](\vv) \le -\varepsilon$ for $\|\vv\| = 1$, together with the level bound $\scL_r(\vu;\vbeta) \le \tau$ on the ball $\|\vu\| \le r_{\max}$ containing the data (radius and tolerances in \Cref{app:protocol}). The margin gives $\scL_r(\vu;\vbeta) \le -\varepsilon\|\vu\|^{K_{\rm eff}} + O(\|\vu\|^{K_{\rm eff}-1})$ as $\|\vu\|\to\infty$, so the two families together bound $\scL_r$ above on $\RR^r$; the level $\tau$ is imposed on the ball only. The two families can conflict when $\tau$ is small (for $r = 1$ and $K_2 = 4$, any leading form with margin $\varepsilon$ forces $\scL_r \ge 2\varepsilon$ somewhere on every ball of radius at least one), which is why nonemptiness is a hypothesis of the ideal problem; what the solver does when its search does not meet the level bound is described below.

One alternative would ask that $-F[\vbeta] - \varepsilon\|\vv\|^{K_{\rm eff}}$ be a nonnegative form and verify this by a semidefinite program. A semidefinite program, however, verifies only sums of squares, in general a strict subset of the nonnegative forms \cite{Hilbert1888, Blekherman2006}, and that route would exclude admissible coefficient vectors. We instead impose the two constraint families directly. Each is linear in the coefficients at a fixed point $\vv$ or $\vu$, and the finite point set is enlarged iteratively. The quadratic program is solved under the current set, $F[\vbeta]$ is maximized over the sphere and $\scL_r(\cdot;\vbeta)$ over the ball, the maximizing points are added, and the iteration repeats until both maxima are within tolerance of their bounds. The iteration is the general algorithm of \cite[\S 2.1]{BlankenshipFalk1976}, and Theorem~2.1 there states that every accumulation point of the iterates solves the program, under continuity of the constraint functions and compactness of the index sets: each constraint here is a polynomial in the coefficients and the index point jointly, and the index sets are the sphere and the ball; \cite{HettichKortanek1993} is a survey. The theorem's remaining hypothesis, compactness of the decision set, is used for the existence of the iterates and of an accumulation point, and both are supplied by the objective: each approximating problem minimizes the strictly convex coercive quadratic over a closed set that contains the feasible set, so its solution exists, is unique, and has objective value at most the constrained minimum, and the iterates therefore lie in one compact sublevel set. The feasible set is convex and the objective strictly convex, so the minimizer is unique.

The theorem assumes exact inner maximizations. The implementation maximizes numerically, $F[\vbeta]$ over $2\times10^5$ random directions on the sphere and $\scL_r(\cdot;\vbeta)$ over $8\times10^4$ random points in the ball, each followed by local refinement from the best candidates; a violating point is added as a cut, together with the points at $0.9$ and $1.1$ times its norm for the level family, the outer of which can lie just outside the ball; and the iteration stops when the largest value found is within $\varepsilon/2$ of the leading-form bound and within $10^{-3}$ of the level bound, or after $30$ rounds, in which case the last iterate is kept and the fit records which family is unmet. Feasibility of the returned $\hvbeta$ is checked by that search to those tolerances; a finite set of pointwise inequalities does not certify a polynomial inequality on the sphere. In the experiments of \Cref{sec:experiments}, at rank $4$, the search met both tolerances in every fit; in the rank scan of Supplement~\ref{sec:pca-ablation} some fits at ranks $6$ and $8$ reach the round cap with the level bound unmet, and the composed estimate \eqref{eq:cas-estimator} is normalizable regardless. The map $\widehat\mb_r \mapsto \hvbeta$ is piecewise affine for a fixed finite cut set.


\section{Identification under exact rank-$r$ dependence: proof of \Cref{lem:pilot-subspace}}
\label{app:pilot-subspace}
\begin{proof}[Proof of \Cref{lem:pilot-subspace}]
The objective is quadratic in $\bm\theta$, so $\scS_\Lambda$ is the $L^2(\pi_{\vZ};\RR^d)$-orthogonal projection of $\scS$ onto $\mathcal D := \mathrm{span}\{\nabla H_\alpha : \alpha \in \Lambda_{K,d}\}$, and the projection does not depend on the basis chosen for $\mathcal D$. The polynomials $\{H_\alpha : \alpha \in \Lambda_{K,d}\}$ span every polynomial of total degree at most $K$ that is $\gauss_d$-orthogonal to constants and to linear functions; both conditions are invariant under rotations of $\RR^d$, so in the coordinates $(\vu, \vw)$ the same span is generated by the products $H_\beta(\vu)H_{\beta''}(\vw)$ with $2 \le |\beta|_1 + |\beta''|_1 \le K$. Exact rank-$r$ dependence gives $\scS(\vz) = \mV_r\,\nabla_{\vu}\log c_r^{U}(\vu)$ and $\pi_{\vZ} = \pi_{\vU}\otimes\gauss_{d-r}$. For a product $p = H_\beta(\vu)H_{\beta''}(\vw)$ the gradient is $\nabla p = \mV_r(\nabla_{\vu}H_\beta)H_{\beta''} + \mV_\perp H_\beta(\nabla_{\vw}H_{\beta''})$. Suppose $\beta'' \neq 0$ and pair $\nabla p$ with $\scS$ in $L^2(\pi_{\vZ};\RR^d)$: the second summand vanishes pointwise because $\mV_\perp^\top\mV_r = 0$, and the first has expectation $\EE_{\pi_{\vU}}[(\nabla_{\vu}H_\beta)\cdot\nabla_{\vu}\log c_r^{U}]\;\EE_{\gauss_{d-r}}[H_{\beta''}] = 0$, since $H_{\beta''}(\vW)$ is independent of $\vU$ and $\EE_{\gauss_{d-r}}[H_{\beta''}] = 0$. Pairing $\nabla p$ with $\nabla g$ for any polynomial $g$ of the active coordinates alone gives zero by the same two observations. The gradients of the products with $\beta'' \neq 0$ therefore span a subspace of $\mathcal D$ orthogonal both to $\scS$ and to the gradients of the polynomials of the active coordinates alone, so the projection of $\scS$ onto $\mathcal D$ equals its projection onto the latter family, and every member of that family has the form $\mV_r\nabla_{\vu}g(\vu)$. This proves the first claim. The remaining claims follow from $\mC_\Lambda = \EE[\scS_\Lambda\scS_\Lambda^\top] = \mV_r\,\EE[\bm m\bm m^\top]\,\mV_r^\top$ with $\bm m := \mV_r^\top\scS_\Lambda$, whose range lies in $\mathrm{span}(\mV_r)$ and equals it exactly when the rank is $r$; thus $\mathbf M = \EE[\bm m\bm m^\top]$.
\end{proof}


\section{Details of the numerical experiments}
\label{app:protocol}

All experiments are reproducible from the released code (Python, NumPy/SciPy) under the fixed random states recorded there. Each realization of \Cref{sec:experiments} at sample size $N$ and seed $s$ draws its estimation and evaluation samples in that order from one generator seeded at $7000 + s$, its observation from a generator seeded at $20000 + s$, its Stage-1 and Stage-2 dictionary selections from split generators seeded at $s$ and $s + 1$, and the split of \Cref{sec:ridge} from a generator seeded at $20260715 + s$; the posterior illustration of \Cref{fig:bip-posteriors} uses seeds $59000$ and $20025$. Rank-Gaussianization resolves ties by average ranks; the transform $\widehat F_i$ takes the values $R_i^{(k)}/(N+1)$ at the sample points, is interpolated between them by monotone piecewise cubic interpolation \cite{FritschCarlson1980}, and takes the values $1/[2(N+1)]$ and $1 - 1/[2(N+1)]$ beyond the sample range, so the composed density jumps at the sample extremes; the density factors $\widehat\pi_i$ are the kernel estimates. The normalizer $\widehat Z_{\vn}$ of \eqref{eq:cas-estimator} is estimated by scrambled Sobol' quasi--Monte Carlo \cite{Sobol1967, Owen1995} against the product of the kernel estimates with $2^{17}$ points, and the Gaussian-copula baseline is normalized the same way. On one fit per example at $N = 50{,}000$, three independent scrambles give $\log\widehat Z_{\vn}$ within a range of $0.0081$ nats at $2^{17}$ points, with mean within $0.002$ nats of that at $2^{19}$ points. The randomized range finder \cite{HalkoMartinssonTropp2011} of \Cref{alg:cas} uses oversampling $p = 10$ with $q_{\text{pwr}} = 2$ power-iteration steps. The multi-index sets of \Cref{sec:method-core} are chosen from Stage-1 candidates with total degree $K_1 = 4$ and interaction orders $q_1 \in \{2, 3, 4\}$ plus even-degree variants, and Stage-2 candidates $(K_2, q_2) \in \{(4,3), (5,3), (6,3), (5,4), (6,4)\}$, taking the smallest multi-index set among candidates within one standard error of the minimum cross-validation criterion \cite[\S 7.10]{HastieTibshiraniFriedman2009}. The constants are chosen per estimate on the grid $C_{\rm samp} \in \{0.3, 0.95, 3, 9.5, 30\}$, $C_{\rm curv} \in \{10^{-6}, 3.2\times10^{-6}, 10^{-5}, 3.2\times10^{-5}, 10^{-4}\}$, scored on half the Stage-2 sample, and the estimate is refitted on the whole of it at the selected pair; the released code names them \texttt{c\_cov} and \texttt{c\_sob}. PCA-HCSM uses the Stage-2 multi-index set selected for \textsc{Cas} on the same sample, the same regularizer with its constants chosen by the same rule on the same grid, the same direct projected coordinates, and the same normalization. The Stage-2 solve of \Cref{app:solver} uses the leading-form margin $\varepsilon = 10^{-3}$ on the unit sphere and enforces the level bound on the ball of radius $r_{\max} = \max(18,\ 1.5\max_n\|\vu_n\|)$. The posterior KL divergence of \Cref{sec:exp-bip} normalizes each candidate on the grid with equal cell weights, drops grid points at which the true posterior is below $10^{-15}$, and floors the candidate at $10^{-15}$.
The amplitude and centering constants of Examples~2 and~3 are computed once by $129$-point Gauss--Hermite quadrature \cite{GolubWelsch1969}, and both laws have closed-form $\log\pi_{\vn}$ with direct forward simulation. The parameter dimension is kept at $d_{\mathrm{par}} = 2$ so that every posterior is tractable on a uniform grid, isolating the noise-estimator effect from high-dimensional posterior sampling. All remaining settings are stated where the corresponding experiments are reported in \Cref{sec:experiments}.

\clearpage
\def\siamprelabel{SM}
\def\siampretitle{Supplementary Materials: }
\let\section\CASsection
\renewcommand\thesection{\siamprelabel\arabic{section}}
\crefalias{section}{section}
\crefalias{subsection}{subsection}
\setcounter{section}{0}\setcounter{equation}{0}\setcounter{figure}{0}
\setcounter{table}{0}\setcounter{theorem}{0}\setcounter{algorithm}{0}
\setcounter{page}{1}
\headers{Copula Active Subspaces}{J.~Chen and P.~J.~van Leeuwen}
\begin{center}
{\bfseries\MakeUppercase{Supplementary Materials: Copula Active Subspaces I: A Score-Covariance Method for Reduced-Order Non-Gaussian Density Estimation}\par}
\vskip .075in
{\footnotesize\scshape Joshua Chen \scriptsize AND \footnotesize Peter Jan van Leeuwen}
\end{center}
\vskip .11in

These supplementary materials carry the empirical support and the optional methodological extensions of Part~I.
\begin{itemize}
\item \Cref{sm:rate-real} measures the dependence of Stage-1 subspace recovery on $(K_1, q_1)$, reports the PCA subspace baseline on Example~1 and the Bayesian inference problem, compares subspaces with Stage-2 estimation removed, and measures the accuracy of the tail-sum diagnostic $\widehat E_r$ against a numerical reference.
\item \Cref{sm:s2reg-empirical} compares the Stage-2 regularizer $\mR$ across the coefficient $C$, specifies the validation choice of the two constants and the centering constraint, and sets $\mR$ against $\mR_{H^2}$ at the best $\kappa$ chosen separately for each $N$, on all three examples.
\item \Cref{sm:adaptivity} describes an optional greedy choice of the Stage-1 degree and interaction order using reusable quadratic solves, with its costs and the meaning of its stopping rule.
\item \Cref{sm:marginals} gives the marginals of the reduced model and bounds their deviation from the true marginals.
\end{itemize}
Throughout, $\vn$ denotes the observation noise in its original (pre-rank-transform) coordinates, with law $\pi_{\vn}$, and all other notation follows the main manuscript. Numbers carrying the SM prefix (\S\ref{sm:rate-real}, \Cref{tab:pca-banana}, \eqref{eq:subspace-only-kl}) refer to these supplementary materials, and numbers without it refer to the main manuscript.

\section{Stage-1 subspace convergence and PCA subspace baseline}
\label{sm:rate-real}

This section reports the dependence of Stage-1 subspace recovery on $(K_1, q_1)$ and the per-$(r, N)$ numbers behind the PCA baseline of \S\ref{sec:exp-testll}.

On the examples of \S\ref{sec:exp-benchmark} the copula score need not be represented exactly by a fixed Hermite multi-index set, which does not by itself imply a subspace error (\Cref{lem:pilot-subspace}). Main-paper \Cref{fig:recovery-row} measures the combined effect of the multi-index set, the Stage-1 penalty and sampling against the reference $\mV_r^{\mC}$ at $N_{\rm ref} = 10^6$, at the multi-index sets the cross-validation criterion chooses; the angle is smaller where a richer set is chosen. The score-covariance spectra of the three example problems appear side by side in \Cref{fig:eigvals-all}, with Example~1 in the left panel.

\begin{figure}[h]
\centering
\includegraphics[width=\textwidth]{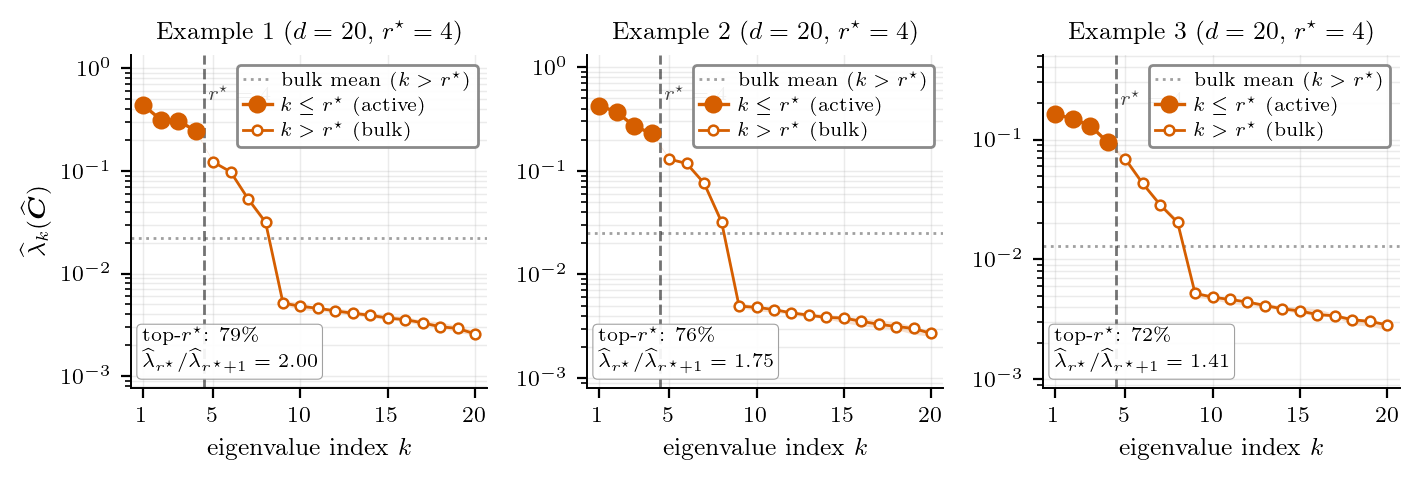}
\caption{Eigenvalue decay of the \emph{estimated} copula-score covariance $\hmC$ is shown on the three example problems at $N=50{,}000$ and the Stage-1 multi-index set $(K_1, q_1) = (4,2)$, as the median across five realizations with the band spanning their range, and with Example~1 in the left panel, Example~2 in the middle, and Example~3 on the right. Filled markers are the top-$r^\star$ eigenvalues, hollow markers the remaining $k > r^\star$, the dashed vertical line $r^\star = 4$, and the dotted horizontal line the mean of the eigenvalues above $r^\star$. These are eigenvalues of $\hmC$, which is assembled from the parameterized score at a coarse multi-index set and is accordingly smaller in scale than the reference $\mC$ whose spectrum \S\ref{sec:exp-benchmark} quotes. The panels show the shape of the decay and the gap at $r^\star$.}
\label{fig:eigvals-all}
\end{figure}

\subsection*{PCA subspace baseline: per-$(r, N)$ numbers}\label{sec:pca-ablation}

We compare \textsc{Cas} with the PCA baseline of \S\ref{sec:exp-testll}. The two use the same rank transform, Stage-2 estimate on the projected coordinates and evaluation. \textsc{Cas} uses $\hmV_r^{\mC}$, the top-$r$ eigenvectors of $\hmC$; PCA uses $\hmV_r^{\widehat\Sigma}$, those of the ordinary covariance $\widehat\Sigma$ on the rank-Gaussianized data. This comparison sweeps $r \in \{2, 3, 4, 6, 8\}$ at $N \in \{500, 1000, 2000\}$ over $20$ independent realizations, and evaluates each estimate on a fifth of its own sample, withheld from the estimate. The comparisons of \S\ref{sec:exp-testll} use an independent evaluation sample.

\paragraph{Example~1 noise law (\S\ref{sec:exp-banana})} \Cref{tab:pca-banana} summarizes the out-of-sample log-likelihood at each $(r, N)$ pair. PCA is flat across $r \in \{2, \ldots, 8\}$ for each $N$ (mean out-of-sample log-likelihood $\approx -40.41$ at $N = 1000$, $\approx -40.29$ at $N = 2000$). Adding rank does not help, because PCA's variance directions miss the non-Gaussian dependence. \textsc{Cas} is better in $12$ of the $15$ pairs and within the $0.05$-nat margin in the other three, all at $N = 500$; the gap reaches $\approx 0.35$ nats for $r = 6$ at $N = 2000$ and increases with $N$ in every row. At ranks $6$ and $8$ the Stage-2 solve of \Cref{app:solver} reaches its round cap with the level bound unmet in $40$ and $70$ of the $120$ fits, respectively, and none at ranks $2$--$4$; those fits keep the solver's last iterate.

\begin{table}[h]
\centering\footnotesize
\setlength{\tabcolsep}{4pt}
\caption{Mean out-of-sample log-likelihood on Example~1 ($d = 20$) is reported over $20$ independent realizations. Bold marks the better estimator at each $(r, N)$ pair ($\geq 0.05$-nat margin). \textsc{Cas} is better in $12$ of the $15$ pairs; the other three, at $N = 500$ for $r \in \{2, 3, 8\}$, are within the margin.}
\label{tab:pca-banana}
\begin{tabular}{r|rr|rr|rr}
\toprule
& \multicolumn{2}{c|}{$N{=}500$} & \multicolumn{2}{c|}{$N{=}1000$} & \multicolumn{2}{c}{$N{=}2000$} \\
$r$ & \textsc{Cas} & PCA & \textsc{Cas} & PCA & \textsc{Cas} & PCA \\
\midrule
2  & $-40.71$ & $-40.75$ & $\mathbf{-40.34}$ & $-40.41$ & $\mathbf{-40.22}$ & $-40.29$ \\
3  & $-40.70$ & $-40.75$ & $\mathbf{-40.26}$ & $-40.41$ & $\mathbf{-40.11}$ & $-40.30$ \\
4  & $\mathbf{-40.68}$ & $-40.74$ & $\mathbf{-40.19}$ & $-40.41$ & $\mathbf{-39.99}$ & $-40.29$ \\
6  & $\mathbf{-40.67}$ & $-40.73$ & $\mathbf{-40.17}$ & $-40.40$ & $\mathbf{-39.94}$ & $-40.29$ \\
8  & $-40.68$ & $-40.72$ & $\mathbf{-40.20}$ & $-40.39$ & $\mathbf{-39.95}$ & $-40.28$ \\
\bottomrule
\end{tabular}
\end{table}

\paragraph{Inference-problem example (\S\ref{sec:exp-bip})} On the Bayesian inference problem at $N = 50{,}000$, PCA does not reduce the posterior KL divergence by more than $3\%$ at any rank. Its mean KL divergence ($0.45$--$0.46$ across $r \in \{2, \ldots, 8\}$) differs by at most $0.013$ from the Gaussian-copula baseline ($0.46$), because the variance-ranked subspace does not contain the score-covariance directions. \textsc{Cas} for $r = 4$ reduces the KL divergence by $50\%$ (mean $0.23$ versus $0.46$). The two estimators share every other component, so the gap is due to the subspace.

\begin{table}[h]
\centering\footnotesize
\caption{The Bayesian inference problem under the Example~1 noise law is evaluated over $30$ observations at $N = 50{,}000$, a separate run from the comparison of \S\ref{sec:exp-bip}: each estimator is computed once and held fixed, so the variation across trials comes from the observation alone. PCA improves on the Gaussian baseline by at most $3\%$ at any $r$, while \textsc{Cas} at $r = 4$ reduces the KL divergence by $50\%$.}
\label{tab:pca-bip}
\begin{tabular}{lcccc}
\toprule
Method & mean $\DKL$ & median $\DKL$ & worst $\DKL$ & reduction vs.\ Gauss \\
\midrule
Gaussian            & $0.46$ & $0.34$ & $1.38$ & n/a \\
Product of marginals & $0.46$ & $0.34$ & $1.38$ & $0\%$ \\
PCA, $r=2$  & $0.46$ & $0.34$ & $1.37$ & $0\%$ \\
PCA, $r=3$  & $0.46$ & $0.35$ & $1.37$ & $-1\%$ \\
PCA, $r=4$  & $0.46$ & $0.35$ & $1.37$ & $-1\%$ \\
PCA, $r=6$  & $0.46$ & $0.33$ & $1.37$ & $0\%$ \\
PCA, $r=8$  & $0.45$ & $0.33$ & $1.35$ & $3\%$ \\
\textsc{Cas}, $r=4$ & $\mathbf{0.23}$ & $\mathbf{0.17}$ & $\mathbf{0.96}$ & $\mathbf{50\%}$ \\
\bottomrule
\end{tabular}
\end{table}

\paragraph{Interpretation} The comparison isolates the choice of subspace: \textsc{Cas} and PCA-HCSM share the rank transform, the Stage-2 estimate, and the evaluation, so the difference between them is what identifying $\mV_r$ from the score covariance adds over identifying it from the ordinary covariance. That difference reaches $\approx 0.35$ nats of out-of-sample log-likelihood on Example~1's noise law ($r = 6$, $N = 2000$), and a $50\%$ KL divergence reduction on the inference problem, against at most $3\%$ for PCA at any rank.

\subsection*{Subspace-only KL divergence: the reduction with Stage 2 removed}
\label{sec:subspace-only-kl}

The comparisons above and in \S\ref{sec:exp-testll} evaluate the full estimator, so the subspace and the Stage-2 estimate contribute together. This comparison removes Stage-2 estimation entirely and compares subspaces alone. For a candidate $\mV_r$ we evaluate
\begin{equation}
  D_{\mathrm{KL}}\bigl(\pi_{\vZ}\,\|\,\pi_{\vZ}(\,\cdot\,;\mV_r)\bigr) \;=\; \EE_{\pi_{\vU}}\bigl[D_{\mathrm{KL}}(\pi_{\vW\mid\vU}\,\|\,\gauss_{d-r})\bigr],
  \label{eq:subspace-only-kl}
\end{equation}
the divergence of the reduction with the exact reduced density on that subspace, as in \S\ref{sec:reduced-form} of the main manuscript. \Cref{fig:random-subspace} reports it on the three examples for $\mV_r$ estimated by \textsc{Cas}, by PCA, and drawn uniformly at random. The \textsc{Cas} divergence decreases with $N$ toward its value at the reference $\mV_r^{\mC}$ on all three. The PCA divergence does not decrease and lies within the spread of the uniformly random draws. The latent construction of \S\ref{sec:exp-banana} leaves no covariance structure that distinguishes the active directions.

\begin{figure}[h]
\centering
\includegraphics[width=\textwidth]{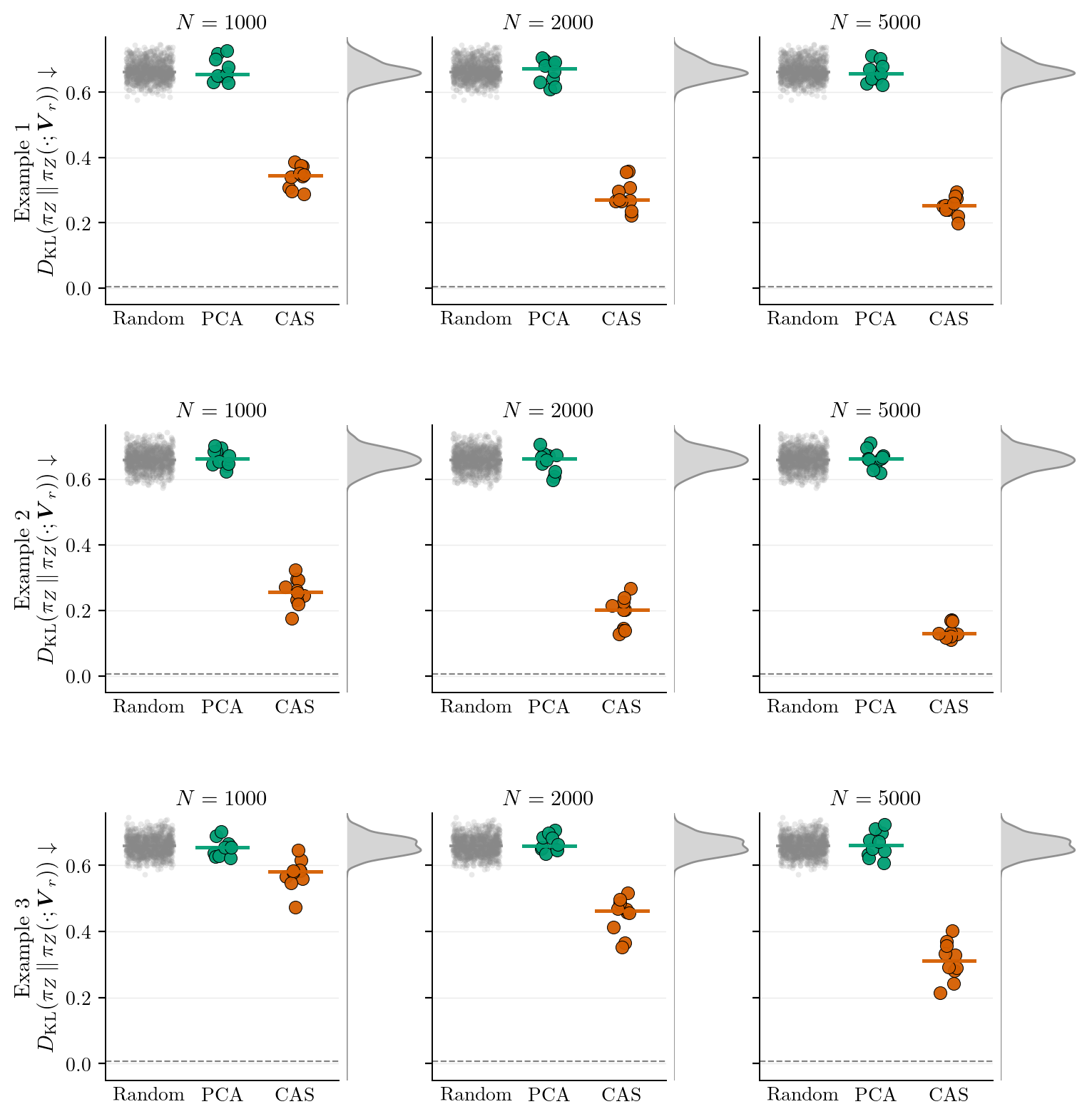}
\caption{The subspace-only KL divergence \eqref{eq:subspace-only-kl} across the three examples compares the choice of subspace with Stage-2 estimation removed, and lower is better. The gray cloud shows $J=100$ uniformly random subspaces in each of the $10$ realizations, the green markers show PCA, the orange markers show \textsc{Cas}, and the dashed line marks the value at the reference $\mV_r^{\mC}$. The \textsc{Cas} divergence decreases with $N$ toward it on all three examples, while the PCA divergence does not decrease and lies within the random-subspace cloud.}
\label{fig:random-subspace}
\end{figure}

The estimate of \eqref{eq:subspace-only-kl} is nested Monte Carlo with outer size $N_{\mathrm{outer}} = 1500$ and inner size $M_{\mathrm{inner}} = 4096$, with the copula density evaluated through marginal estimates fitted on $2\times10^6$ samples, for $N \in \{1000, 2000, 5000\}$ with $10$ independent realizations at each $N$. The uniformly random subspaces are spanned by $\mV_r^{\mathrm{rand}} = \mathrm{QR}(\mG)$ with $\mG \sim \gauss^{d\times r}$, whose column space is uniform on $\Gr(r,\RR^d)$ \cite{Mezzadri2007}, drawn independently for each realization.

\subsection*{Accuracy of the tail-sum diagnostic}
\label{sec:ehat-tightness}

\Cref{prop:diagnostic} bounds the error of the tail sum $\widehat E_r$ by the singular values $\sigma_j$ of the error of $\hmC$, $|\widehat E_r - E_r(\mC_\Lambda)| \le \sum_{j\le d-r}\sigma_j(\hmC - \mC_\Lambda) \le (d-r)\,\|\hmC - \mC_\Lambda\|_{\rm op}$ (\Cref{eq:Ehat-weyl}), and the same inequalities hold with any symmetric matrix in place of $\mC_\Lambda$. \Cref{tab:ehat-tightness} evaluates them against a numerical reference $\mC_{\rm ref}$, the score second-moment matrix at the Stage-1 coefficients fitted at $N_{\rm ref} = 10^6$, which carry the penalty of \eqref{eq:tikhonov}, evaluated by Monte Carlo on $2\times10^5$ fresh samples; its error column measures disagreement with that reference, not with $\mC_\Lambda$. It reports, per example at $N = 50{,}000$ over $20$ independent realizations, at the Stage-1 multi-index set $(K_1, q_1) = (4, 2)$ of the main experiments and rank $r = 4$: the sample tail sum $\widehat E_r$, the reference value $E_r(\mC_{\rm ref})$, the largest error $\max_s |\widehat E_r - E_r(\mC_{\rm ref})|$ across realizations, the across-realization means of $\sum_{j\le d-r}\sigma_j$ and of $(d-r)\,\sigma_1$, and the across-realization mean of the per-realization ratio of the first bound to the error.

\begin{table}[h]
\centering\small
\setlength{\tabcolsep}{3pt}
\caption{The accuracy of the tail-sum diagnostic and the tightness of \Cref{eq:Ehat-weyl} are reported per example at $N = 50{,}000$ with $20$ independent realizations, at the Stage-1 multi-index set $(K_1,q_1) = (4,2)$ of the main experiments and rank $r = 4$. Columns give the tail sum $\widehat E_r$ (mean $\pm$ one standard deviation across realizations), the reference $E_r(\mC_{\rm ref})$, the largest error across realizations, the across-realization means of the two bounds of \Cref{eq:Ehat-weyl}, $\sum_{j\le d-r}\sigma_j(\hmC-\mC_{\rm ref})$ and $(d-r)\,\sigma_1(\hmC-\mC_{\rm ref})$, and the mean per-realization ratio of the first bound to the error.}
\label{tab:ehat-tightness}
\begin{tabular}{l|cccccc}
\toprule
 & $\widehat E_r$ & $E_r(\mC_{\rm ref})$ & $\max_s|\widehat E_r - E_r(\mC_{\rm ref})|$ & $\sum_{j\le d-r}\sigma_j$ & $(d{-}r)\,\sigma_1$ & ratio \\
\midrule
Example~1 & $0.3490 \pm 0.0061$ & $0.2884$ & $0.0729$ & $0.1118$ & $0.3498$ & $1.9$ \\
Example~2 & $0.3994 \pm 0.0063$ & $0.3388$ & $0.0691$ & $0.1168$ & $0.3470$ & $1.9$ \\
Example~3 & $0.2058 \pm 0.0045$ & $0.1483$ & $0.0664$ & $0.0823$ & $0.2083$ & $1.4$ \\
\bottomrule
\end{tabular}
\end{table}

The ratio column is the factor by which the first bound of \Cref{eq:Ehat-weyl} exceeds the actual error. The error $\widehat E_r - E_r(\mC_{\rm ref}) = \sum_{i>r}(\widehat\lambda_i - \lambda_i(\mC_{\rm ref}))$ is a signed sum of $d-r$ eigenvalue differences, while the first bound sums the $d-r$ largest singular values of $\hmC-\mC_{\rm ref}$ without signs, and the second replaces every singular value by the largest. The two bound columns thus separate the slack from sign cancellation and the slack from the uniform norm.

\Cref{tab:ehat-tightness} shows three things. First, the error is one-signed: $\widehat E_r > E_r(\mC_{\rm ref})$ on every realization of every example. Second, the Ky Fan bound \cite{Bhatia1997} exceeds the realized error by mean factors of $1.9$, $1.9$ and $1.4$ on Examples~1--3 (the ratio column). Third, the operator-norm bound $(d-r)\,\sigma_1$ exceeds $E_r(\mC_{\rm ref})$ itself on all three examples: at $d - r = 16$ that form of \Cref{eq:Ehat-weyl} is uninformative here, and only the Ky Fan form is informative. The measured error is the error column and is smaller than either bound.

\section{\texorpdfstring{Stage-2 regularizer: empirical comparisons across $C$}{Stage-2 regularizer: empirical comparisons across C}}
\label{sm:s2reg-empirical}

\S\ref{sec:ridge} of the main paper establishes the basis invariance and the conditioning bound of the Stage-2 regularizer $\mR$, and the companion analysis \cite{CAS_partII} bounds its Fisher divergence error. This section compares $\mR$ across the coefficient $C$ and against $\mR_{H^2}$ with the best $\kappa$ chosen separately for each $N$. The comparisons in \S\ref{sec:sm2-tr-ablations}--\ref{sec:sm2-dw-vs-h} are against $\mR_{H^2}$ on Example~1, and \S\ref{sec:sm-cross-benchmark} compares $\mR$ with alternative designs on all three examples. The Stage-2 solves in this section are the centered closed form \eqref{eq:stage2-closedform}, without the inequality constraints of \Cref{app:solver} and with the fitted potential capped above at evaluation. For noise estimates the compared quantities are the held-out log-likelihood $\EE_{\pi_{\vn}}[\log q]$ and log-score $\EE_{\pi_{\vn}}[\log\pi_{\vn} - \log q]$, where $q$ is the estimate normalized by a $2^{13}$-point reduced normalizer alone in place of $\widehat Z_{\vn}$ of \eqref{eq:cas-estimator}; each differs from its normalized counterpart by the estimate's log-mass, and a paired difference by the difference of two log-masses. Posterior divergences are normalized on the grid.

\subsection{Choice of the two constants}
\label{sec:sm2-selection}
The two constants $(C_{\rm samp}, C_{\rm curv})$ of \S\ref{sec:ridge} are chosen from the sample by cross-validation, with no hand-tuning and no reference density. The Stage-2 sample is partitioned into an estimation part and a validation part, and the score-matching quadratic of \S\ref{sec:ridge} is assembled on each part. For each candidate pair on the $5\times5$ grid of \Cref{app:protocol}, the estimator is solved on the estimation part and scored by the validation quadratic. The minimizing pair is chosen, and the estimator is re-estimated on the full sample at the chosen values. For fixed coordinates and a candidate fitted without the validation observations, the validation quadratic estimates the reduced Fisher divergence up to a candidate-independent constant; here the transform and the subspace are estimated before the split, so we use it as an empirical tuning criterion, which requires no normalizer and no noise model. The pair $(C_{\rm samp}, C_{\rm curv}) = (3, 10^{-5})$ is the center of that grid and the fixed pair at which the comparisons of \Cref{sec:sm2-tr-ablations,sec:sm2-dw-vs-h} hold the proposal.

\subsection{\texorpdfstring{Stage-2 coefficient $C$: stability and margin over $\mR_{H^2}$}{Stage-2 coefficient C: stability}}
\label{sec:sm2-tr-ablations}

We assess the stability of the Stage-2 coefficient $C := C_{\rm samp}$ and the margin of $\mR$ over $\mR_{H^2}$ on Example~1 of \S\ref{sec:exp-banana}. Holding the subspace, rank, and marginals fixed, we compare $\mR$ for $C \in \{1, 3, 10\}$ against $\mR_{H^2}$ (with $\kappa = 10^{-3}$) on paired out-of-sample log-likelihood across $r \in \{2,3,4,6,8\}$, $N \in \{500,1000,2000\}$, $20$ independent realizations ($15$ $(r, N)$ pairs). At each $(r, N)$ and each realization all four schemes share the \textsc{Cas} subspace, so the paired gap isolates the Stage-2 regularizer. \Cref{tab:s2reg-ablations} reports the statistics over the $15$ pairs. The broader family of alternative regularizer designs (identity- and Sobolev-scaled, in fixed and $c/N$ variants, and an unregularized solve) is compared separately in the cross-example comparison of \S\ref{sec:sm-cross-benchmark}, which makes the cross-family robustness point across all three examples.

\begin{table}[h]
\centering\small
\caption{Stage-2 $C$-stability on Example~1 is measured as the paired out-of-sample-LL gap ($\mR$ with coefficient $C$ minus $\mR_{H^2}$ with $\kappa = 10^{-3}$) across $r \in \{2,3,4,6,8\}$ and $N \in \{500,1000,2000\}$ with $20$ independent realizations; statistics are over the $15$ $(r, N)$ pairs. A pair is ``lost'' when its mean gap is below $-0.05$ nats and ``won'' when above $+0.05$. All three coefficient values improve on $\mR_{H^2}$ on average and none loses a pair; the proposed $C=3$ has the largest mean gain, and the largest-magnitude worst gap, $-0.036$ nats at $C=10$, is within the $0.05$-nat margin. Bold marks the proposed row's worst gap and lost count.}
\label{tab:s2reg-ablations}
\begin{tabular}{l|rrr|rr}
\toprule
scheme & mean gap & min gap & max gap & lost & won \\
\midrule
$\mR$, $C=1$  & $+0.076$ & $-0.015$ & $+0.431$ & 0 & 7 \\
$\mR$, $C=3$ (proposed) & $+0.101$ & $\mathbf{-0.015}$ & $+0.583$ & $\mathbf{0}$ & 7 \\
$\mR$, $C=10$ & $+0.095$ & $-0.036$ & $+0.614$ & 0 & 6 \\
\bottomrule
\end{tabular}
\end{table}

All three coefficient values improve on $\mR_{H^2}$ on average across the grid (mean gaps $+0.076$, $+0.101$, and $+0.095$ nats for $C = 1, 3, 10$), winning $7$, $7$, and $6$ of the $15$ pairs, with gains of up to $0.43$, $0.58$, and $0.61$ nats. No coefficient value loses a pair by more than the $0.05$-nat margin. The worst gaps are $-0.015$ nats at $C = 1$ and $C = 3$ and $-0.036$ nats at $C = 10$, the over-regularized end. Of the three, the center value $C=3$ attains the largest mean gain.

\subsection{\texorpdfstring{Full $\mR_{H^2}$ $\kappa$-grid and the effect of a fixed $\kappa$}{Full RH2 kappa-grid and the effect of a fixed kappa}}
\label{sec:sm2-dw-vs-h}

This subsection presents the direct comparison of the proposed regularizer against $\mR_{H^2}$ on the inference problem of \S\ref{sec:exp-bip}, plus the full $\kappa$-grid for $\mR_{H^2}$. The main numbers are summarized in \Cref{tab:stage2-headline}: the proposed estimator with fixed $(C_{\rm samp}, C_{\rm curv}) = (3, 10^{-5})$ is within $0.04$ nats of $\mR_{H^2}$ at its per-$N$ best $\kappa$ at every $N$, below it at $N = 500$, while that best $\kappa$ shifts from $3\cdot10^{-2}$ at $N = 100$ into the band $\kappa \le 10^{-3}$ for $N \ge 2{,}500$.

\begin{table}[h]
\centering\small
\caption{Posterior KL divergence ($\downarrow$) is reported across $N$ for $r=4$ on the inference problem of \S\ref{sec:exp-bip}, with $50$ paired trials at each $N$. ``$\mR_{H^2}$ best'' uses the per-$N$ best $\kappa$ from a validation grid (unavailable in practice), and ``proposed'' uses fixed $(C_{\rm samp}, C_{\rm curv}) = (3, 10^{-5})$. Bold marks the best per row; ties at the displayed precision are both bolded.}
\label{tab:stage2-headline}
\begin{tabular}{r|rrcr}
\toprule
$N$ & no reg. & $\mR_{H^2}$ best & best $\kappa$ & proposed (fixed) \\
\midrule
$100$      & $3.02$ & $\mathbf{0.89}$ & $3\!\cdot\!10^{-2}$ & $0.92$\\
$500$      & $0.85$ & $0.64$ & $10^{-2}$ & $\mathbf{0.62}$\\
$2{,}500$  & $\mathbf{0.28}$ & $\mathbf{0.28}$ & $3\!\cdot\!10^{-4}$ & $0.31$\\
$12{,}500$ & $\mathbf{0.22}$ & $\mathbf{0.22}$ & $10^{-5}$ & $0.26$\\
$50{,}000$ & $\mathbf{0.21}$ & $\mathbf{0.21}$ & $10^{-5}$ & $0.25$\\
\bottomrule
\end{tabular}
\end{table}

Two further questions are how $\mR_{H^2}$ behaves across the full $\kappa$ grid, and how far a single fixed $\kappa$ falls from the per-$N$ optimum. \Cref{tab:dw-kappa-grid} reports the full grid: mean \textsc{Cas} posterior KL divergence across the Example~1 inference problem for $r=4$, eight $\kappa$ values $\times$ five $N$ values. The minimum across each row is the best entry of \Cref{tab:stage2-headline}. The grid shows two patterns. The optimum drifts with $N$: for $N=100$ the largest value on the grid, $\kappa = 3\cdot10^{-2}$, is best, and from $N \geq 2{,}500$ the row is flat within $0.01$ nats across $\kappa \in [10^{-5}, 10^{-3}]$, with larger $\kappa$ worse. The row optimum and the fixed column $\kappa = 10^{-3}$ separate by $0.31$ nats at $N = 100$, by $0.06$ at $N = 500$, and by less than $0.01$ for $N \ge 2{,}500$.

\begin{table}[h]
\centering\small
\caption{Mean \textsc{Cas} posterior KL divergence under $\mR_{H^2}$ is reported across $\kappa$ and $N$, with $50$ trials at each $N$. Bold marks the minimum across each row (the best $\kappa$ for that $N$, used in \Cref{tab:stage2-headline}); ties at the displayed precision are both bolded. The column $\kappa = 10^{-3}$ is the midpoint of the grid, used below as the single fixed choice.}
\label{tab:dw-kappa-grid}
\begin{tabular}{r|rrrrrrrr}
\toprule
$N \backslash \kappa$ & $10^{-5}$ & $3\cdot 10^{-5}$ & $10^{-4}$ & $3\cdot 10^{-4}$ & $10^{-3}$ & $3\cdot 10^{-3}$ & $10^{-2}$ & $3\cdot 10^{-2}$ \\
\midrule
$100$       & $2.92$ & $2.74$ & $2.28$ & $1.71$ & $1.21$ & $0.98$ & $0.91$ & $\mathbf{0.89}$ \\
$500$       & $0.85$ & $0.84$ & $0.81$ & $0.77$ & $0.69$ & $\mathbf{0.64}$ & $\mathbf{0.64}$ & $0.66$ \\
$2{,}500$   & $\mathbf{0.28}$ & $\mathbf{0.28}$ & $\mathbf{0.28}$ & $\mathbf{0.28}$ & $\mathbf{0.28}$ & $0.29$ & $0.35$ & $0.44$ \\
$12{,}500$  & $\mathbf{0.22}$ & $\mathbf{0.22}$ & $\mathbf{0.22}$ & $\mathbf{0.22}$ & $0.23$ & $0.25$ & $0.32$ & $0.42$ \\
$50{,}000$  & $\mathbf{0.21}$ & $\mathbf{0.21}$ & $\mathbf{0.21}$ & $\mathbf{0.21}$ & $0.22$ & $0.24$ & $0.31$ & $0.41$ \\
\bottomrule
\end{tabular}
\end{table}

The effect of choosing a single $\kappa$ for $\mR_{H^2}$ across all $N$ is illustrated by the mid-grid value $\kappa = 10^{-3}$. That column equals $1.21$, $0.69$, $0.28$, $0.23$, and $0.22$ at the five sample sizes. That fixed choice lies $0.31$ nats above the row optimum at $N=100$, where the optimum is the largest value on the grid, $\kappa = 3\cdot10^{-2}$, $0.06$ above it at $N=500$, and within $0.01$ nats of it for $N \geq 2{,}500$, and the best single $\kappa$, $3\cdot10^{-3}$, is within $0.09$ nats of the row optimum over the range. The regularizer $\mR$ at the fixed pair $(C_{\rm samp}, C_{\rm curv}) = (3, 10^{-5})$ gives $0.92, 0.62, 0.31, 0.26, 0.25$ across the same $N$ values (\Cref{tab:stage2-headline}, last column). These values are within $0.04$ nats of the $\mR_{H^2}$ row optimum at every $N$ and below it at $N = 500$, with no per-$N$ grid search.

The two regularizers differ in their prefactors. $\mR_{H^2}$ has the single global prefactor $\kappa\|\widehat\mA_r\|_{\rm op}$, and a single $\kappa$ concedes up to $0.09$ nats to the per-$N$ optimum over the $N$ range of \Cref{tab:dw-kappa-grid}. The Stage-2 regularizer has a term $C_{\rm samp}\widehat\sigma_k^2$ that scales as $1/N$ and an $N$-independent curvature term $C_{\rm curv}\|\widehat\mA_r\|_{\rm op}\, k^2$. The comparison above holds its two constants fixed at $(C_{\rm samp}, C_{\rm curv}) = (3, 10^{-5})$ at every $N$ while $\mR_{H^2}$ receives its per-$N$ best $\kappa$ from the grid, and the proposal still lands within $0.04$ nats of that oracle at every $N$ and below it at $N = 500$. The estimator of the main text selects the pair for each estimate by the validation procedure of \Cref{sec:sm2-selection}, on a grid centered at those two values (\Cref{app:protocol}).

\subsection{Cross-example regularizer comparison}
\label{sec:sm-cross-benchmark}

This subsection presents the full cross-example regularizer comparison underpinning the empirical claim of \Cref{sec:ridge} that the two constants need no per-problem hand-tuning: one fixed pair $(C_{\rm samp}, C_{\rm curv}) = (3, 10^{-5})$ stays within $0.18$ nats of the best alternative on all three examples, and a per-realization held-out choice improves on that fixed pair by more than $0.01$ nats only for $N \le 1000$, by the amounts in the $\mR_{\rm th,CV}$ row of \Cref{tab:cross-bench-results}. We compare that fixed pair against three alternative regularizer designs, each with a fixed and a $c/N$-scaled coefficient, against a per-realization held-out choice of $C_{\rm samp}$ with $C_{\rm curv}$ fixed, and against an unregularized solve, across three structurally distinct example problems.

\paragraph{Example problems} The three examples differ in how well a Hermite expansion represents the copula score. All three share the $d=20$ construction of \S\ref{sec:exp-benchmark}: a latent vector of independent standard Gaussian coordinates in which two of the first four coordinates are replaced by bounded nonlinear functions of the other two, leaving four coupled coordinates and $16$ independent standard Gaussian coordinates; a uniformly random rotation $\mQ \sim \mathrm{Unif}(\mathrm{O}(8))$ mixing the first $8$ coordinates and acting as the identity on the remaining $12$ (orthogonal, so the latent coordinates remain uncorrelated with unit variance, while the four dependent directions are not axis-aligned; the reference $\mV_r^{\mC}$ throughout is the top-$r$ eigenspace of the score covariance $\mC$, evaluated by Monte Carlo at $N_{\rm ref} = 10^6$ as in \S\ref{sec:exp-benchmark}); and the componentwise map $n_j = \mathrm{logit}(\Phi(z_j))$ applied to $\vz = \mQ\vW$.
\begin{itemize}\itemsep0pt
\item \textbf{Example~1} (main paper \S\ref{sec:exp-banana}): the map \eqref{eq:ex1-block} taking $(W_0,W_1)$ to $(W_2,W_3)$, the real and imaginary parts of $(W_0+iW_1)^2$ with bounded $\tanh$-saturated conditional means. Its score is smooth. Reference score-covariance spectrum: top-$4$ eigenvalues $\approx 2.29, 2.28, 0.95, 0.92$, with the remaining $16$ summing to $\approx 0.04$; \Cref{fig:eigvals-all} (left) shows the estimated spectrum at the Stage-1 multi-index set, whose values are smaller because the multi-index set represents only part of the copula score: $\tr(\mC_\Lambda) = \tr(\mC) - \|\scS - \scS_\Lambda\|^2_{L^2(\pi_{\vZ})}$ (\S\ref{sec:theory}).
\item \textbf{Example~2.} Two input--output pairs, $W_0$ determining $W_1$ and $W_2$ determining $W_3$; writing $(W_{\rm in}, W_{\rm out})$ for a generic pair:
\[
W_{\rm out} = c\bigl(\tanh(\gamma (W_{\rm in}^2 - 1)) - m_0\bigr) + \sqrt{1 - a^2}\,\varepsilon, \qquad W_{\rm in}, \varepsilon \sim \mathcal N(0,1),
\]
$\gamma = 0.5$, $a = 0.7$, $m_0 = \EE[\tanh(\gamma(W_{\rm in}^2-1))]$ (centering) and $c = a/\sqrt{\Var}$ (variance-normalizing). The map $W_{\rm in} \mapsto W_{\rm out}$ is even, hence non-injective, and its copula score is harder for a Hermite expansion to represent; the active subspace is still recovered, and it is the density representation that suffers, which is why this example is included. Reference spectrum: top-$4 \approx 2.28, 2.16, 0.95, 0.80$, with the remaining $16$ summing to $\approx 0.09$; the estimated spectrum is in \Cref{fig:eigvals-all} (middle).
\item \textbf{Example~3.} The map taking $(W_0,W_1)$ to $(W_2,W_3)$, the real and imaginary parts of $(W_0+iW_1)^3$:
\[
\begin{gathered}
W_2 = c\,\tanh\!\Bigl(\tfrac{W_0^3 - 3 W_0 W_1^2}{M}\Bigr) + \sqrt{1-a^2}\,\varepsilon_2, \\
W_3 = c\,\tanh\!\Bigl(\tfrac{3 W_0^2 W_1 - W_1^3}{M}\Bigr) + \sqrt{1-a^2}\,\varepsilon_3,
\end{gathered}
\]
with $M = 4$, $a = 0.7$, and $c \approx 1.48$ (variance-normalizing; both output means are $0$ by symmetry). The score is smooth and the easiest of the three to represent, and the three-to-one map $z \mapsto z^3$ gives a posterior with mass on three preimage branches. Reference spectrum: top-$4 \approx 4.42, 4.39, 0.97, 0.94$, with the remaining $16$ summing to $\approx 0.04$; the estimated spectrum is in \Cref{fig:eigvals-all} (right).
\end{itemize}

\paragraph{Comparison regularizers} For each example problem and each $N$ in $\{100, 500,\allowbreak\,1000, 2500,\allowbreak\,5000, 10000,\allowbreak\,25000, 50000,\allowbreak\,100000\}$, with $20$ paired realizations, we measure the held-out log-score of the Stage-2 estimator for each of nine regularizer choices. The reference $\mR_{\rm th,nominal}$ fixes the representative pair $(C_{\rm samp}, C_{\rm curv}) = (3, 10^{-5})$ \eqref{eq:R-combined}. Six alternatives are $\mR_{\rm samp}$- and $\mR_{\rm curv}$-replacements with identity scaling and Sobolev exponents $|\alpha|_1$ and $|\alpha|_1^2$ respectively, each in fixed and $c/N$-scaling variants. $\mR_{\rm th,CV}$ chooses $C_{\rm samp}$ for each realization by held-out log-likelihood on a single split, with $C_{\rm curv}$ fixed; $\mR_{\rm no\,reg}$ is the unregularized solve.

\paragraph{Aggregate results} \Cref{tab:cross-bench-results} reports the mean held-out log-score gap (regularizer minus $\mR_{\rm th,nominal}$, in nats) for each $N$ on each example problem, over $20$ paired realizations.

\begin{table}[h]
\centering\footnotesize
\setlength{\tabcolsep}{4pt}
\caption{The cross-example comparison reports the mean held-out log-score gap versus $\mR_{\rm th,nominal}$, where negative is better than the selected pair, over $20$ paired realizations at each $N$; the reference log-score is that of $\mR_{\rm th,nominal}$.}
\label{tab:cross-bench-results}
\begin{tabular}{l|rrrrrrrrr}
\toprule
$N \to$ & $100$ & $500$ & $1{,}000$ & $2{,}500$ & $5{,}000$ & $10^4$ & $2.5 \cdot 10^4$ & $5 \cdot 10^4$ & $10^5$ \\
\midrule
\multicolumn{10}{l}{\textbf{Example~1} \quad (ref log-score: $2.35$, $1.21$, $0.82$, $0.55$, $0.43$, $0.37$, $0.32$, $0.30$, $0.29$)} \\
$\mR_{\rm th,CV}$         & $-0.14$ & $-0.04$ & $-0.00$ & $-0.00$ & $+0.00$ & $+0.00$ & $-0.00$ & $-0.00$ & $-0.00$ \\
$\mR_{\rm cov,id\,fix}$    & $-0.07$ & $-0.05$ & $+0.02$ & $+0.00$ & $+0.00$ & $+0.00$ & $-0.00$ & $-0.00$ & $+0.00$ \\
$\mR_{\rm sob,k\,fix}$     & $-0.15$ & $-0.04$ & $+0.01$ & $+0.00$ & $+0.00$ & $+0.00$ & $+0.00$ & $-0.00$ & $-0.00$ \\
$\mR_{\rm sob,k^2\,fix}$   & $-0.08$ & $-0.06$ & $+0.00$ & $+0.00$ & $+0.00$ & $+0.00$ & $+0.00$ & $-0.00$ & $-0.00$ \\
$\mR_{\rm no\,reg}$       & $+2.88$ & $+0.18$ & $+0.04$ & $+0.01$ & $+0.00$ & $+0.00$ & $-0.00$ & $-0.00$ & $-0.00$ \\
\midrule
\multicolumn{10}{l}{\textbf{Example~2} \quad (ref log-score: $2.27$, $1.13$, $0.76$, $0.49$, $0.38$, $0.31$, $0.27$, $0.25$, $0.23$)} \\
$\mR_{\rm th,CV}$         & $-0.14$ & $-0.04$ & $-0.01$ & $-0.00$ & $+0.00$ & $-0.00$ & $-0.00$ & $-0.00$ & $-0.00$ \\
$\mR_{\rm cov,id\,fix}$    & $-0.07$ & $-0.05$ & $+0.02$ & $+0.00$ & $-0.00$ & $-0.00$ & $-0.00$ & $-0.00$ & $-0.00$ \\
$\mR_{\rm sob,k\,fix}$     & $-0.14$ & $-0.04$ & $+0.00$ & $+0.00$ & $+0.00$ & $-0.00$ & $-0.00$ & $-0.00$ & $-0.00$ \\
$\mR_{\rm sob,k^2\,fix}$   & $-0.08$ & $-0.06$ & $-0.00$ & $+0.00$ & $+0.00$ & $-0.00$ & $-0.00$ & $-0.00$ & $-0.00$ \\
$\mR_{\rm no\,reg}$       & $+2.76$ & $+0.17$ & $+0.04$ & $+0.01$ & $+0.00$ & $-0.00$ & $-0.00$ & $-0.00$ & $-0.00$ \\
\midrule
\multicolumn{10}{l}{\textbf{Example~3} \quad (ref log-score: $2.39$, $1.37$, $1.11$, $0.77$, $0.61$, $0.54$, $0.50$, $0.46$, $0.45$)} \\
$\mR_{\rm th,CV}$         & $-0.18$ & $-0.12$ & $-0.04$ & $+0.01$ & $-0.00$ & $-0.00$ & $-0.00$ & $-0.00$ & $-0.00$ \\
$\mR_{\rm cov,id\,fix}$    & $-0.08$ & $-0.11$ & $-0.05$ & $-0.00$ & $-0.00$ & $-0.00$ & $-0.00$ & $-0.00$ & $-0.00$ \\
$\mR_{\rm sob,k\,fix}$     & $-0.16$ & $-0.13$ & $-0.04$ & $+0.00$ & $-0.00$ & $-0.00$ & $-0.00$ & $-0.00$ & $-0.00$ \\
$\mR_{\rm sob,k^2\,fix}$   & $-0.09$ & $-0.11$ & $-0.04$ & $+0.00$ & $-0.00$ & $-0.00$ & $-0.00$ & $-0.00$ & $-0.00$ \\
$\mR_{\rm no\,reg}$       & $+3.63$ & $+0.17$ & $+0.04$ & $+0.00$ & $-0.00$ & $-0.00$ & $-0.00$ & $-0.00$ & $-0.00$ \\
\bottomrule
\end{tabular}
\end{table}

\paragraph{Three findings} The comparisons support the cross-family robustness claim of \S\ref{sec:ridge}.

\textit{(i) The chosen pair transfers across all three problems.} One pair of constants $(C_{\rm samp}, C_{\rm curv}) = (3, 10^{-5})$ stays within $0.18$ nats of the best alternative regularizer at every $N$ on all three structurally distinct examples, with no per-problem hand-tuning. Its gap to the per-realization held-out choice is at most $0.14$--$0.18$ nats at $N=100$ (Example~1: $0.14$; Example~2: $0.14$; Example~3: $0.18$) and closes to within $\pm 0.01$ nats for $N \geq 2500$, where the fixed and held-out choices are indistinguishable.

\textit{(ii) Omitting the regularizer costs $2.8$--$3.6$ nats at $N = 100$ and at most $0.01$ nats for $N \ge 2500$.} The gap of $\mR_{\rm no\,reg}$ ranges from $+2.88$ at $N=100$ to $\approx 0$ at $N=10^5$ on Example~1, $+2.76$ to $\approx 0$ on Example~2, and $+3.63$ to $\approx 0$ on Example~3: an inflation that falls to at most $+0.042$ nats by $N = 1000$ and at most $+0.01$ nats for $N \geq 2500$ on all three examples. The tables do not separate the variance and regularization-bias contributions to these gaps.

\textit{(iii) Scaling the coefficient as $c/N$ changes the mean gaps by at most $0.10$ nats.} The $c/N$ variants of the three fixed-coefficient regularizers, omitted from \Cref{tab:cross-bench-results}, give mean gaps within $0.10$ nats of their fixed-coefficient counterparts at every $N$ on all three examples, and within $0.002$ nats for $N \ge 2500$.

\paragraph{Conclusion} The fixed pair $(C_{\rm samp}, C_{\rm curv}) = (3, 10^{-5})$ stays within $0.18$ nats of the best alternative regularizer at every $N$ tested on all three examples, without per-problem tuning. Its distance from the per-realization held-out choice is the $\mR_{\rm th,CV}$ row of \Cref{tab:cross-bench-results}: at most $0.18$ nats at $N = 100$ and at most $0.01$ nats for $N \ge 2500$. Omitting the regularizer costs $2.8$--$3.6$ nats at $N = 100$ and at most $0.01$ nats for $N \ge 2500$.

\subsection{Centering constraint: the in-sample identity and its basis-truncation bias}
\label{sec:sm2-centering}

The Stage-2 estimator \eqref{eq:stage2-closedform} enforces the empirical centering identity $\widehat{\bm B}\bm\theta = 0$ of \eqref{eq:centering-mat}, the empirical form of $\EE_{\pi_{\vU}}[\nabla\log c_r^{U}] = 0$. By \Cref{prop:centering}, imposing it contracts the regularization bias and adds an $N$-independent basis-truncation bias determined by the mean of the Fisher truncation residual, which vanishes when the reduced score is representable in $\Lambda_r$. This subsection measures the resulting net change in held-out log-score on the three example problems by comparing the centered closed-form solve \eqref{eq:stage2-closedform} against the uncentered solve $\widehat{\bm\theta}_{\rm unc} = \widehat{\bm M}^{-1}\widehat{\bm b}_r$ on the same $(\widehat{\bm A}_r + \mR, \widehat{\bm b}_r)$, without the inequality constraints of \Cref{app:solver}, holding subspace, rank, marginals, and regularizer fixed so the paired gap isolates the centering constraint.

\paragraph{In-sample identity} The centered solve holds the in-sample residual $\|\widehat{\bm B}\bm\theta\|$ at machine precision, with a largest value of $1.1\times10^{-16}$ over every realization at every $N$ on all three examples. The uncentered residual does not decrease to zero. Its across-realization mean is $0.074$ at $N = 500$ and $0.028$ at $N = 50{,}000$ on Example~1, $0.041$ and $0.022$ at those sample sizes on Example~3, and lies between $0.058$ and $0.073$ at every $N$ on Example~2. Its size at $N = 50{,}000$ orders the three examples by the difficulty of representing each copula score in a Hermite expansion: $0.058$ on Example~2, whose even map is non-injective, against $0.022$ on Example~3, the easiest of the three to represent.

\paragraph{Net effect} \Cref{tab:sm2-centering} reports the signed centering gap $\Delta_{\rm ctr}$, the held-out log-score of the centered solve minus that of the uncentered one, across $N$ at $(K_2, q_2) = (5, 3)$ over $20$ independent realizations. The centered solve attains the lower mean log-score at every sample size on all three examples. The gap equals $-0.0046$ nats at $N = 500$ on Examples~1 and~3 and lies between $-0.0031$ and $-0.0001$ nats at $N = 50{,}000$.

\begin{table}[h]
\centering\small
\setlength{\tabcolsep}{5pt}
\caption{The centering comparison reports the signed log-score gap $\Delta_{\rm ctr}$ (centered minus uncentered, in nats) for $(K_2, q_2) = (5, 3)$, as the mean $\pm$ one standard deviation over $20$ independent realizations at each $N$; negative favors the centered solve. The final column gives the across-realization mean of the uncentered residual on Example~1, which does not decrease to zero.}
\label{tab:sm2-centering}
\begin{tabular}{r|ccc|c}
\toprule
 & \multicolumn{3}{c|}{$\Delta_{\rm ctr}$ (nats)} & Ex.~1 residual \\
$N$ & Example~1 & Example~2 & Example~3 & uncentered \\
\midrule
$500$      & $-0.0046{\pm}0.0054$ & $-0.0037{\pm}0.0042$ & $-0.0046{\pm}0.0040$ & $0.074$ \\
$2{,}500$  & $-0.0024{\pm}0.0024$ & $-0.0036{\pm}0.0021$ & $-0.0010{\pm}0.0012$ & $0.062$ \\
$12{,}500$ & $-0.0008{\pm}0.0007$ & $-0.0031{\pm}0.0012$ & $-0.0003{\pm}0.0003$ & $0.036$ \\
$50{,}000$ & $-0.0005{\pm}0.0006$ & $-0.0031{\pm}0.0008$ & $-0.0001{\pm}0.0001$ & $0.028$ \\
\bottomrule
\end{tabular}
\end{table}

\paragraph{Dependence on the Stage-2 degree} \Cref{tab:sm2-centering-k2} sweeps the inner degree $K_2 \in \{3, 4, 5, 6\}$ at $N = 50{,}000$ over $15$ independent realizations. At $K_2 = 3$, where $|\Lambda_r| = 30$, the constraint costs $+0.049$ nats on Example~1 and $+0.051$ nats on Example~2, against gaps of magnitude at most $0.0046$ nats in \Cref{tab:sm2-centering}, and $|\Delta_{\rm ctr}| \le 0.006$ nats at every $K_2 \ge 4$ on both. The $K_2 = 5$ column agrees with the $N = 50{,}000$ row of \Cref{tab:sm2-centering} to within $5\times10^{-5}$ nats on all three examples, the two tables reporting the same $(K_2, N) = (5, 50{,}000)$ comparison. Example~3 departs from the pattern at $K_2 = 4$, where $\Delta_{\rm ctr} = +0.0250 \pm 0.0171$ nats against $-0.0000$ at $K_2 = 3$ and $-0.0001$ at $K_2 = 5$; its standard deviation is the largest in the table, and $(K_2, q_2) = (5,3)$ is the setting of \Cref{tab:sm2-centering}.

\begin{table}[h]
\centering\small
\setlength{\tabcolsep}{5pt}
\caption{The centering comparison across the Stage-2 inner degree $K_2$ at $N = 50{,}000$ reports the signed log-score gap $\Delta_{\rm ctr}$ (nats) as the mean $\pm$ one standard deviation over $15$ independent realizations, together with the basis size $|\Lambda_r|$. The cost is largest at the coarsest degree $K_2 = 3$ on Examples~1 and~2, and at $K_2 = 4$ on Example~3.}
\label{tab:sm2-centering-k2}
\begin{tabular}{rr|ccc}
\toprule
$K_2$ & $|\Lambda_r|$ & Example~1 & Example~2 & Example~3 \\
\midrule
$3$ & $30$  & $+0.0490{\pm}0.0165$ & $+0.0505{\pm}0.0126$ & $-0.0000{\pm}0.0001$ \\
$4$ & $64$  & $-0.0005{\pm}0.0007$ & $-0.0056{\pm}0.0013$ & $+0.0250{\pm}0.0171$ \\
$5$ & $116$ & $-0.0004{\pm}0.0006$ & $-0.0031{\pm}0.0007$ & $-0.0001{\pm}0.0001$ \\
$6$ & $190$ & $-0.0001{\pm}0.0004$ & $-0.0005{\pm}0.0005$ & $-0.0002{\pm}0.0005$ \\
\bottomrule
\end{tabular}
\end{table}

\paragraph{Interpretation} The two tables show two regimes. At $K_2 = 3$ on Examples~1 and~2, the coarsest multi-index set, centering raises the log-score by $0.05$ nats. At $(K_2, q_2) = (5, 3)$, the setting of \Cref{tab:sm2-centering}, the absolute paired log-score difference is at most $0.005$ nats and the centered solve has the lower mean log-score at every sample size on all three examples; Example~3 at $K_2 = 4$ is the exception noted above. The tables compare held-out log-scores; they do not measure the two terms of \eqref{eq:centering-bias} separately.

\section{Greedy choice of the Stage-1 polynomial space}
\label{sm:adaptivity}

The experiments select the Stage-1 polynomial space from a prescribed list by cross-validation. This section describes an alternative for exploring the degree and interaction order when that list is not known in advance, following the general idea of adaptive polynomial enrichment \cite{BlatmanSudret2011} and using the quadratic score-matching objective to reuse earlier computations. It was not used in the reported experiments, and its statistical performance and runtime have not been assessed.

\subsection{A fixed quadratic during the search}
\label{sec:adapt-fixed}

We adapt $(K_1, q_1)$ in Stage~1, holding the sample and the coordinate transform fixed throughout the search; the transform is constructed from the estimation sample alone and applied unchanged to the validation sample. Choose a scalar $\eta > 0$ once, for example from the initial estimation-sample Gram matrix, and use
\[
  R_\Lambda = \eta\,\diag(|\alpha|_1^2)_{\alpha\in\Lambda}, \qquad H_\Lambda = \widehat A_\Lambda + R_\Lambda,
\]
with $\widehat A_\Lambda$ the Stage-1 Gram matrix \eqref{eq:A-b-formulas} on $\Lambda$. This differs from recomputing the dictionary-dependent scale $\kappa\|\widehat A_\Lambda\|_{\rm op}$ of \eqref{eq:tikhonov} after every enlargement: with the scale fixed, an enlarged quadratic restricts to the same quadratic on the old coefficients. Since $\widehat A_\Lambda \succeq 0$ and every retained degree is at least two, $H_\Lambda \succ 0$. Write the estimation objective and its minimizer as
\[
  Q_\Lambda(\theta) = \tfrac12\theta^\top H_\Lambda\theta - \widehat b_\Lambda^\top\theta, \qquad \widehat\theta_\Lambda = H_\Lambda^{-1}\widehat b_\Lambda,
\]
and the unpenalized quadratic on the validation sample as $\widehat J^{\rm val}_\Lambda(\theta) = \tfrac12\theta^\top\widehat A^{\rm val}_\Lambda\theta - (\widehat b^{\rm val}_\Lambda)^\top\theta$. In fixed coordinates satisfying the score-matching integration hypotheses this criterion estimates the Fisher divergence up to a candidate-independent constant; with an estimated transform we use it as an empirical tuning criterion.

\subsection{Exact reuse when a block is added}

From $\Lambda_t = \Lambda_{K_t,q_t}$ the candidate blocks are the complete degree and interaction increments
\[
  G_K = \Lambda_{K_t+1,q_t}\setminus\Lambda_t, \qquad G_q = \Lambda_{K_t,q_t+1}\setminus\Lambda_t;
\]
empty blocks are omitted, and candidate sizes are checked against the budget before their features are assembled.

\begin{proposition}[Reusable block solve]
\label{prop:sm2-block}
Let $H_{tt} = L_tL_t^\top$ be the cached Cholesky factorization on $\Lambda_t$ \cite{GolubVanLoan2013}, with $\widehat\theta_t = H_{tt}^{-1}\widehat b_t$. For a candidate block $G$, define
\[
  B_G = L_t^{-1}\widehat A_{tG}, \qquad S_G = H_{GG} - B_G^\top B_G, \qquad r_G = \widehat b_G - \widehat A_{Gt}\widehat\theta_t.
\]
Then $S_G \succ 0$; with $S_G = L_GL_G^\top$ its Cholesky factorization, the enlarged optimum and Cholesky factor are
\begin{equation}
  \nu_G = S_G^{-1}r_G, \qquad
  \widehat\theta_{t\cup G} = \begin{pmatrix}\widehat\theta_t - L_t^{-\top}B_G\nu_G\\ \nu_G\end{pmatrix}, \qquad
  L_{t\cup G} = \begin{pmatrix}L_t & 0\\ B_G^\top & L_G\end{pmatrix},
  \label{eq:sm2-deltatheta}
\end{equation}
and the decrease in the penalized estimation objective is
\begin{equation}
  Q_t(\widehat\theta_t) - Q_{t\cup G}(\widehat\theta_{t\cup G}) = \tfrac12 r_G^\top S_G^{-1}r_G \ge 0.
  \label{eq:sm2-deltaj}
\end{equation}
\end{proposition}

\begin{proof}
The enlarged Hessian is positive definite, so its Schur complement $S_G$ is positive definite. Expanding the enlarged objective at $(\widehat\theta_t + u, \nu)$ and using $H_{tt}\widehat\theta_t = \widehat b_t$ gives
\[
  Q_t(\widehat\theta_t) + \tfrac12u^\top H_{tt}u + u^\top\widehat A_{tG}\nu + \tfrac12\nu^\top H_{GG}\nu - r_G^\top\nu.
\]
Minimizing over $u$ leaves $Q_t(\widehat\theta_t) + \tfrac12\nu^\top S_G\nu - r_G^\top\nu$, which gives the coefficient and objective formulas. Multiplying $L_{t\cup G}$ by its transpose gives the enlarged Hessian.
\end{proof}

An accepted candidate keeps the factorization computed for its trial, so acceptance needs no second full solve. The gain \eqref{eq:sm2-deltaj} concerns the estimation objective; acceptance uses the validation criterion.

\begin{algorithm}[t]
\caption{Proposed greedy Stage-1 dictionary search.}
\label{alg:adaptive-cas}
\small
\begin{algorithmic}[1]
\Require fixed estimation and validation coordinates; initial $(K_0,q_0)$; target rank $r$; fixed $\eta>0$; tolerance $\tau_{\rm val}\ge0$; degree, interaction, dictionary-size, memory and work limits.
\State Fit on $\Lambda_0$; retain its Cholesky factor and estimation and validation quadratic blocks, and evaluate its validation criterion.
\While{at least one nonempty neighboring expansion is affordable}
  \State Form the nonempty neighboring blocks $G_K, G_q$ within the prescribed bounds, recording any move excluded by a bound or budget.
  \State For each affordable candidate, assemble only the missing quadratic blocks, extend the factorization by \Cref{prop:sm2-block}, and evaluate $\widehat J^{\rm val}_{t\cup G}(\widehat\theta_{t\cup G})$.
  \State Select the candidate with the smallest validation criterion, breaking ties in favor of fewer added terms.
  \If{its improvement over the current criterion exceeds $\tau_{\rm val}$}
    \State Accept its coefficients and cached factorization; update $(K_t,q_t)$.
  \Else
    \State Record a local stop if every nonempty neighbor was evaluated, and a budget-limited stop otherwise; \textbf{break}.
  \EndIf
\EndWhile
\State If the loop ended because no expansion was affordable, record the limiting budget or search bound.
\State Form the final score matrix and its leading eigenspace, and run Stage~2 at that subspace as in the main text.
\end{algorithmic}
\end{algorithm}

No Stage-2 fit, density normalizer, or eigenspace computation is needed to compare Stage-1 candidates. Changing $\eta$, the transform, or the estimation sample defines a new quadratic and requires a new factorization.

\subsection{Cost}
\label{app:adaptive-cost}

With $m = |\Lambda_t|$ and $g = |G|$, the dense factorization extension costs $O(m^2g + mg^2 + g^3)$, the $m^2g$ term including the triangular solve for $B_G$. Since $m^2g + mg^2 + g^3 \le (m+g)^3 - m^3$, these costs sum along the accepted path, including the initial factorization, to $O(p_T^3)$ for the final dictionary size $p_T$, the order of one dense factorization at the final size, with no assumption on the growth of the dictionary sizes. This bound covers the accepted factorizations only. Rejected candidates, feature assembly (of order $(N_{\rm est}+N_{\rm val})\,d\,(mg+g^2)$ for the new Gram blocks), and validation add work, and the cached matrices and factors need storage quadratic in the dictionary size. At $d = 20$, $|\Lambda_{4,2}| = 1{,}200$, $|\Lambda_{5,2}| = 1{,}980$ and $|\Lambda_{4,3}| = 5{,}760$, so from $(4, 2)$ the degree increment adds $780$ terms and the interaction increment $4{,}560$; the search therefore carries explicit limits on candidate size, memory and cumulative work, including rejected candidates.

\subsection{Meaning of the stopping decision}
\label{sm:adapt-variants}

If every nonempty neighboring expansion within the declared search domain was evaluated and none improves validation by more than $\tau_{\rm val}$, the current dictionary is locally adequate for that criterion within that domain; a stop caused by a computational limit is reported separately. The tolerance is a resolution parameter. It is neither a confidence level nor a bound in KL units, and the greedy path can miss a better dictionary that requires simultaneous increases. Repeated use of the validation sample is tuning, so predictive performance should be assessed on unused data.

The fitted tail $\widehat E_r$ and the Stein cross-moment residual \cite{Stein1981} may be inspected as diagnostics, but neither certifies accuracy. The latter checks the identity $\EE[\scS(\vZ)\vZ^\top] = \EE[\vZ\vZ^\top] - \mI$, which follows from integration by parts under the population regularity assumptions and tests only those cross-moments. At the population level, \Cref{thm:kl-trunc} at the leading projector of $\mC_\Lambda = \EE[\scS_\Lambda\scS_\Lambda^\top]$ and the triangle inequality give
\[
  \DKL\bigl(\pi_{\vn}\,\|\,\pi_{\vn}(\,\cdot\,;\mV_{r,\Lambda})\bigr) \le \tfrac12\Bigl[\sqrt{E_r(\mC_\Lambda)} + \|\scS - \scS_\Lambda\|_{L^2(\pi_{\vZ})}\Bigr]^2,
\]
in which the score error is essential and the exact reduced model on the left is not the fitted Stage-2 density. Stage~2 keeps its own dictionary selection, regularization and constraints; extending the search to them, or choosing the rank jointly, is not addressed here.

\section{Marginals of the reduced model}
\label{sm:marginals}

The rank transform separates marginal estimation from dependence estimation, and \S\ref{sec:reduced-form} composes the two by multiplying the marginals by a factor that depends on $\mV_r^\top\scZ(\vx)$. The composition does not leave the marginals unchanged. This section gives the marginals of a member of $\mathcal M(\mV_r, \scZ, \{\pi_i\})$, identifies the reductions under which they equal $\{\pi_i\}$ for every member, and bounds their deviation at the KL divergence projection.

Throughout, $\va_i^\top$ denotes the $i$-th row of $\mV_r$, so $\va_i = \mV_r^\top\ve_i \in \RR^r$ for $i = 1, \ldots, d$, and $\|\va_i\|^2 = (\mV_r\mV_r^\top)_{ii}$ measures how much observation coordinate $i$ participates in the active subspace.

\begin{proposition}[Marginals of the reduced model]
\label{prop:sm-marginals}
Let the member $\pi'$ of $\mathcal M(\mV_r, \scZ, \{\pi_i\})$ have reduced factor $c' \in L^2(\gauss_r)$, and let $\widehat c_\alpha := \EE_{c'\gauss_r}[H_\alpha]$ be its Hermite coefficients. Then:

\emph{(i)} the $i$-th marginal of $\pi'$ equals $\pi_i \cdot (g_i\circ\scZ_i)$, where, with the series converging in $L^2(\gauss_1)$,
\begin{equation}
  g_i(z) \;=\; 1 + \sum_{k\ge1} d_{i,k}\,h_k(z),
  \qquad
  d_{i,k} \;:=\; \sum_{|\alpha|_1 = k}\sqrt{k!/\alpha!}\;\va_i^{\alpha}\,\widehat c_\alpha,
  \label{eq:sm-marginal-drift}
\end{equation}
and $d_{i,k} = \|\va_i\|^k\,\EE_{c'\gauss_r}\bigl[h_k(\va_i^\top\vU/\|\va_i\|)\bigr]$ whenever $\va_i \neq 0$;

\emph{(ii)} $g_i \equiv 1$ for every reduced factor $c'$ if and only if $\va_i = 0$. Since $\sum_{i=1}^d\|\va_i\|^2 = r$, at most $d-r$ rows of $\mV_r$ vanish, with equality exactly when $\mathrm{span}(\mV_r)$ is spanned by $r$ canonical basis vectors. Under that choice, in which the reduction selects $r$ observation components, the $d-r$ unselected marginals equal $\pi_i$ for every member of the family, and the $r$ selected ones are not guaranteed to;

\emph{(iii)} at the KL divergence projection \eqref{eq:pi-star-form}, $\DKL(\pi_i\,\|\,\pi_i\,(g_i\circ\scZ_i)) \le \tfrac12\tr((\mI - \mV_r\mV_r^\top)\mC)$ for every $i$, which equals $\tfrac12\sum_{j>r}\lambda_j(\mC)$ when $\mV_r$ spans the leading eigenspace, and every marginal equals $\pi_i$ when that trace vanishes.
\end{proposition}

\begin{proof}
\emph{(i)} Integrating $\pi'$ over $\vx_{-i}$ and changing variables by $\pi_j(x_j)\,dx_j = \gauss_1(z_j)\,dz_j$ in each $j \neq i$ gives $\pi_i(x_i)\,\EE_{\gauss_d}[c'(\vU)\mid Z_i = z_i]$ with $z_i = \scZ(\vx)_i$. Under $\gauss_d$ the pair $(\vU, Z_i)$ is jointly Gaussian with $\vU\sim\gauss_r$, $Z_i\sim\gauss_1$ and $\Cov(\vU, Z_i) = \mV_r^\top\ve_i = \va_i$, so $\vU \mid Z_i = z \sim \mathcal N(\va_i z,\ \mI_r - \va_i\va_i^\top)$. Applying that conditional law to the Hermite generating function $\sum_\alpha t^\alpha H_\alpha(\vu)/\sqrt{\alpha!} = \exp(t\cdot\vu - \|t\|^2/2)$ gives $\exp(sz - s^2/2)$ with $s := t\cdot\va_i$, whose expansion is $\sum_k s^k h_k(z)/\sqrt{k!}$. Substituting $s^k = \sum_{|\alpha|_1=k}(k!/\alpha!)\,t^\alpha\va_i^\alpha$ and matching coefficients of $t^\alpha$,
\begin{equation}
  \EE_{\gauss_d}\bigl[H_\alpha(\vU)\mid Z_i = z\bigr] \;=\; \sqrt{|\alpha|_1!/\alpha!}\;\va_i^{\alpha}\,h_{|\alpha|_1}(z).
  \label{eq:sm-cond-hermite}
\end{equation}
Multiplying by $\widehat c_\alpha$ and summing gives \eqref{eq:sm-marginal-drift}, with the constant term $\widehat c_0 = 1$ because $c'$ is a probability density with respect to $\gauss_r$; the Hermite series of $c'$ converges in $L^2(\gauss_r)$ and conditional expectation is an $L^2$ contraction, so the series in \eqref{eq:sm-marginal-drift} converges in $L^2(\gauss_1)$. The one-dimensional form follows from $\sum_{|\alpha|_1=k}\sqrt{k!/\alpha!}\,\va^\alpha H_\alpha(\vu) = \|\va\|^k h_k(\va^\top\vu/\|\va\|)$.

\emph{(ii)} If $\va_i = 0$ then $\va_i^\alpha = 0$ for every $|\alpha|_1 \ge 1$, so every $d_{i,k}$ vanishes. If $\va_i \neq 0$, set $\vv := \va_i/\|\va_i\|$ and take $c'$ to be the density of $\mathcal N(0, \mI_r + (\sigma^2-1)\vv\vv^\top)$ relative to $\gauss_r$ for some $\sigma^2 \in (0,2)$, $\sigma^2 \neq 1$, which keeps $c'$ in $L^2(\gauss_r)$ and its log-density in the second-order modes retained in $\Lambda_r$; then $\va_i^\top\vU/\|\va_i\| \sim \mathcal N(0,\sigma^2)$ under $c'\gauss_r$ and $d_{i,2} = \|\va_i\|^2(\sigma^2-1)/\sqrt2 \neq 0$. For the counting statement, $\|\va_i\|^2 \in [0,1]$ as a diagonal entry of an orthogonal projector and $\sum_i\|\va_i\|^2 = \tr(\mV_r\mV_r^\top) = r$, so if $m$ rows vanish the remaining $d-m$ rows have norms at most one summing in square to $r$, giving $m \le d-r$; equality forces $\|\va_i\| \in \{0,1\}$ for every $i$, so the $r$ nonzero rows are unit vectors and the submatrix of $\mV_r$ on those rows is orthogonal, whence $\mathrm{span}(\mV_r) = \mathrm{span}\{\ve_i : \|\va_i\| = 1\}$. At a selected coordinate $\va_i$ is a unit vector and $g_i$ is the density of $\va_i^\top\vU$ under $c'\gauss_r$ relative to $\gauss_1$, which equals one exactly when that marginal is standard Gaussian.

\emph{(iii)} The coordinate map $\vx \mapsto x_i$ is deterministic, so the data-processing inequality for the KL divergence bounds the divergence of the $i$-th marginals by the divergence of the joint laws, which \Cref{thm:kl-trunc} bounds by $\tfrac12\tr((\mI - \mV_r\mV_r^\top)\mC)$, equal to $\tfrac12\sum_{j>r}\lambda_j(\mC)$ at the leading eigenspace (\Cref{cor:optimal-subspace}). When that trace vanishes the bound is zero, so each marginal equals $\pi_i$.
\end{proof}

\paragraph{Consequences for the estimator} The reductions of \S\ref{sec:experiments} use a $\mV_r$ with no vanishing row, so no marginal of $\widehat\pi_{\vn}$ is guaranteed to equal $\widehat\pi_i$. Part~(iii) bounds the deviation at the KL divergence projection; the marginals of the estimate carry in addition the estimation errors of the reduced factor, the transform, and the marginals.

Restoring the marginals exactly means imposing $d_{i,k} = 0$ for every $i$ and $k$. Those conditions are linear in the Hermite coefficients $\widehat c_\alpha$ of the reduced factor, by \eqref{eq:sm-marginal-drift}, and Stage 2 parameterizes $\log c'$, in which the same conditions are nonlinear. They therefore cannot be appended to the linear system \eqref{eq:stage2-closedform} as the centering constraint \eqref{eq:centering-mat} is. A parameterization of $c'$ itself, for instance $c' = 1 + \sum_{\alpha\in\Lambda_r}\beta_\alpha H_\alpha$, makes those conditions linear and the normalization exact, since $\EE_{\gauss_r}[H_\alpha] = 0$ for every $\alpha \neq 0$. Under that parameterization the score-matching objective \eqref{eq:JW-pop} is no longer quadratic in the coefficients, being defined through the log-density, and positivity of $c'$, automatic for an exponential form, becomes a constraint of the kind the scheme of \Cref{app:solver} imposes. We leave that estimator to future work.

\clearpage
\def\siamprelabel{}

\begin{thebibliography}{10}

\bibitem{BakryGentilLedoux2014}
{\sc D.~Bakry, I.~Gentil, and M.~Ledoux}, {\em Analysis and Geometry of Markov Diffusion Operators}, Springer, Cham, 2014.

\bibitem{BaptistaBrennanMarzouk2025}
{\sc R.~Baptista, M.~C. Brennan, and Y.~Marzouk}, {\em Dimension reduction via score ratio matching}, Trans. Mach. Learn. Res., (2025).

\bibitem{BaptistaMarzoukZahm2023representation}
{\sc R.~Baptista, Y.~Marzouk, and O.~Zahm}, {\em On the representation and learning of monotone triangular transport maps}, Found. Comput. Math., 24 (2024), pp.~2063--2108.

\bibitem{Bhatia1997}
{\sc R.~Bhatia}, {\em Matrix Analysis}, Grad. Texts in Math. 169, Springer, New York, 1997.

\bibitem{BigoniMarzoukPrieurZahm2022}
{\sc D.~Bigoni, Y.~Marzouk, C.~Prieur, and O.~Zahm}, {\em Nonlinear dimension reduction for surrogate modeling using gradient information}, Inf. Inference, 11 (2022), pp.~1597--1639.

\bibitem{BjorckGolub1973}
{\sc {\AA}.~Bj{\"o}rck and G.~H. Golub}, {\em Numerical methods for computing angles between linear subspaces}, Math. Comp., 27 (1973), pp.~579--594.

\bibitem{BlankenshipFalk1976}
{\sc J.~W. Blankenship and J.~E. Falk}, {\em Infinitely constrained optimization problems}, J. Optim. Theory Appl., 19 (1976), pp.~261--281.

\bibitem{BlatmanSudret2011}
{\sc G.~Blatman and B.~Sudret}, {\em Adaptive sparse polynomial chaos expansion based on least angle regression}, J. Comput. Phys., 230 (2011), pp.~2345--2367.

\bibitem{Blekherman2006}
{\sc G.~Blekherman}, {\em There are significantly more nonnegative polynomials than sums of squares}, Israel J. Math., 153 (2006), pp.~355--380.

\bibitem{BormannBauerGeer2011}
{\sc N.~Bormann, P.~Bauer, and A.~J. Geer}, {\em Estimates of observation-error characteristics in clear and cloudy regions for microwave imager radiances from numerical weather prediction}, Q. J. R. Meteorol. Soc., 137 (2011), pp.~1946--1955.

\bibitem{BrennanBigoniZahmSpantiniMarzouk2020}
{\sc M.~C. Brennan, D.~Bigoni, O.~Zahm, A.~Spantini, and Y.~Marzouk}, {\em Greedy inference with structure-exploiting lazy maps}, in Advances in Neural Information Processing Systems 33, 2020, pp.~8330--8342.

\bibitem{CAS_partII}
{\sc J.~Chen and P.~J. van Leeuwen}, {\em Copula Active Subspaces II: Error decomposition, a posteriori estimation, and sharpness of the bounds}, submitted to SIAM/ASA J. Uncertain. Quantif., 2026.

\bibitem{ChenArnaudBaptistaZahm2024}
{\sc Q.~Chen, \'E.~Arnaud, R.~Baptista, and O.~Zahm}, {\em Coupled input-output dimension reduction: Application to goal-oriented Bayesian experimental design and global sensitivity analysis}, arXiv:2406.13425, 2024.

\bibitem{Constantine2015book}
{\sc P.~G. Constantine}, {\em Active Subspaces: Emerging Ideas for Dimension Reduction in Parameter Studies}, SIAM, Philadelphia, 2015.

\bibitem{ConstantineDowWang2014}
{\sc P.~G. Constantine, E.~Dow, and Q.~Wang}, {\em Active subspace methods in theory and practice: Applications to kriging surfaces}, SIAM J. Sci. Comput., 36 (2014), pp.~A1500--A1524.

\bibitem{CotterRobertsStuartWhite2013}
{\sc S.~L. Cotter, G.~O. Roberts, A.~M. Stuart, and D.~White}, {\em MCMC methods for functions: Modifying old algorithms to make them faster}, Statist. Sci., 28 (2013), pp.~424--446.

\bibitem{CuiMartinMarzoukSolonenSpantini2014}
{\sc T.~Cui, J.~Martin, Y.~Marzouk, A.~Solonen, and A.~Spantini}, {\em Likelihood-informed dimension reduction for nonlinear inverse problems}, Inverse Problems, 30 (2014), 114015.

\bibitem{CuiTongZahm2022}
{\sc T.~Cui, X.~T. Tong, and O.~Zahm}, {\em Prior normalization for certified likelihood-informed subspace detection of Bayesian inverse problems}, Inverse Problems, 38 (2022), 124002.

\bibitem{DesroziersBerreChapnikPoli2005}
{\sc G.~Desroziers, L.~Berre, B.~Chapnik, and P.~Poli}, {\em Diagnosis of observation, background and analysis-error statistics in observation space}, Q. J. R. Meteorol. Soc., 131 (2005), pp.~3385--3396.

\bibitem{DreanoTandeoPulido2017}
{\sc D.~Dreano, P.~Tandeo, M.~Pulido, B.~Ait-El-Fquih, T.~Chonavel, and I.~Hoteit}, {\em Estimating model-error covariances in nonlinear state-space models using Kalman smoothing and the expectation--maximisation algorithm}, Q. J. R. Meteorol. Soc., 143 (2017), pp.~1877--1885.

\bibitem{EckartYoung1936}
{\sc C.~Eckart and G.~Young}, {\em The approximation of one matrix by another of lower rank}, Psychometrika, 1 (1936), pp.~211--218.

\bibitem{EnglHankeNeubauer1996}
{\sc H.~W. Engl, M.~Hanke, and A.~Neubauer}, {\em Regularization of Inverse Problems}, Math. Appl. 375, Kluwer Academic Publishers, Dordrecht, 1996.

\bibitem{FowlerVanLeeuwen2013}
{\sc A.~M. Fowler and P.~J. van Leeuwen}, {\em Observation impact in data assimilation: The effect of non-Gaussian observation error}, Tellus A, 65 (2013), 20035.

\bibitem{FritschCarlson1980}
{\sc F.~N. Fritsch and R.~E. Carlson}, {\em Monotone piecewise cubic interpolation}, SIAM J. Numer. Anal., 17 (1980), pp.~238--246.

\bibitem{GolubVanLoan2013}
{\sc G.~H. Golub and C.~F. Van Loan}, {\em Matrix Computations}, 4th ed., Johns Hopkins University Press, Baltimore, MD, 2013.

\bibitem{GolubWelsch1969}
{\sc G.~H. Golub and J.~H. Welsch}, {\em Calculation of {G}auss quadrature rules}, Math. Comp., 23 (1969), pp.~221--230.

\bibitem{Gross1975}
{\sc L.~Gross}, {\em Logarithmic {S}obolev inequalities}, Amer. J. Math., 97 (1975), pp.~1061--1083.

\bibitem{HalkoMartinssonTropp2011}
{\sc N.~Halko, P.-G. Martinsson, and J.~A. Tropp}, {\em Finding structure with randomness: Probabilistic algorithms for constructing approximate matrix decompositions}, SIAM Rev., 53 (2011), pp.~217--288.

\bibitem{Hall1987}
{\sc P.~Hall}, {\em On {K}ullback--{L}eibler loss and density estimation}, Ann. Statist., 15 (1987), pp.~1491--1519.

\bibitem{HastieTibshiraniFriedman2009}
{\sc T.~Hastie, R.~Tibshirani, and J.~Friedman}, {\em The Elements of Statistical Learning}, 2nd ed., Springer, New York, 2009.

\bibitem{HettichKortanek1993}
{\sc R.~Hettich and K.~O. Kortanek}, {\em Semi-infinite programming: Theory, methods, and applications}, SIAM Rev., 35 (1993), pp.~380--429.

\bibitem{Hilbert1888}
{\sc D.~Hilbert}, {\em \"{U}ber die {D}arstellung definiter {F}ormen als {S}umme von {F}ormenquadraten}, Math. Ann., 32 (1888), pp.~342--350.

\bibitem{HuVanLeeuwenGeer2024}
{\sc C.-C. Hu, P.~J. van Leeuwen, and A.~J. Geer}, {\em A non-parametric way to estimate observation errors based on ensemble innovations}, Q. J. R. Meteorol. Soc., 150 (2024), \url{https://doi.org/10.1002/qj.4710}.

\bibitem{Hyvarinen2005}
{\sc A.~Hyv\"arinen}, {\em Estimation of non-normalized statistical models by score matching}, J. Mach. Learn. Res., 6 (2005), pp.~695--709.

\bibitem{Hyvarinen2007}
{\sc A.~Hyv\"arinen}, {\em Some extensions of score matching}, Comput. Statist. Data Anal., 51 (2007), pp.~2499--2512.

\bibitem{JanjicBormannBocquet2018survey}
{\sc T.~Janji\'c, N.~Bormann, M.~Bocquet, J.~A. Carton, S.~E. Cohn, S.~L. Dance, S.~N. Losa, N.~K. Nichols, R.~Potthast, J.~A. Waller, and P.~Weston}, {\em On the representation error in data assimilation}, Q. J. R. Meteorol. Soc., 144 (2018), pp.~1257--1278.

\bibitem{Janson1997gaussian}
{\sc S.~Janson}, {\em Gaussian Hilbert Spaces}, Cambridge Tracts in Math. 129, Cambridge University Press, Cambridge, 1997.

\bibitem{KaipioSomersaloBook}
{\sc J.~Kaipio and E.~Somersalo}, {\em Statistical and Computational Inverse Problems}, Appl. Math. Sci. 160, Springer, New York, 2005.

\bibitem{LiuLaffertyWasserman2009nonparanormal}
{\sc H.~Liu, J.~Lafferty, and L.~Wasserman}, {\em The nonparanormal: Semiparametric estimation of high dimensional undirected graphs}, J. Mach. Learn. Res., 10 (2009), pp.~2295--2328.

\bibitem{MarzoukMoselhyParnoSpantini2016}
{\sc Y.~Marzouk, T.~Moselhy, M.~Parno, and A.~Spantini}, {\em Sampling via measure transport: An introduction}, in Handbook of Uncertainty Quantification, R.~Ghanem, D.~Higdon, and H.~Owhadi, eds., Springer, Cham, 2016, pp.~1--41.

\bibitem{Mezzadri2007}
{\sc F.~Mezzadri}, {\em How to generate random matrices from the classical compact groups}, Notices Amer. Math. Soc., 54 (2007), pp.~592--604.

\bibitem{MorrisonBaptistaMarzouk2017}
{\sc R.~Morrison, R.~Baptista, and Y.~Marzouk}, {\em Beyond normality: Learning sparse probabilistic graphical models in the non-{G}aussian setting}, in Advances in Neural Information Processing Systems 30, 2017, pp.~2359--2369.

\bibitem{Neal2011HMC}
{\sc R.~M. Neal}, {\em {MCMC} using {H}amiltonian dynamics}, in Handbook of Markov Chain Monte Carlo, S.~Brooks, A.~Gelman, G.~L. Jones, and X.-L. Meng, eds., Chapman \& Hall/CRC, Boca Raton, FL, 2011, pp.~113--162.

\bibitem{Nelsen2006book}
{\sc R.~B. Nelsen}, {\em An Introduction to Copulas}, 2nd ed., Springer, New York, 2006.

\bibitem{NocedalWright2006}
{\sc J.~Nocedal and S.~J. Wright}, {\em Numerical Optimization}, 2nd ed., Springer, New York, 2006.

\bibitem{Owen1995}
{\sc A.~B. Owen}, {\em Randomly permuted $(t,m,s)$-nets and $(t,s)$-sequences}, in Monte Carlo and Quasi-Monte Carlo Methods in Scientific Computing, H.~Niederreiter and P.~J.-S. Shiue, eds., Lecture Notes in Statist. 106, Springer, New York, 1995, pp.~299--317.

\bibitem{PapamakariosEtAl2021}
{\sc G.~Papamakarios, E.~Nalisnick, D.~J. Rezende, S.~Mohamed, and B.~Lakshminarayanan}, {\em Normalizing flows for probabilistic modeling and inference}, J. Mach. Learn. Res., 22 (2021), pp.~1--64.

\bibitem{RadhakrishnanBeagleholePanditBelkin2024}
{\sc A.~Radhakrishnan, D.~Beaglehole, P.~Pandit, and M.~Belkin}, {\em Mechanism for feature learning in neural networks and backpropagation-free machine learning models}, Science, 383 (2024), pp.~1461--1467.

\bibitem{RobertsTweedie1996}
{\sc G.~O. Roberts and R.~L. Tweedie}, {\em Exponential convergence of {L}angevin distributions and their discrete approximations}, Bernoulli, 2 (1996), pp.~341--363.

\bibitem{Silverman1986}
{\sc B.~W. Silverman}, {\em Density Estimation for Statistics and Data Analysis}, Chapman \& Hall, London, 1986.

\bibitem{Sklar1959}
{\sc A.~Sklar}, {\em Fonctions de r\'epartition \`a $n$ dimensions et leurs marges}, Publ. Inst. Statist. Univ. Paris, 8 (1959), pp.~229--231.

\bibitem{Sobol1967}
{\sc I.~M. Sobol'}, {\em On the distribution of points in a cube and the approximate evaluation of integrals}, USSR Comput. Math. Math. Phys., 7 (1967), pp.~86--112.

\bibitem{Sriperumbudur2017}
{\sc B.~Sriperumbudur, K.~Fukumizu, A.~Gretton, A.~Hyv\"arinen, and R.~Kumar}, {\em Density estimation in infinite dimensional exponential families}, J. Mach. Learn. Res., 18 (2017), pp.~1--59.

\bibitem{Stein1981}
{\sc C.~M. Stein}, {\em Estimation of the mean of a multivariate normal distribution}, Ann. Statist., 9 (1981), pp.~1135--1151.

\bibitem{Stone1990}
{\sc C.~J. Stone}, {\em Large-sample inference for log-spline models}, Ann. Statist., 18 (1990), pp.~717--741.

\bibitem{Szego1939}
{\sc G.~Szeg\H{o}}, {\em Orthogonal Polynomials}, Amer. Math. Soc. Colloq. Publ. 23, American Mathematical Society, Providence, RI, 1939.

\bibitem{TandeoAilliotBocquet2020}
{\sc P.~Tandeo, P.~Ailliot, M.~Bocquet, A.~Carrassi, T.~Miyoshi, M.~Pulido, and Y.~Zhen}, {\em A review of innovation-based methods to jointly estimate model and observation error covariance matrices in ensemble data assimilation}, Mon. Weather Rev., 148 (2020), pp.~3973--3994.

\bibitem{VanDerVaart1998}
{\sc A.~W. van der Vaart}, {\em Asymptotic Statistics}, Cambridge University Press, Cambridge, 1998.

\bibitem{VuLei2013}
{\sc V.~Q. Vu and J.~Lei}, {\em Minimax sparse principal subspace estimation in high dimensions}, Ann. Statist., 41 (2013), pp.~2905--2947.

\bibitem{ZahmConstantinePrieurMarzouk2020}
{\sc O.~Zahm, P.~G. Constantine, C.~Prieur, and Y.~Marzouk}, {\em Gradient-based dimension reduction of multivariate vector-valued functions}, SIAM J. Sci. Comput., 42 (2020), pp.~A534--A558.

\bibitem{ZahmCuiLawSpantiniMarzouk2022}
{\sc O.~Zahm, T.~Cui, K.~Law, A.~Spantini, and Y.~Marzouk}, {\em Certified dimension reduction in nonlinear {B}ayesian inverse problems}, Math. Comp., 91 (2022), pp.~1789--1835.

\end{thebibliography}
\end{document}